\documentclass[11pt]{article}
\usepackage{amsmath}
\usepackage{amsthm}
\usepackage{amssymb}
\usepackage{algorithm}
\usepackage{subfig}
\usepackage{color}
\usepackage[english]{babel}
\usepackage{graphicx}
\usepackage{wrapfig,epsfig}
\usepackage{epstopdf}
\usepackage{url}
\usepackage{graphicx}
\usepackage{color}
\usepackage{epstopdf}
\usepackage{algpseudocode}
\usepackage{scrextend}
\usepackage[T1]{fontenc}
\usepackage{bbm}
\usepackage{comment}
\usepackage{authblk}
\usepackage{multicol}
\usepackage{thmtools,thm-restate}

\let\C\relax
\usepackage{tikz}
\usepackage{hyperref}  %%% arxiv don't allow this.
\hypersetup{colorlinks=true,citecolor=red,linkcolor=blue} %%% Zhao : maybe we should comment this in submission.
\usetikzlibrary{arrows}
\usepackage[margin=1in]{geometry}
\graphicspath{{./figs/}}

\title{An Improved Cutting Plane Method for Convex Optimization, Convex-Concave Games and its Applications\thanks{A preliminary version of this paper appears in Proceedings of the 52nd ACM Symposium on Theory of Computing (STOC 2020). The authors would like to express their sincere gratitude to matrix multiplicationer Josh Alman and Jeroen Zuiddam for their patient and answers of our exponential number of questions about fast matrix multiplication.}} 
\author{Haotian Jiang\thanks{\texttt{jhtdavid@uw.edu}. University of Washington.}
\quad
Yin Tat Lee\thanks{\texttt{yintat@uw.edu}. University of Washington and Microsoft Research Redmond. Research supported in part by NSF Awards CCF-1740551, CCF-1749609, and DMS-1839116.} 
\quad
Zhao Song\thanks{\texttt{zhaos@ias.edu}. Princeton University and Institute for Advanced Study. Work done while visiting University of Washington.} 
\quad
Sam Chiu-wai Wong\thanks{ \texttt{samwon@microsoft.com}. Microsoft Research Redmond.}
}

\date{}

\newtheorem{theorem}{Theorem}[section]
\newtheorem{lemma}[theorem]{Lemma}
\newtheorem{definition}[theorem]{Definition}

\newtheorem{corollary}[theorem]{Corollary}
\newtheorem{conjecture}[theorem]{Conjecture}
\newtheorem{assumption}[theorem]{Assumption}

\newtheorem{fact}[theorem]{Fact}

\newtheorem{claim}[theorem]{Claim}

\newtheorem*{theorem*}{Theorem}

\newcommand{\wh}{\widehat}
\newcommand{\wt}{\widetilde}

\newcommand{\R}{\mathbb{R}}
\newcommand{\T}{\mathcal{T}}

\newcommand{\norm}[1]{\left\lVert#1\right\rVert}
\renewcommand{\varepsilon}{\epsilon}
\renewcommand{\tilde}{\wt}
\renewcommand{\hat}{\wh}

\renewcommand{\d}{\mathrm{d}}

\DeclareMathOperator*{\E}{{\bf {E}}}
\DeclareMathOperator*{\Var}{{\bf {Var}}}
\DeclareMathOperator*{\C}{\mathbb{C}}

\DeclareMathOperator{\poly}{poly}

\DeclareMathOperator{\rank}{rank}

\DeclareMathOperator{\dis}{dis}

\DeclareMathOperator{\tr}{tr}

\DeclareMathOperator{\diag}{diag}

\DeclareMathOperator{\cts}{cts}
\DeclareMathOperator{\jl}{jl}
\DeclareMathOperator{\new}{new}
\DeclareMathOperator{\simp}{simp}
\DeclareMathOperator{\comp}{comp}
\DeclareMathOperator{\inn}{inn}
\DeclareMathOperator{\out}{out}
\DeclareMathOperator{\ctr}{ctr}

\DeclareMathOperator{\SO}{SO}
\DeclareMathOperator{\mat}{mat}
\renewcommand{\mid}{\mathrm{mid}}

\newcommand{\interior}{\mathsf{int}}
\newcommand{\vol}{\mathsf{vol}}

\newcommand{\Zhao}[1]{{\color{red}[Zhao: #1]}}
\newcommand{\haotian}[1]{{\color{red}[Haotian: #1]}}

\makeatletter
\newcommand*{\RN}[1]{\expandafter\@slowromancap\romannumeral #1@}
\makeatother

\usepackage{lineno}

\usepackage{etoolbox}

\makeatletter

\newcommand{\define}[4][ignore]{%
  \ifstrequal{#1}{ignore}{}{
  \@namedef{thmtitle@#2}{#1}}%
  \@namedef{thm@#2}{#4}%
  \@namedef{thmtypen@#2}{lemma}%
  \newtheorem{thmtype@#2}[theorem]{#3}%
  \newtheorem*{thmtypealt@#2}{#3~\ref{#2}}%
}

\newcommand{\state}[1]{%
  \@namedef{curthm}{#1}
  \@ifundefined{thmtitle@#1}{
  \begin{thmtype@#1}
    }{
  \begin{thmtype@#1}[\@nameuse{thmtitle@#1}]
  }
    \label{#1}
    \@nameuse{thm@#1}
  \end{thmtype@#1}
  \@ifundefined{thmdone@#1}{
  \@namedef{thmdone@#1}{stated}%
  }{}
}

\newcommand{\restate}[1]{%
  \@namedef{curthm}{#1}
  \@ifundefined{thmtitle@#1}{
    \begin{thmtypealt@#1}
    }{
  \begin{thmtypealt@#1}[\@nameuse{thmtitle@#1}]
  }
    \@nameuse{thm@#1}
  \end{thmtypealt@#1}
  \@ifundefined{thmdone@#1}{
  \@namedef{thmdone@#1}{stated}%
  }{}
}

\newcommand{\thmlabel}[1]{
  \@ifundefined{thmdone@\@nameuse{curthm}}{\label{#1}
    }{\tag*{\eqref{#1}}}
}
\makeatother
\begin{document}
%\linenumbers

\begin{titlepage}
  \maketitle
  \begin{abstract}
Given a separation oracle for a convex set $K \subset \mathbb{R}^n$ that is contained in a box of radius $R$, the goal is to either compute a point in $K$ or prove that $K$ does not contain a ball of radius $\epsilon$. 
We propose a new cutting plane algorithm that uses an optimal $O(n \log (\kappa))$ evaluations of the oracle and an additional $O(n^2)$ time per evaluation, where $\kappa = nR/\epsilon$.
\begin{itemize}
\item This improves upon Vaidya's $O( \text{SO} \cdot n \log (\kappa) + n^{\omega+1} \log (\kappa))$ time algorithm [Vaidya, FOCS 1989a] in terms of polynomial dependence on $n$, where $\omega < 2.373$ is the exponent of matrix multiplication and $\text{SO}$ is the time for oracle evaluation.
\item This improves upon Lee-Sidford-Wong's $O( \text{SO} \cdot n \log (\kappa) + n^3 \log^{O(1)} (\kappa))$ time algorithm [Lee, Sidford and Wong, FOCS 2015] in terms of dependence on $\kappa$.
\end{itemize}
For many important applications in economics, $\kappa = \Omega(\exp(n))$ and this leads to a significant difference between $\log(\kappa)$ and $\poly(\log (\kappa))$. 
We also provide evidence that the $n^2$ time per evaluation cannot be improved and thus our running time is optimal.

A bottleneck of previous cutting plane methods is to compute {\em leverage scores}, a measure of the relative importance of past constraints.
Our result is achieved by a novel multi-layered data structure for leverage score maintenance, which is a sophisticated combination of diverse techniques such as random projection, batched low-rank update, inverse maintenance, polynomial interpolation, and fast rectangular matrix multiplication. Interestingly, our method requires a combination of different fast rectangular matrix multiplication algorithms. 

Our algorithm not only works for the classical convex optimization setting, but also generalizes to convex-concave games. We apply our algorithm to improve the runtimes of many interesting problems, e.g., Linear Arrow-Debreu Markets, Fisher Markets, and Walrasian equilibrium.

%Due to some evidence, we believe each iteration cannot avoid $n^2$ time cost and $n \log ( n R / \epsilon )$ iteration is necessary. We conjectured that our result is the optimal. \Zhao{I am not sure the best way to mention a conjecture.} 

%Finally, we would like to thank Josh Alman for answering infinitely many questions about matrix multiplication.

%Let $T(n,n,r)$ denote the time to compute the multiplication of an $n \times n$ matrix and an $n \times r$ matrix. 
%Previous algorithm for LP uses result that $T(n,n,r) = n^{\omega+o(1)}$ for $r = n$ ([Vaidy, FOCS 1989b], [Cohen, Lee and Song STOC 2019]) and $T(n,n,r) = n^{2+o(1)}$ for $r=n^{\alpha}$ ([Cohen, Lee and Song STOC 2019]). We also need to use $T(n,n,r)= n^2 \poly(\log n)$ for $r = n^{0.17}$ and $T(n,n,r) = n^2$ for $r = \poly(\log n)$.
%{\color{red} @Yintat: we need your help to make sure the abstract makes sense.}

  \end{abstract}
  \thispagestyle{empty}
\end{titlepage}

\newpage
{\hypersetup{linkcolor=black}
\tableofcontents
}
\newpage

\section{Introduction}

The cutting plane methods are a class of optimization primitives for
solving Linear and Convex Programming, which is encapsulated by the
\emph{feasibility} problem of finding a point in a given convex set
$K$ equipped with a separation oracle. In each iteration, cutting
plane methods query the \emph{separation} oracle which returns a hyperplane
separating the query point from $K$. Since Khachiyan's breakthrough
result \cite{k80} on the \emph{ellipsoid method}, cutting plane methods have
played a central role in theoretical computer science, showing that
a plethora of problems from diverse areas admit polynomial time algorithms.
Early prominent examples include linear programming and submodular
function minimization.

While more applications of the cutting plane methods are still being
discovered to this date, progress on faster cutting plane methods
had stagnated until the recent work of Lee, Sidford and Wong \cite{lsw15}.
Prior to their work, cutting plane methods were considered too slow to be of relevance in theory or in practice (other than establishing polynomiality). 
\cite{lsw15} debunked this by showing how their cutting plane method is faster than tailor-made algorithms for a long list of well-studied problems.
% I rearranging this a bit to make sure the reader know the too slow is a misconception, it is not really because LSW is faster.

The fastest algorithm before their work was Vaidya's algorithm \cite{v89}, which runs in $O(n^{\omega+1})$ time. Here $\omega<2.373$ is the exponent of matrix multiplication \cite{cw87,w12,ds13,l14}.
Leveraging recent advances in optimization and numerical linear algebra,
LSW lowered the dependence on the dimension $n$ at the expense of
additional log factors in the accuracy $1/\kappa$, namely $\log^{2}(\kappa)$.
The following table compares the runtimes of Vaidya's and LSW methods,
which both achieve the optimal number of oracle calls of $n \log(\kappa)$.

\begin{table}[!h]
\centering%
\begin{tabular}{|l|l|l|l|}
\hline 
{\bf Authors} &  {\bf \#Oracle Calls} & {\bf Runtimes} \\
\hline 
\hline 
Vaidya & $n \log (\kappa)$ & $n^{\omega+1} \log (\kappa)$ \\ \hline
Lee-Sidford-Wong & $n \log (\kappa)$ & $n^{3} \log^{O(1)} (n) \log^{3} (\kappa) $\\ \hline
\end{tabular}
\end{table}

This extra overhead of $\log^{2}(\kappa)$, as exemplified by various
problems in combinatorial optimization in their paper, translates
into only a log-squared factor in the maximum value $M$ of the input
by taking $\epsilon=O(1/M)$. This is because solutions to (polynomial-time
solvable) problems in combinatorial optimization are guaranteed to
be integral. Despite being a nuisance, this small overhead is relatively
mild and can even be absorbed completely for unweighted problems and
strongly polynomial time algorithms. Indeed, armed with this faster
cutting plane method, they improved the state-of-the-art running times
for a host of problems such as semidefinite programming, matroid intersection
and submodular minimization.

\paragraph{Importance of $\log (\kappa)$ vs $\log^3 (\kappa)$.} For many other problems including linear programming and
market equilibrium computation, $\epsilon$ must be taken as $1/M^{O(n)}$
resulting in an additional factor $n^{2}$ between $\log (\kappa)$ and
$\log^{3} (\kappa)$. For these applications, LSW is slower than Vaidy's
by a factor of $\tilde{O}(n^{4-\omega})$. This raises a natural question:\emph{
is there a cutting plane method that simultaneously runs in $O(n\log (\kappa))$ calls
and $O(n^{3}\log^{O(1)} (n) \log (\kappa))$ time? That is, can we achieve
the best of both worlds of Vaidya's and LSW methods in terms of the
dependence on $n$ and $\kappa$?}

In this paper, we answer this question in the affirmative. Somewhat
surprisingly, we are able to remove $\log^{O(1)} (n)$ dependence in
LSW as well.

\begin{theorem}[Main result] \label{thm:cutting_plane_method_intro}
There is a cutting plane method which runs in time $O(n\cdot \SO \log (\kappa)+n^{3}\log (\kappa))$,
where $\SO$ is the time complexity of the separation oracle.
\end{theorem}

\begin{table}[!t]
\centering
    \begin{tabular}{ | l | l | l | l | l | }
    \hline
    {\bf Reference} & {\bf Year} & {\bf Algorithm} & {\bf Complexity} \\ \hline \hline
     \cite{s77,yn76,k80} & 1979 & Ellipsoid Method & $n^2  \SO  \log (\kappa) + n^4 \log (\kappa)$ \\ \hline
     \cite{kte88,nn89} & 1988 & Inscribed Ellipsoid & $n \SO \log (\kappa) + (n \log (\kappa))^{4.5}$ \\ \hline 
     \cite{v89} & 1989 & Volumetric Center & $n \SO \log (\kappa) + n^{\omega+1} \log (\kappa)$ \\ \hline
     \cite{av95} & 1995 & Analytic Center & $n \SO \log^2 (\kappa) + n^{\omega+1} \log^2 (\kappa) + ( n \log (\kappa) )^{ 2 + \omega / 2 } $ \\ \hline
     \cite{bv02} & 2004 & Random Walk & $ n \SO \log (\kappa) + n^7 \log (\kappa)$ \\ \hline
     \cite{lsw15} & 2015 & Hybrid Center & $n \SO \log (\kappa) + n^3 \log^3 (\kappa)$ \\ \hline
     Theorem~\ref{thm:cutting_plane_method_intro} & 2019 & Volumetric Center & $n \SO \log (\kappa) + n^3 \log (\kappa) $ \\ \hline
    \end{tabular}
    \caption{ \small{Algorithms for the Feasibility Problem. Let $\kappa = nR/\epsilon$. All methods can be used to solves a more general problem where only a membership oracle is given \cite{lsv18}.}}\label{tab:cutting_plane_method} %We ignore the Big-O.
\end{table}

\begin{table}[htp!]
    \centering
    \begin{tabular}{ | l | l | l | l | l | }
        \hline
        {\bf Reference} & {\bf Year} & {\bf \#Operations} & {\bf Poly Type } \\ \hline \hline
        \cite{e75}    & 1975    & Finite      & Not poly \\ \hline
        \cite{j07} & 2007  & Polynomial     & Weakly poly \\ \hline
        \cite{y08}         & 2008    & $ n^6 \log(nU)$  & Weakly poly \\ \hline
        \cite{dpsv08}  & 2008    & Polynomial   & Weakly poly \\ \hline
        \cite{dm15} & 2015 & $ n^9 \log(nU)$  &  Weakly poly \\ \hline
        \cite{dgm16}   & 2016    & $ n^6 \log^2 (nU)$ & Weakly poly \\ \hline
        \cite{gv19}   & 2019    & $ m n^{9} \log^2 (n)$  & Strongly poly \\ \hline
        Theorem~\ref{thm:arrow_debreu_market_intro} & 2019 & $ m n^2 \log (nU) $ & Weakly poly \\ \hline
    \end{tabular}
    \caption{\small{Linear Arrow-Debreu Markets. Let $n$ be the number of agents and $m$ the number of edges. Each operation of the given algorithms involves $O(n \log(nU))$-bit numbers.} }% $M(l)$ denotes the time to perform a basic arithmetic operations on $l$-bit numbers.}
    \label{tab:arrow_debreu_market_intro}
\end{table}

\begin{table}[htp!]
    \centering
    \begin{tabular}{ | l | l | l | l | l | }
        \hline
        {\bf Reference} & {\bf Year} & {\bf \# Operations}  &  {\bf Poly Type } \\ \hline \hline
        \cite{v10} & 2010 &  $ n^3 (n+m)^2 \log (U) \cdot \T_{\text{max-flow}}$  & Weakly poly \\ \hline
        \cite{v16} & 2016 & $mn^3 + m^2 (m+n \log(n))\log(m)$ & Strongly poly  \\ \hline
        \cite{wan16} & 2016 & $m^3 n + m^2 \log(n) (n \log(n) + m) $  & Strongly poly \\ \hline
        Theorem~\ref{thm:fisher_market_intro} & 2019 & $m n^2 \log (n U)$ & Weakly poly \\ \hline 
    \end{tabular}
    \caption{\small{Fisher Markets with Spending Constraint Utilities. Let $n$ denote the total number of buyers and sellers, and $m$ the total number of segments. 
    $\T_{\text{max-flow}}$ denotes the number of operations needed for a max-flow computation. Each operation of the given algorithms involves $O(n \log(nU))$-bit numbers. } }% $M(l)$ denotes the time to perform a basic arithmetic operations on $l$-bit numbers.} 
    \label{tab:fisher_markets_intro}
\end{table}

\begin{table}[htp!]
    \centering
    \begin{tabular}{ | l | l | l | l | l | }
        \hline
        {\bf Reference} & {\bf Year} & {\bf \# Operations} &  {\bf Poly Type } \\ \hline \hline
        \cite{kc82} & 1982 & Finite & Not poly \\ \hline
        \cite{p99} & 1999 & Finite & Pseudo poly  \\ \hline
        \cite{pu02} & 2002 & Finite & Pseudo poly \\ \hline
        \cite{am02} & 2002 & Finite & Pseudo poly \\ \hline
        \cite{dsv07} & 2007 & Finite & Pseudo poly \\ \hline
        \cite{lw17} & 2017 &  $ n^2 \T_{\mathrm{AD}} \log(SMn)+n^6 \log^{O(1)}(SMn)$ & Weak poly  \\ \hline
        Theorem~\ref{thm:walrasian_equilibrium_intro} & 2019 & $ n^2 \T_{\mathrm{AD}} \log(SMn)+n^4 \log(SMn)$ & Weak poly \\ \hline
    \end{tabular}
    \caption{\small{Walrasian equilibrium for general buyer valuations and fixed supply. Let $\T_{\mathrm{AD}}$ denotes the runtime of aggregate demand oracle, $n$ denotes the number of goods. }} 
    \label{tab:walrasian_equilibrium_intro}
\end{table}

As with previous methods, our result achieves the asymptotic optimal
oracle complexity of $O(n\cdot \SO\log (\kappa))$. Moreover, we conjecture
that the runtime is likely the best possible. Note that in each iteration
where a separation oracle call is made, even basic matrix operations
require $O(n^{2})$ time. Thus our runtime is essentially tight unless
properties like sparsity can be surprisingly exploited to update the feasible region
and compute the next query point for the separation oracle.

Our result is obtained by marrying Vaidya's method with recent advances
in numerical linear algebra. Similar to LSW, we implement each iteration
of Vaidya's via tools from fast numerical linear algebra. The main
issue is that such tools rely on approximation and would lead to errors
which must be carefully controlled.

Our first innovation is to run Vaidya's method in \emph{phases}, each
of which consist of a certain number of iterations. Between phases
we ``recompute'' to eliminate the errors accumulated within a phase.
Because of this recomputation we can afford to tolerate higher errors
in an iteration. Secondly, we present a sophisticated data structure
that enables us to implement each iteration of Vaidya's efficiently.
Our data structure leverages recent advances on applying numerical
techniques to optimization. We hope that these numerical tools, as
well as our approach to applying them, would play a greater role in
future development of optimization.

\subsection{Applications}

We highlight some of the key applications of our faster cutting plane method. Interestingly, even though our cutting plane method is a general purpose algorithm, we are able to improve the runtimes of tailor-made algorithms for various problems.

Using a standard reduction of convex minimization to the feasiblity problem (\cite{nemi94} and Theorem 42 of~\cite{lsw15}), we can minimize a convex function with an optimal $\tilde{O}(n)$ subgradient oracle calls and an additional $O(n^2)$ time per oracle call. 

\begin{theorem}[Informal version of Theorem~\ref{thm:convex}]\label{thm:convex_intro}
Let $f$ be a convex function on $\R^n$ and $S$ be a convex set that contains a minimizer of $f$. Suppose we have a subgradient oracle for $f$ with cost $\mathcal{T}$ and $S\subset B(0,R)$. Using $B(0,R)$ as the initial polytope for our Cutting Plane Method, for any $0<\alpha<1$, we can compute $x\in S$ such that $f(x)-\min_{y\in S}f(y)\leq \alpha\left(\max_{y\in S}f(y)-\min_{y\in S}f(y)\right)$, with high probability in $n$ and with a running time of $O(\T \cdot n \log(\kappa)+n^{3}\log (\kappa))$, where $\kappa =  n\gamma / \alpha$ and $\gamma=R /(\mathrm{minwidth}(S))$.
\end{theorem}

Our convex minimization result can be further generalized to convex-concave games.

\begin{theorem}[Informal version of Theorem~\ref{thm:saddle_point}]\label{thm:saddle_point_intro}
Given convex sets $\mathcal{X}\subset B(0,R)\subset\R^{n}$
and $\mathcal{Y}\subset B(0,R)\subset\R^{m}$ such that both $\mathcal{X}$ and $\mathcal{Y}$ contain a ball of radius $r$. Let $f(x,y):\mathcal{X}\times\mathcal{Y}\rightarrow\R$ be
an $L$-Lipschitz function 
that is convex in $x$ and concave in $y$. 
Define $\kappa = \frac{n+m}{\epsilon} \frac{R}{r}$.
For any $ \epsilon \in (0,1/2]$,
we can find $(\hat{x},\hat{y})$ such that
$ \max_{y\in\mathcal{Y}}f(\hat{x},y)-\min_{x\in\mathcal{X}}f(x,\hat{y})\leq\epsilon L r $
in time 
$
O \left( \mathcal{T}  \cdot (n + m) \log (\kappa) + (n+m)^{3} \log (\kappa) \right)
$
with high probability in $n+m$ where $\mathcal{T}$ is the cost of computing subgradient $\nabla f$.
\end{theorem}

Leveraging this improved dependence on $\kappa$ (and hence $\epsilon$),
our cutting plane method can be used to improve the runtimes of a wide range of problems, especially those on market equilibrium computation (see Table~\ref{tab:arrow_debreu_market_intro},~\ref{tab:fisher_markets_intro} and~\ref{tab:walrasian_equilibrium_intro} for a summary of previous runtimes). 
In all our applications, $1/\epsilon$ needs to be exponentially large in $n$ which renders the $\log(\kappa)$ factor polynomially large.
%The formal discussions of our applications are given in Section~\ref{sec:application}.

We show the following runtime improvement for the problem of computing a market equilibrium in linear exchange markets:

\begin{theorem}[Informal version of Theorem~\ref{thm:arrow_debreu_market}]\label{thm:arrow_debreu_market_intro}
There exists a weakly polynomial algorithm that computes a market equilibrium in linear exchange markets in time $O(mn^2 \log (nU))$.
\end{theorem}

The celebrated result of Arrow and Debreu~\cite{ad54} shows the existence of a market equilibrium for a broad class of utility functions. Since then researchers have attempted to design efficient algorithms to compute market equilibria.
One prominent special case of linear utilities has enjoyed significant attention as demonstrated by the long line of work in Table~\ref{tab:arrow_debreu_market_intro}.
Essentially, this problem can be captured by a convex program with linear constraints~\cite{dgv16} (see (\ref{eqn:GoodFullySold})).
While this convex program exhibits certain advantageous features over previous ones~\cite{c89,j07,np83}, it had not led to improved runtimes since the objective is not separable~\cite{gv19}. 
Moreover, the number of variables in the convex program can be as large as $O(n^2)$, which prohibits a fast runtime for the cutting plane method if applied directly. 
Our approach in Theorem~\ref{thm:arrow_debreu_market_intro} is to transform the convex program into a convex-concave game with $O(n)$ variables and apply Theorem~\ref{thm:saddle_point_intro}.

A similar technique can be applied to the problem of computing market equilibrium for Fisher markets with spending utility constraints where a convex program is given in~\cite{bdx10} (see (\ref{eqn:fisher_market_program})). 
However, directly transforming it into a convex-concave game does not reduce the dimension of the variables, as both the number of variables and constraints in the original convex program is $\Theta(m)$, where $m$ is the total number of segments and it can be much larger than $n^2$. 
In order to reduce the dimension of the convex-concave game to $O(n)$, we express part of the variables as functions of $O(n)$ variables and show that this does not increase the time of the first-order oracle.
This leads to the following runtime improvement: 

\begin{theorem}[Informal version of Theorem~\ref{thm:fisher_market}]\label{thm:fisher_market_intro}
There exists a weakly polynomial algorithm that computes a market equilibrium in Fisher markets with spending constraint utilities in time $O(mn^2 \log (nU))$.
\end{theorem}

Yet another economics application of our cutting plane method is the problem of computing a Walrasian equilibrium in a market with fixed supply. 
In this economy, buyers may have arbitrary valuation functions and we would like to compute prices so that the market clears, i.e. the aggregate demand of the buyers matches the fixed supply. 
Recently Paes Leme and Wong~\cite{lw17} gave a polynomial time algorithm for this problem \emph{provided} that an equilbrium actually exists. They achieved this by showing that the cutting plane method can be modified so that the convex program for the equilibrium can be solved under the aggregate demand oracle.

By leveraging our faster cutting plane method we obtain an improved runtime:

\begin{theorem}[Informal version of Theorem~\ref{thm:walrasian_equilibrium}]\label{thm:walrasian_equilibrium_intro}
There is an algorithm that runs in time \\$O(n^{2} \T_{\mathrm{AD}}\log(SMn)+n^{4}\log (SMn))$
for computing a market equilibrium in an economy with general buyer
valuation in the aggregate demand model.
\end{theorem}

\subsection{Previous works}

The cutting plane methods solve the following \emph{feasbility} problem
that conveniently abstracts the applications to specific scenarios.
\begin{itemize}
\item[] \textbf{Feasibility Problem}: Given a separation oracle for a set
$K$ contained in a box of radius $R$ either find a point $x\in K$
or prove that $K$ does not contain a ball of radius $\epsilon$.
\end{itemize}
All cutting plane methods maintain a candidate region $\Omega$ and
solve the feasibility problem by iteratively refining $\Omega$ based
on the present $\Omega$ and the new separating hyperplane. In each
iteration:\\
1. The separation oracle is queried at some point $x\in\Omega$.\\
2. If $x\in K$ we have solved the feasibility problem.\\
3. Otherwise, the separation oracle returns a separating hyperplane from
which $\Omega$ is further refined and the next query point $x$ is
computed.\\

Previous works differ in how $x$ is selected and how $\Omega$ is
refined. For instance, the classic ellipsoid method maintains $\Omega$
as an ellipsoid and $x$ as its center. Given $\Omega$ and the new
separating hyperplane, the new $\Omega$ is chosen to be the smallest
ellipsoid containing their intersection. Table~\ref{tab:cutting_plane_method} lists
the running times for solving the feasibility problem in the literature.

We focus our discussion on the trade-off between the oracle complexity
and the runtime per iteration. The oracle complexity has a lower bound $\Omega(n\log (\kappa))$~\cite{nemi83}. While the ellipsoid method achieves
a suboptimal $O(n^{2}\log (\kappa))$ in oracle complexity, the runtime
per iteration is $O(n^{2})$ which is faster than all subsequent methods.
%until LSW. No..LSW still have extra polylog
 The good runtime follows from the simple calculations needed
to update the ellipsoid, whereas the suboptimal oracle complexity
can be attributed to the ``looseness'' of maintaining only an ellipsoid
as a proxy to the intersection of past separating half-spaces.

Indeed, one can attain the optimal oracle complexity $O(n^{2}\log (\kappa))$
by maintaining \emph{all} previous separating half-spaces. This is
the random walk method~\cite{bv02} where the query point $x$ is chosen to be
its (approximate) center of gravity. Updating $x$ involves performing
a random walk in this polytope and is computationally expensive.

Other cutting methods improve oracle complexity by maintaining more
fine-grained information about past separating hyperplanes in a way
that is computationlly friendly. Of particular relevance is Vaidya's
volumetric center method~\cite{v89}. Vaidy's $\Omega$, similar to the random
walk method, is also a polytope $\Omega=\{x\in\mathbb{R}^{n}\,:\,Ax\geq b\}$.
As the name suggests, the query point is chosen to be the \emph{volumetric
center}, which is the minimizer of the following convex function defined
by the feasible region $\Omega$:
\begin{align*}
\frac{1}{2}\log\det\left( A^{\top} S_{x}^{-2} A \right)\enspace\text{ where }\enspace S_{x}:=\text{diag}(Ax-b)\text{ is the diagonal matrix of the slacks}
\end{align*}

Nevertheless, by judiciously including only a representative subset
of previous half-spaces $ax\geq b$, Vaidya showed that the volumetric
center and $\Omega$ can be updated by basic matrix operations which
run in $O(n^{\omega})$ time.

We defer the discussion of LSW to the next subsection, which explains
how Vaidya's method can be sped up using machineries from numerical
linear algebra.

\subsection{Lee-Sidford-Wong method (LSW)}

LSW's key observation is that Vaidya's method relies heavily on \emph{leverage
scores}, which measure the relative importance of each separating
hyperplane. In Vaidya's method, naively updating leverage scores requires
$O(n^{\omega})$ time and is a bottleneck. Inspired by the work of
Spielman and Srivistava \cite{ss11}, LSW attempted to address this by using random Johnson-Lindenstrauss \cite{jl84} (JL) projection to approximate \emph{changes}
in leverage scores. Leverage scores can then be updated by summing
over the differences.

Approximating leverage scores changes via JL projection however still
requires solving a linear system which would still take $O(n^{\omega})$
time. LSW further overcame this barrier by resorting to a recent work
that efficiently solves ``slowly-changing'' linear system in amortized
$\tilde{O}(n^{2})$ time. Thus after paying $O(n^{\omega})$ initially,
they can solve such linear systems in $\tilde{O}(n^{2})$ time per
iteration.

\smallskip
\noindent \textbf{Error accumulation.} 
Nevertheless, as an approximate
method JL introduces errors which would accumulate across iterations.
While Vaidya's method tolerates small errors in leverage scores, the
total errors incurred in JL projection (as accumulated across iterations)
would eventually become too big and destroy the performance guarantee
of Vaidya's method.

LSW handled this by modifying Vaidya's framework to take into account
of the error in the convex function to be minimized. This approach
gives rise to a ``hybrid'' center algorithm, which involves a complicated
interplay of optimization and linear algebra. In particular, to reduce
the error accumulated the error parameter $\epsilon_0$ in JL projection
has to be as small as $\epsilon_0=1/\log (\kappa)$. As the runtime of
JL depends on $1/\epsilon_0^{2}$, this unfortunately introduces the
$\log^{2}\kappa$ overhead in LSW runtime when compared to Vaidya's
volumetric center method.

\subsection{Overview of our approach} \label{subsec:tech}

%Our method is inspired by LSW approach of estimating differences in leverage scores. 
To achieve the desired $O(n\SO\log (\kappa)+n^{3}\log (\kappa))$
runtime, 
%we invent several new ideas to address the error accumulation issue.
%\paragraph{Overview of data structure : Layered Approach.} 
we build a data structure that approximates
changes in leverage scores to within $c$ in $\ell_2$ norm for small enough constant $c$. 
This would
suffice for Vaidya's cutting plane method. The desired runtime then
directly follows from the performance guarantee of the data structure,
which handles each iteration in amortized $O(n^{2})$ time.
Designing such a data structure requires several new ideas to control the error accumulation issue.

Our data structure employs a layered approach where different
layers are associated with different error tolerances, 
and achieve different accuracy-efficiency tradeoffs in approximating the changes in leverage scores.
%and approximates the changes in leverage scores in different finess by leveraging different formulas. 
The more inner a layer is, the more error it can tolerate and the faster is the runtime.
Whenever the error accumulated in a layer becomes too much, the layer
above would take over and produce a finer error estimate which would,
of course, be more time costly. But because of our layered approach,
the higher layer is called on less often and afford to spend more
time.

Such a layered approach further leads to the following issue.
In the middle and outer layers, we batch the updates of multiple steps into one which allows us to make use of fast rectangular matrix multiplication. However, our algorithm needs to handle possibly exponential weight changes. We show there are not too many such weights and can be handled separately in groups of $\poly \log (n)$ size using low-rank update formula.

Our data structure also draws on various numerical tools such as fast rectangular
matrix multiplication, ``tall'' JL projection, preconditioning, 
inverse maintenance, and polynomial interpolation for approximating
integrals.

A more in-depth discussion of our techniques can be found in Section~\ref{sec:tech}.

\subsection{Discussion of optimality}

Similar to previous methods our cutting plane method achieves the optimal oracle complexity $n \log (\kappa)$~\cite{nemi83}. We present some evidence that our running time of $O(n^3\log (\kappa))$ is also tight.

A bottleneck of Vaidya's method is to solve the inverse maintenance problem. Formally, given a sequence of positive vectors $w^{1}, w^{2}, \cdots w^{T}$, let $P(w)$ be defined as $$P(w) = \sqrt{W} A ( A^\top W A )^{-1} A^\top \sqrt{W},$$ where $W$ is the diagonal matrix such that $W_{i,i} = w_i$. The goal is to output a sequence of vectors $v^{1}, v^{2}, \cdots, v^{T}$ such that 
\begin{align*}
v^t \approx w^t \text{~~~and~~~} P(v^t) \approx P(w^t), \forall t \in [T].
\end{align*} 

There is a long line of research on inverse maintenance and dynamic matrix data-structure problems \cite{k80,v89_lp,s04,ls15,hkns15,cls19,lsz19,song19,bns19,b20}. 
This task can be done naively by spending $n^\omega$ time so the goal is to achieve $o(n^{\omega})$ amortized cost per iteration. For example, in the LP setting the number of iterations is $O(\sqrt{n})$ and Vaidya \cite{v89_lp} combined fast matrix multiplication with inverse maintenance to achieve $O(n^2)$ amortized cost per iteration, which gives an $O(n^{2.5})$ time algorithm. This remained a barrier until recent works \cite{cls19,lsz19} combined sampling and sketching techniques with fast rectangular matrix multiplication and inverse matrix maintenance to give an $O(n^{\omega})$ time algorithm. 

One of the major computation required in each step is matrix-vector multiplication, e.g., $P(w) \cdot h$. Naively, this step takes $O(n^2)$ time per iteration. To achieve $o(n^2)$ amortized cost per iteration, previous works \cite{cls19,lsz19} used an idea called ``iterating and sketching'' which was formally described in \cite{song19}. This idea is very different from the classical ``sketch and solve'' \cite{cw13} and ``guess a sketch'' \cite{rsw16}. The classical idea usually applying sketching matrices only once without modifying the solver itself. However, the ``iterating and sketching'' idea has to modify the solver and applying sketching/sampling matrices over each iteration. 

In \cite{cls19}, they use the diagonal sampling matrix $D \in \R^{n \times n}$ which has roughly $\sqrt{n}$ nonzeros. They use that sampling matrix to {\em sample on the right} hand side:
\begin{align*}
\sqrt{W} A ( A^\top W A )^{-1} A^\top \sqrt{W} \underbrace{D}_{\text{sample~right}} h.
\end{align*}
In \cite{lsz19}, they use the subsampled randomized Hadamard/Fourier matrix \cite{ldfu13,psw17} $R\in \R^{\sqrt{n} \times n}$. They use the sketching matrix to {\em sketch on the left} hand side:
\begin{align*}
\underbrace{ R^\top R }_{ \text{sketch~left} } \sqrt{W} A ( A^\top W A )^{-1} A^\top \sqrt{W} h.
\end{align*}

Compared to the cutting plane method, LP is an easier maintenance task as the matrix $A$ is fixed throughout.
In the cutting plane method, however, rows get inserted into or deleted from $A$ from continuously. 
One critical idea used in all previous works on LP~\cite{v89_lp,cls19,lsz19} is to delay low-rank updates on $( A^\top W A )^{-1}$. 
However, in the cutting plane method, the low rank updates to $A$ cannot be delayed. Thus it appears that previous techniques are inapplicable.
%there has to be a rank-$1$ update since the matrix $A$ is constantly changing. Therefore, we cannot delay the low rank update.
%\haotian{Do we want to mention this lower bound?}
Moreover, $n^2$ lower bounds have recently been established for natural matrix maintenance tasks (e.g. determinant, inverse) with row/column insertions/deletions under the Online Matrix-Vector conjecture (e.g.~\cite{hkns15,bns19}).
Therefore, we believe our algorithm is tight and conjecture the following:

\begin{conjecture}
Solving the feasibility problem requires $\Omega(n^3 \log (\kappa))$ time. Hence our cutting plane method achieves the optimal runtime.
\end{conjecture}

\subsection{Related works}

\paragraph{Leverage scores.} Leverage scores are a fundamental concept in graph problems and numerical linear algebra. There are many works about how to approximate leverage scores \cite{ss11,dmmw12,cw13,nn13} or more general version of leverages, e.g. Lewis weights \cite{l78,blm89,cp15} and ridge leverage scores \cite{cmm17}. From graph perspective, it was applied to solve max-flow \cite{m13_flow,m16}, generate random spanning trees \cite{s18}, and sparsify graphs \cite{ss11}. From matrix perspective, it was used to give matrix CUR decomposition \cite{bw14,swz17,swz19} and tensor CURT decomposition \cite{swz19}. From optimization perspective, it was used for approximating the John Ellipsoid \cite{ccly19}, accelerating the kernel methods \cite{akmmvz17,akmmvz19}, showing the convergence of the deep neural network \cite{lsswy20}, cutting plane methods, e.g. \cite{v89,lsw15} and this paper.

\paragraph{Linear Program.} Linear Program is a fundamental problem in convex optimization and can be treated as an special case where one can apply the cutting plane method. There is a super long list of work focused on fast algorithms for linear program \cite{d47,k80,k84,v87,v89_lp,ls14,ls15,sidford15,lee16,cls19,lsz19,song19,b20,blss20}.

\paragraph{Membership oracle.}
Besides the separation oracle considered in this paper, there is another line of work on using the membership oracle to solve the feasibility problem \cite{p96,kv06,lv06,gls12,lsv18}. For a query point $x$, this oracle outputs $x\in K$ or $x\notin K$.

%\paragraph{Fast matrix multiplication.}

%The product of two matrices is one of the most basic operations. The naive algorithm runs in $O(n^3)$ time. The first subcubic time algorithm for matrix multiplication was proposed by Strassen in $O(n^{2.808})$. This paper opened a field which gradually reduced the matrix multiplication exponent $\omega$. In 1978, Pan \cite{p78} proved that $\omega < 2.796$. In 1988, Bini , Capovani, Romani and Lotti \cite{bcrl79} proved that $\omega < 2.78$. In 1981, Schonhage \cite{s81} proved that $\omega < 2.548$. In 1982, Romani \cite{r82} showed that $\omega < 2.517$. In 1982, Coppersmith and Winograd \cite{cw82} proved that $\omega < 2.496$ which is the first result that breaks 2.5 barrier. In 1986, Strassen \cite{s86} proposed laser method that is able to improve the bound to $\omega < 2.479$. In 1987, Coppersmith and Winograd \cite{cw87} proved that $\omega < 2.376$. Then, this result remained for more than two decades.

%In 2012, \cite{w12} Williams proved that $\omega < 2.3729$. In 2014 \cite{l14} proved that $\omega < 2.3728639$.

%From the other perspective, the best known lower bound for matrix multiplication is $\Omega(n^2 \log n)$ \cite{r02}. 

 %%% Section 1
%\newpage

\section{Our Techniques}\label{sec:tech}

Section~\ref{subsec:tech} provided a quick overview. Here we take a deeper dive into our techniques.

\subsection{Efficient approximation of changes in leverage scores}

The key to fast maintenance of leverage scores is an efficient way to approximate their changes between consecutive steps.  
While a fine-grained approximation leads to an accurate approximation, the time to compute such an approximation might be unaffordable. 
On the other hand, a coarse-grained approximation can be efficiently computed, but might lead to accumulating errors that blow up after a small number of steps. 
This leads to a tradeoff between accuracy and efficiency. 

Central to our data structure are a coarse-grained formula (Lemma~\ref{lem:leverage_score_moving_simple}) and a fine-grained formula for
the change in leverage scores (Lemma~\ref{lem:leverage_score_moving_complicated}). 
While the coarse-grained formula approximates the leverage score's change via a single integral, 
the fine-grained formula is a cocktail involving integrals, matrix inverse and matrix multiplication. 
For the coarse-grained formula, we simply estimate the integral by a single point along the integral (Lemma~\ref{lem:compute_sigma_change_simple}).
For the fine-grained formula, however, the approximation (Lemma~\ref{lem:compute_sigma_change}) is more involved and requires appropreiate numerical tools.
%While certain terms in the fine-formula admit direct exact computations, others must be approximated using appropriate numerical tools: 

Our course-grained and fine-grained formulas lead to two different data structures for leverage score maintenance: a simple {\em deterministic} data structure that achieves low running time but introduces a large error in each step, and a more complicated {\em randomized} data structure that incurs small error each step at the cost of efficiency.

Observe that we are allowed to recompute the leverage scores exactly after every $O(n^{\omega - 2 + o(1)})$ steps as exact computation of leverage scores takes $O(n^{\omega + o(1)})$ time. 
Therefore it suffices to control the error accumulation in $O(n^{\omega-2 + o(1)})$ steps. To achieve this, both the simple and complicated data structures are crucial.
While the simple data structure achieves the desired $O(n^2)$ time per step, the error accrued in the $O(n^{\omega-2 + o(1)})$ steps are too large for the cutting plane method. 
For the complicated data structure, we are able to achieve $\ell_2$-error $o(1)$ in $O(n^{\omega-2 + o(1)})$ steps, but the amortized running time would be $O(n^{2 + o(1)})$ per step.

\subsection{Layered data structure}

To achieve the best of both worlds of the simple and complicated data structures, we propose a data structure that combines both data structures and interpolates between accuracy and efficiency. 
Specifically, our data structure approximates
changes in leverage scores to $\ell_2$-error within $1/\log^{O(1)} (n)$ which would suffice for Vaidya's cutting plane method and handles each iteration in amortized $O(n^{2})$ time.

Our data structure employs a layered approach (see Figure~\ref{fig:our_algorithm}) where different
layers are associated with different error tolerances. 
The more inner a layer is, the more error it can tolerate and the faster is the runtime.
Whenever the error accumulated in a layer becomes too high, the layer
above would take over and produce a finer error estimate which would,
of course, be more time costly. But because of our layered approach,
the higher layer is called on less often and can afford to spend more
time.
More specifically, our layered data structure contains three layers. 
The inner and middle layers both employ the simple data structure to achieve computational efficiency while the outer layer uses the complicated data structure to ensure a low error.

Our layered approach builds on several fundamental results on matrix multiplication which we summarize in Theorem~\ref{thm:matrix_multiplication_intro}. 
We remark that Vaidya~\cite{v89_lp} used the first result, a recent LP solver~\cite{cls19}
used the first two, while our cutting plane method crucially depends on all the results in the table.

\begin{theorem}[Fast rectangular matrix multiplication results]\label{thm:matrix_multiplication_intro}
For any $n,r > 0$, denote $\T_{\mathrm{mat}}(n,n,r)$ the time to compute the multiplication of an $n \times n$ matrix and an $n \times r$ matrix\footnote{Note that $\T_{\mathrm{mat}}(n,n,r) = \T_{\mathrm{mat}}(n,r,n)$}. Then we have: \\

\begin{center}\vspace{-4mm} 
\begin{tabular}{ | l | l | l | l | l | l|}
    \hline
    \bf{Reference} & \bf{$r$} & \bf{$\T_{\mathrm{mat}}(n,n,r)$} & \bf{Layer}  \\
    \hline \hline
    \rm{\cite{l14}} & $n$ & $O(n^{\omega+o(1)})$ & \rm{Restart} \\ 
    \hline
    \rm{\cite{gu18}} & $n^{0.31}$ & $O(n^{2+o(1)})$ & \rm{Outer layer} \\
    \hline
    \rm{\cite{c82}}  & $n^{0.17}$ & $O(n^2 \log^2 n)$  & \rm{Middle layer} \\
    \hline 
    \rm{\cite{bd76}} & $\log^{O(1)} n$ & $O(n^2)$  & \rm{Inner layer} \\
    \hline
    \end{tabular}
\end{center}

\begin{comment}
\begin{center}\vspace{-4mm} 
\begin{tabular}{ | l | l | l | l | l | l|}
    \hline
    \bf{Reference} & \bf{$r$} & \bf{$\T_{\mathrm{mat}}(n,n,r)$} & \bf{Layer} & {\bf Exponent} \\
    \hline \hline
    \rm{\cite{l14}} & $n$ & $O(n^{\omega+o(1)})$ & \rm{Restart} & $\omega = 2.373$ \\ 
    \hline
    \rm{\cite{gu18}} & $n^{0.31}$ & $O(n^{2+o(1)})$ & \rm{Outer layer} & $\alpha_{\mat} = 0.31$ \\
    \hline
    \rm{\cite{c82}}  & $n^{0.17}$ & $O(n^2 \log^2 n)$  & \rm{Middle layer} & $\beta_{\mat} = 0.17$ \\
    \hline 
    \rm{\cite{bd76}} & $\log^{O(1)} n$ & $O(n^2)$  & \rm{Inner layer} & $\gamma_{\mat} = \log^{O(1)} n$ \\
    \hline
    \end{tabular}
\end{center}
\end{comment}
%0. \cite{l14}. For $r = n$, $\T_{\mathrm{mat}}(n,n,r) = O(n^{\omega+o(1)})$. Define $\omega = 2.373$. \\
%1. \cite{gu18}. For $r = n^{0.31}$, $\T_{\mathrm{mat}}(n,n,r) = O(n^{2+o(1)})$. Define $\alpha_{\mat} = 0.31$. \\
%2. \cite{c82}. For $r = n^{0.17}$, $\T_{\mathrm{mat}}(n,n,r) = O(n^2 \log^2 n)$. Define $\beta_{\mat} = 0.17$.\\
%3. \cite{bd76}. For $r = \log^{O(1)} n$, $\T_{\mathrm{mat}}(n,n,r) = O(n^2)$. Define $\gamma_{\mat} = \log^{O(1)} n$.
\end{theorem}

Our data structure also draws on various numerical tools such as fast rectangular
matrix multiplication, ``tall'' JL projection, preconditioning,
inverse maintenance, and polynomial interpolation for approximating
integrals, which we discuss in more details below.

\begin{figure}[htp!]
\centering
\includegraphics[width=0.98\textwidth]{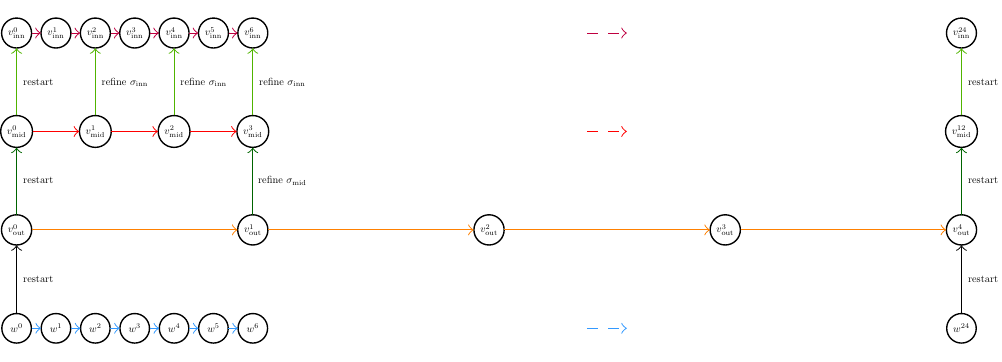}
\caption{\small Illustration of our three-level data structure with $T_{\text{inn}} = 2 $, $T_{\text{mid}} = 3$ and $T_{\text{out}} = 4$ approximating the leverage scores of the sequence $\{w^0, w^1, \cdots\}$. The three levels maintain three approximate sequences $\{ v_{\text{inn}}^0 , v_{\text{inn}}^1, \cdots \}$, $ \{ v_{\text{mid}}^0 , v_{\text{mid}}^1, \cdots \} $ and $\{ v_{\text{out}}^0 , v_{\text{out}}^1 , \cdots \}$, with errors $\epsilon_{\text{inn}} = \norm{ \log(w) - \log(v_{\text{inn}})}_\infty$, $\epsilon_{\text{mid}} = \norm{ \log(w) - \log(v_{\text{mid}})}_\infty$ and $\epsilon_{\text{out}} = \norm{ \log(w) - \log(v_{\text{out}})}_\infty$ that satisfy $1 \gg \epsilon_{\text{inn}} \gg \epsilon_{\text{mid}} \gg \epsilon_{\text{out}} > 0$.
The inner level takes a step for every $w$-update, the middle step takes a step in every $T_{\inn}$ inner steps, and the outer step takes a step in every $T_{\mid}$ middle steps. 
The entire data structure is restarted after $T_{\out}$ outer steps.
Each middle step refines the inner approximation of leverage scores, and each outer step refines the approximation of both the inner and middle approximations of the leverage scores.
For the actual choice of the parameters $T_{\text{inn}}$, $T_{\text{mid}}$, $T_{\text{out}}$ and $\epsilon_{\text{inn}}$, $\epsilon_{\text{mid}}$, $\epsilon_{\text{out}}$ in our data structure, see Table~\ref{table:table_of_constants_main}.}
\label{fig:our_algorithm}
\end{figure}

\subsection{Batched low-rank update}

The layered approach in the previous subsection leads to the following technical hurdle. 
While the inner layer reacts to each update, the middle and outer layers perform one step only after a certain number of updates.
In the cutting plane method, each update of the vector $w \in \R_+^m$ satisfies an $\ell_2$-closeness property: $\norm{\log(w^{\new}) - \log(w)}_2 = O(1)$. 
For a sequence of $T$ updates $w^{(0)}, w^{(1)}, \cdots, w^{(T)}$, the $\ell_2$-closeness property is no longer satisfied for $w^{(0)}$ and $w^{(T)}$, and the weight changes between $w^{(0)}$ and $w^{(T)}$ can be exponential in $T$. 

To resolve this issue, we employ a batched low-rank update method that computes a vector $v^{\mid}$ that is $\ell_\infty$-close to the vector $w^{(T)}$ and an accurate estimate of the leverage score change $\sigma(v^{\mid}) - \sigma(w^{0})$. This accurate estimate is done by a low-rank update rule for matrices (Woodbury matrix identity). Unfortunately, the Woodbury formula involves the inverse of certain matrices but the inverse maintenance data structure only maintains the inverse of a nearby matrix. Therefore, we need to use the maintained inverse as a preconditioner to approximately compute the Woodbury formula. Due to the exponential changes in the weight, we need to solve certain linear systems to exponential accuracy and it can be too expensive even with a very good preconditioner (Lemma \ref{lem:simple_low_rank}).

To rescue this preconditioner idea, we transform and split the sequence into pieces of size $\poly \log(n)$. We ensure the weight changes by only a quasi-polynomial factor and this decreases the cost of solving linear systems to $\poly \log(n)$ steps. Since we batch the task of handling $\poly \log(n)$ weight changes into one rectangular matrix multiplication which can be performed in $O(n^2)$ time, we make sure the cost per weight change is exactly $O(n^2)$ time (Theorem \ref{thm:leverage_batch}).

\subsection{Illustration of our analysis}

We describe the numerical tools used in our analysis, and provide simple illustrations of our applications of these tools. The actual way in which they are used in Sections~\ref{sec:simple_leverage_score}-\ref{sec:complicated_leverage_score} are more involved.

\paragraph{Discrete sampling for multiple variable integrals}

The most standard way to approximate an integral is by 
discretization, which takes a weighted sum of the integrand over a
set of points in the domain. 
Unlike common discretization tools like the trapezoidal method, for our purpose we need to interpolate multiple variable
integral using a polynomial for higher accuracy (Theorem~\ref{thm:multiple_variable_integral}).
We give a simple example to illustrate our application of polynomial interpolation for multiple variable integrals as follows. In our fine-grained formula of the leverage score change from $w$ to $w^{\new}$, one of the integral terms is 
\begin{align*}
\sigma_{i,\textrm{cts}} = \int_0^1 \int_0^1 \int_0^1 \gamma_{i,s,t}^\top \gamma_{i,s',t} \mathrm{d} s \mathrm{d} s' \mathrm{d} t ,
\end{align*}
where 
\begin{align*}
\gamma_{i,s,t} = & ~ \sqrt{ W^{\mid} - W^{\new}} \cdot Q(y_{s,t}) \cdot (Z_t - X_t) \cdot Q(y_{s,t}) \cdot \sqrt{ W^{\new} } \cdot e_i .
\end{align*}
In order to approximate such an integral, we take a set $\T \subseteq [0,1]$ of $N = \log^{O(1)} (n)$ points along the integration together with weights $\{\omega_t \}_{t \in \T}$ as in Theorem~\ref{thm:multiple_variable_integral}, and approximate the integral by 
\begin{align*}
\sigma_{i,\textrm{dis}} = \sum_{t \in \mathcal{T}} \sum_{s \in \mathcal{S}} \sum_{s' \in \mathcal{S}} \omega_t \omega_s \omega_{s'} \gamma_{i,s,t}^\top \gamma_{i,s', t}.
\end{align*} 
Theorem~\ref{thm:multiple_variable_integral} then shows that the $\ell_2$-error can be bounded as 
$
\norm{ \sigma_{\textrm{cts}} - \sigma_{\textrm{dis}} }_2 \leq \poly(n) / 2^{2N} , 
$
which is negligible by our choice of $N = \log^{O(1)} (n)$.

\paragraph{Projection maintenance and preconditioning}

Inverse maintenance was first proposed in \cite{k80} as a method for solving ``slowly-changing'' linear system.

Given a sequence of positive vectors $w^{1}, w^{2}, \cdots w^{T} \in \R_+^m$, let $P(w)$ be defined as $P(w) = \sqrt{W} A ( A^\top W A )^{-1} A^\top \sqrt{W}$, where $W \in \R^{m \times m}$ is the diagonal matrix with $W_{i,i} = w_i$.  
The goal is to output a sequence of vectors $v^{1}, v^{2}, \cdots, v^{T}$ such that $(1- \epsilon) v^t  \leq w^t \leq (1+ \epsilon) v^t$, $\forall t \in [T]$, and efficiently computes $P(v^t) u$, for
query vector $u \in \R^n$. 
The recent work of \cite{cls19} gave an efficient way to perform such a task.
For our purpose, however, rather than just computing $P(v^t) u$, we also need explicit approximations to the matrices $Q(w) = A ( A^\top W A )^{-1} A^\top$ and $M(w)^{-1} = (A^\top W A)^{-1}$. 
The matrices $Q(w)$ and $M(w)^{-1}$ appear frequently in our formula for the changes in leverage scores, and we need their approximations as pre-conditioners for accelerating the computation of certain matrix rectangular multiplication involving $Q(w)$ and $M(w)^{-1}$ (Lemma~\ref{lem:precondition}). 

\paragraph{``Tall'' JL \& fast rectangular matrix multiplication}

From the previous paragraph on discrete sampling, it suffices to compute $\gamma_{i,s,t}$ for all $i$, where 
\begin{align*}
\gamma_{i,s,t} = & ~ \sqrt{ W^{\mid} - W^{\new}} \cdot Q(y_{s,t}) \cdot (Z_t - X_t) \cdot Q(y_{s,t}) \cdot \sqrt{ W^{\new} } \cdot e_i .
\end{align*}
Notice that computing $\gamma_{i,s,t}$ for all $i$ is essentially computing the  matrix products
\begin{align*}
\sqrt{ W^{\mid} - W^{\new}} \cdot Q(y_{s,t}) \cdot (Z_t - X_t) \cdot Q(y_{s,t}) \cdot \sqrt{ W^{\new} } ,
\end{align*}
which would take $O(n^{\omega + o(1)})$ time if computed exactly.
To improve the time while ensuring keeping the error small, 
we invoke JL with dimension $n^c$ for small constant $c$ by computing
\begin{align*}
\sigma_{i,\textrm{jl}} = \sum_{t \in \mathcal{T}} \sum_{s \in \mathcal{S}} \sum_{s' \in \mathcal{S}} \omega_t \omega_s \omega_{s'} \gamma_{i,s,t}^\top R_{\gamma,s,s',t}^\top R_{\gamma,s,s',t} \gamma_{i,s', t} ,
\end{align*}
where $R_{\gamma,s,s',t} \in \mathbb{R}^{n^c \times n}$ is a random matrix. It is essential that $c$ is picked such that the rectangular matrix multiplication can be done in time roughly $n^2$.
To obtain a small error, $\sigma_{i,\textrm{jl}}$ should be a good estimate with a small variance.
We note that $\sigma_{i,\textrm{jl}}$ is indeed an unbiased estimator of $\sigma_{i,\textrm{dis}}$ and its variance can be bounded as
$
\sum_{i=1}^m \Var[ \sigma_{i,\textrm{jl}}] \leq O( \epsilon^2 / n^c ) ,
$
where $\epsilon$ is the error for the projection maintenance used by the data structure (Sections~\ref{subsec:upper_bound_eta_alpha_beta_gamma} and~\ref{subsec:variance_bound_random_gaussian_matrix}). 
Leveraging fast rectangular matrix multiplication, we pick $c = 0.31$ and $\epsilon = n^{-0.1}$.
The variance would then be bounded by $n^{-0.51}$ which is sufficiently small after $n^{\omega - 2 + o(1)}$ steps before the data structure restarts.

\subsection{Much faster rectangular matrix multiplication implies deterministic cutting plane method}

Our algorithm crucially relies on different kinds of fast rectangular matrix multiplication results. 
We also show that if these results are improved, then we are able to get a {\em deterministic} cutting plane method immediately.

\begin{corollary}[Informal version of Corollary~\ref{cor:det_cutting_plane}]\label{cor:frmm_imply_cutting_plane_method}
If $\mathcal{T}_{\mathrm{mat}}(n,n,r) = O(n^2 \log^{O(1)}(n))$ for $r = n^{\beta}$ with $\beta > 2/3$, then there is a deterministic cutting plane method which runs in time 
\begin{align*}
O(n\cdot \SO \log (\kappa)+n^{3}\log (\kappa)),
\end{align*}
where $\SO$ is the time complexity of the separation oracle.
\end{corollary}

Let $\alpha$ denote the dual exponent of matrix multiplication, which is the largest number $\alpha > 0$ such that ${\cal T}_{\mathrm{mat}}(n,n,n^{\alpha}) = n^{2+o(1)}$. Let $\beta$ denote the largest number such that ${\cal T}_{\mathrm{mat}}(n,n,n^{\beta}) = n^2 \log^{O(1)} n$. A very recent result by Christandl, Le Gall, Lysikov and Zuiddam \cite{cglz20} showed the limitations of several tensor techniques: they proved that $\alpha < 0.625$ for certain tensors. 
%showed a limitation for $\alpha$ for certain tensors ($\alpha < 0.625$). 
We believe our work initiated two interesting open questions in the area of fast matrix multiplication: (1) whether one can prove a better upper bound on $\beta$ (compared to $\alpha$) for certain tensor techniques, and (2) if there is a non-trivial inequality between $\beta$ and $\alpha$.

%In Section~\ref{sec:tech} we present a roadmap of our algorithm, which consists of three ``layers'': inner, middle and outer. We also include more detailed discussion of our techniques. 

\bigskip
\noindent{\bf Organization.} We introduce basic notations, backgrounds and tools in Section~\ref{sec:preli}. In Section~\ref{sec:vaidya}, we present the statement that Vaidya's cutting plane method  tolerates perturbed leverage scores.
We present our main data-structure for maintaining leverage scores in Section~\ref{sec:main_leverage_score}. Our main data-structure uses two different leverage score maintenance data-structures in Sections~\ref{sec:simple_leverage_score} and \ref{sec:complicated_leverage_score} with three different settings of the error parameter. 
Both our data-structures in Section~\ref{sec:simple_leverage_score} and~\ref{sec:complicated_leverage_score} reply on the batched low rank algorithm in Section~\ref{sec:batched}.

In Section~\ref{sec:perturb}, we prove that Vaidya's cutting plane method  tolerates perturbed leverage scores.   
We provide several modified versions of projection maintenance data-structures in Section~\ref{sec:projection_maintenance}. Finally, we explain how to handle convex-concave game optimization in Section~\ref{sec:saddle_point}, and present our applications to market equilibrium computations in Section~\ref{sec:application}. %%% Section 2

\section*{Acknowledgments}

The authors would like to express their sincere gratitude to matrix multiplicationer Josh Alman for his patient and answers of our exponential number of questions about fast matrix multiplication. 

The authors would like to thank Swati Padmanabhan for very useful discussions at the early stage of this project. The authors would like to thank Lijie Chen, Nikhil Devanur, Irit Dinur, Simon S. Du, Wei Hu, Jason Lee, Jerry Li, Ruoqi Shen, Aaron Schild, Aaron Sidford, Santosh Vempala, Xin Yang, Peilin Zhong, and Danyang Zhuo.
%%%%abcdefghijklmnopqrstuvwxyz

The authors would like to thank Josh Alman, Irit Dinur, Avi Wigderson, and Jeroen Zuiddam for useful discussion about limitation of fast matrix multiplication.

The authors would like to thank Sanjeev Arora, Alexandr Andoni, Ainesh Bakshi, Yangsibo Huang, Rajesh Jayaram, Ravindran Kanna, Michael Kapralov, Adam Klivans, Kai Li, Christos Papadimitriou, Eric Price, Daniel Roy, Clifford Stein, Omri Weinstein, David P. Woodruff, and Hengjie Zhang for asking interesting questions at the end of the talk of this paper.

This project was supported in part by NSF awards CCF-1749609, CCF-1740551, DMS-1839116, and Microsoft Research Faculty Fellowship.

This project was supported in part by Special Year on Optimization, Statistics, and Theoretical Machine Learning (being led by Sanjeev Arora) at Institute for Advanced Study.

\newpage
\addcontentsline{toc}{section}{References}
\bibliographystyle{alpha}
\bibliography{ref}
\newpage

\newpage

\section{Preliminaries}\label{sec:preli}

In this section we introduce the notions and tools used throughout this paper.

\subsection{Notations}
For a positive integer $n$, let $[n]$ denote the set $\{1,2,\cdots,n\}$. For any function $f$, we define $\wt{O}(f)$ to be $f \cdot \log^{O(1)}(f)$. In addition to $O(\cdot)$ notation, for two functions $f$, $g$, we use the shorthand $f \lesssim g$ (resp. $\gtrsim$) to indicate that $f \leq C \cdot g$ (resp. $\geq$) for some absolute constant $C.$

For vector $x \in \R^n$, we use $\| x \|_2$ to denote its $\ell_2$ norm, we use $\| x \|_1$ to denote its $\ell_1$ norm, we use $\| x \|_{\infty}$ to denote its $\ell_{\infty}$ norm. For matrix $A \in \R^{m \times n}$, we use $\| A \|$ to denote the spectral norm of $A$, we use $\| A \|_F$ to denote its Frobenius norm ($\| A \|_F = ( \sum_{i=1}^m \sum_{j=1}^n A_{i,j}^2 )^{1/2}$), we use $\| A \|_1$ to denotes its entry-wise $\ell_1$ norm $\| A \|_1 = \sum_{i=1}^m \sum_{j=1}^n |A_{i,j}|$. We use $A^\dagger$ to denote its Moore–Penrose inverse. We use $A^\top$ denote the transpose of $A$. 

For square matrix $A$, we use $\tr[A]$ to denote the trace of $A$. We say $A$ is positive-semidefinite (psd) if $A = A^\top$ and $x^\top A x \geq 0$, $\forall x \in \R^n$. We use $\succeq, \preceq$ to denote the semidefinite ordering, e.g. $A \succeq 0$ means that $A$ is psd.

For square full-rank matrix $A$, we use $A^{-1}$ to denote the inverse of $A$.

For vectors $a,b \in \R^n$ and accuracy parameter $\epsilon \in (0,1)$, we use $a \approx_{\epsilon} b$ to denote that $(1-\epsilon) b_i \leq a_i \leq (1+\epsilon) b_i$, $\forall i \in [n]$. Similarly, for any scalar $t$, we use $a \approx_{\epsilon} t$ to denote that $(1-\epsilon) t \leq a_i \leq (1+\epsilon) t$, $\forall i \in [n]$.

\subsection{Operators}

\begin{definition}[Operators $\tau$ and $\sigma$]\label{def:tau_sigma}
Given a matrix $A \in \R^{m \times n}$ and a vector $v \in \R^m$, we define the leverage scores $\sigma(v)_i$ and (unnormalized) leverage scores $\tau(v)_i$ as follows
\begin{align*}
\tau(v)_i =  ( A (A^\top V A)^{-1} A^\top )_{i,i} \mathrm{~and~} \sigma(v)_i = ( \sqrt{V} A ( A^\top V A )^{-1} A^\top \sqrt{V} )_{i,i} .
\end{align*}
\end{definition}

\begin{definition}[Operators $M$, $Q$ and $P$] \label{defn:operator_Q_and_P}
Given a matrix $A \in \R^{m \times n}$, we define three operators : $M : \R^m \rightarrow \R^{n \times n}$, $Q : \R^m \rightarrow \R^{m \times m}$ and $P : \R^m \rightarrow \R^{m \times m}$ such that for any vector $v \in \R^m$
\begin{align*}
M(v) = & ~  A^\top V A  \\
Q(v) = & ~ A ( A^\top V A )^{-1} A^\top \\
P(v) = & ~ \sqrt{V} A (A^\top V A)^{-1} A^\top \sqrt{V}.
\end{align*}
\end{definition}

\subsection{Different types of running time}

\begin{definition}[Pseudo-polynomial time]
An algorithm runs in {\em pseudo-polynomial} time if its running time is a polynomial in the length of the input (the number of bits required to represent it) and the numeric value of the input (the largest integer present in the input).
\end{definition}

\begin{definition}[Strongly-polynomial time]
An algorithm runs in {\em strongly polynomial} time if (1) the number of operations in the arithmetic model of computation is bounded by a polynomial in the number of rational numbers in the input instance, and (2) the space used by the algorithm is bounded by a polynomial in the size of the input.
\end{definition}

\begin{definition}[Weakly-polynomial time]
An algorithm runs in {\em weakly-polynomial} time if its running time is upper bounded by a polynomial in the size of the input, and it doesn't run in strongly polynomial time.
\end{definition}

\subsection{Basic results on matrices}

\begin{fact}[Woodbury matrix identity, \cite{w49,w50}]\label{fac:woodbury_matrix_identity}
Given a square invertible $n \times n$ matrix $A$, an $n \times k$ matrix $U$ and a $k \times n$ matrix $V$, let $B$ be an $n \times n$ matrix such that $B = A + UCV$. Then, assuming $(I_k + V A^{-1} U )$ is invertible, we have 
\begin{align*}
B^{-1} = A^{-1} - A^{-1} U ( C^{-1} + V A^{-1} U )^{-1} V A^{-1}.
\end{align*}
\end{fact}

\subsection{Fast matrix multiplication}

\begin{definition}[Matrix multiplication time]
For any $n, r > 0$, denote $\T_{\mat}(n,n,r)$ the time to compute the multiplication of an $n \times n$ matrix and an $n \times r$ matrix. 
\end{definition}

\begin{theorem}[Fast matrix multiplication] \label{thm:fast_matrx_multiplication}
We have the following upper bound on $\T_{\mat}(n,n,r)$:
\begin{enumerate}
	\item~\cite{gu18}. For $r = n^{0.31}$, we have $\T_{\mat}(n,n,r) = O(n^{2 + o(1)})$.
	\item~\cite{c82}. For $r = n^{0.17}$, we have $\T_{\mat}(n,n,r) = O( n^2 \log^2 n )$.
	\item~\cite{bd76}. For $r = \log^c n$ for any constant $c > 0$, then we have $\T_{\mat}(n,n,r) = O(n^2)$. 
\end{enumerate}
\end{theorem}

\begin{comment}
\begin{theorem}[\cite{gu18}]
For $r = n^{0.31}$, we have $T(n,n,r) = O(n^{2 + o(1)})$. 
%\url{https://arxiv.org/abs/1708.05622} and references therein.
\end{theorem}

\begin{theorem}[\cite{c82}]
For $r = n^{0.17}$, we have $T(n,n,r) = O(n^2 \log^2 n)$.
%\url{http://theory.stanford.edu/~virgi/cs367/papers/copper-rect1.pdf}
\end{theorem}

\begin{theorem}[\cite{bd76}] \label{thm:fast_matrx_multiplication_polylog_n}
For $r = \log^c n$ for any constant $c > 0$, then we have $T(n,n,r) = O(n^2)$. 
%\url{https://epubs.siam.org/doi/abs/10.1137/0205041}.
\end{theorem}
\end{comment}
\subsection{Multiple variable polynomial interpolation} 

We first state a one-variable version interpolation theorem.
\begin{theorem}[One variable {\cite[Eq 10 in Page 331]{isaacson1994analysis}}]\label{thm:one_variable_integral}
Given any $2N$ times differentiable function $f : \R \rightarrow \R$. There exits $N$ points $s_1,\cdots, s_N$ and $N$ weights $\omega_1,\cdots,\omega_N \geq 0$ such that $\sum_i \omega_i = 1$ and
\begin{align*}
\left| \int_0^1 f(t) \mathrm{d} t - \sum_{i=1}^N \omega_i f(s_i) \right| \leq 
O(1) \cdot \frac{M_{2N}}{(2N)! 4^N}
\end{align*}
where 
\begin{align*}
M_{2N} = \max_{t \in [0,1] } \left| \frac{ \partial^{(2N)} f } { \partial t^{ ( 2N ) } } ( t ) \right|.
\end{align*}
\end{theorem}

%In the most of paper, we only need Theorem~\ref{thm:one_variable_integral} which hides the details of definitions of weights and points. Here we state a complete version of theorem which contains all the details.
%\begin{theorem}[One variable, detailed version]\label{thm:one_variable_integral_detailed}
%Given any $N+1$ times differentiable function $f : \R \rightarrow \R$. We have
%\begin{align*}
%\left| \int_0^1 f(t) \mathrm{d} t - \sum_{i=0}^N \omega_i f(s_i) \right| \leq \frac{ M_{N+1} }{ (N+1) ! } \cdot \max_{x \in [0,1] } \left| \prod_{i=0}^N (x - s_i) \right|
%\end{align*}
%where 
%\begin{align*}
%M_{N+1} = \max_{ t \in [0,1] } \left| \frac{ \partial^{(N+1)} f } { \partial t^{ ( N + 1 ) } } ( t ) \right|.
%\end{align*}
%In addition, $s_i$ and $\omega_i$ are defined as follows,
%\begin{align*}
%s_i = \frac{1}{2} + \frac{1}{2} \cos ( \frac{ 2i + 1 }{ 2N + 2 } \pi ) , \omega_i = \int_0^1 l_i(x) \mathrm{d} x , l_i(x) = \prod_{0 \leq j \leq N, j \neq i } \frac{ x - s_j }{ s_i - s_j } .
%\end{align*}
%\end{theorem}

Now, we explain how to use one variable interpolation result (Theorem~\ref{thm:one_variable_integral}) to prove a multiple variable interpolation result (Theorem~\ref{thm:multiple_variable_integral}).
\begin{theorem}[$d$-variable case]\label{thm:multiple_variable_integral}
Given a function $f : \R^d \rightarrow \R$ such that for any $j \in [d]$, for any $t_1, \cdots, t_{j-1}, t_{j+1}, \cdots t_d \in [0,1]$, $f(t_1, \cdots, t_{j-1}, t_j, t_{j+1}, \cdots, t_d)$ is $N+1$ times differentiable respect to $t_j$. Then there exists $N$ points $s_1, \cdots, s_N$ and $N$ weights $\omega_1, \cdots, \omega_N$ such that 
\begin{align*}
\left| \int_{[0,1]^d} f(t) \mathrm{d} t - \sum_{i_1 = 1}^N \cdots \sum_{i_d = 1}^N \omega_{i_1} \cdots \omega_{i_d} f( s_{i_1}, \cdots, s_{i_d} ) \right| \leq \frac{1}{(2N)! 2^{2N}} \cdot  \sum_{j=1}^d \max_{t \in [0,1]^d} \left| \frac{ \partial^{(2N)} f }{ \partial t_j^{(2N)} } (t) \right| .
\end{align*}
\end{theorem}

\begin{proof}
Let $\tau$ denote 
\begin{align*}
\tau = \frac{1}{(2N)! 2^{2N}}
\end{align*}

Using one variable theorem~\ref{thm:one_variable_integral}, we have that : for all $t_1, t_2, \cdots, t_{d-1} \in [0,1]$,
\begin{align*}
\left| \int_0^1 f(t_{[d-1]}, t_d) \mathrm{d} t_d - \sum_{i_d=1}^N \omega_{i_d} f( t_{[d-1]} , s_{i_d} ) \right| \leq  \tau \cdot \max_{ t_{[d-1]} \in [0,1]^{d-1} } \max_{t_d \in [0,1]} | f_{t_d}^{(2N)} (t) |
\end{align*}

To write recursive thing in an easy way, we let $g_d$ denote function $f$. We use $t_{[j]}$ to denote $t_1, t_2, \cdots, t_j$.

For each $j \in \{ d-1, \cdots,1\}$, we define function $g_j : \R^j \rightarrow \R$ such that
\begin{align*} 
g_j (t_{[j]})  = \sum_{ i_{j+1} = 1}^N \omega_{ i_{j+1} } g_{j+1} ( t_{[j]} , s_{i_{j+1}} )
\end{align*}
Let $g_0 = \sum_{i_1 = 1 }^N \omega_{ i_1 } g_1 ( s_{i_1} )$, we can also rewrite $g_0$ as $\int_{ [0,1]^0 } g_0 ( t_{[0]} ) \mathrm{d} t_{[0]}$.

Finally, we want to bound
\begin{align*}
 & ~ \left| \int_0^1 \cdots \int_0^1 f(t_1,\cdots,t_d) \mathrm{d} t_1 \cdots \mathrm{d} t_d - \sum_{i_1 = 0}^N \cdots \sum_{i_d = 1}^N \omega_{s_i} f(s_{i_1}, \cdots, s_{i_d} ) \right| \\
\leq & ~ \sum_{j=1}^d \left| \int_{ [0,1]^j } g_j( t_{[j]} ) \mathrm{d} t_{[j]} - \int_{[0,1]^{j-1}} g_{j-1} ( t_{[j-1]} ) \mathrm{d} t_{[j-1]} \right| \\
\leq & ~ \tau \sum_{j=1}^d \max_{ t \in [0,1]^d } \left| \frac{ \partial^{(2N)} f}{ \partial t_j^{(2N)} (t) }  \right|,
\end{align*}
where the last step follows by Claim~\ref{cla:bound_difference_between_j_and_j-1}.
\end{proof}

\begin{claim}[Bounding the difference between $j$-th term and $j-1$-th term]\label{cla:bound_difference_between_j_and_j-1}
For each $j \in [d]$, 
\begin{align*}
\left| \int_{ [0,1]^j } g_j ( t_{ [j] } ) \mathrm{d} t_{[j]} - \int_{ [0,1]^{j-1} } g_{j-1} ( t_{ [j-1] } ) \mathrm{d} t_{ [j-1] } \right| \leq \tau \cdot \max_{ t \in [0,1]^d } \left| \frac{ \partial^{(2N)} f}{ \partial t_j^{(2N)} (t) } \right|.
\end{align*}
\end{claim}
\begin{proof}
First, using one variable theorem~\ref{thm:one_variable_integral}, we can upper bound
\begin{align*}
 & ~ \left| g_j ( t_{[j-1]}, t_j ) \mathrm{d} t_j - \sum_{i_j=1}^N \omega_{i_j} g_j ( t_{[j-1]}, s_{i_j} ) \right| \\
\leq & ~ \tau \cdot \max_{t_{[j-1]} \in [0,1]^{j-1} } \max_{t_j \in [0,1]} | g_{j,t_j}^{(N+1)} ( t_{[j]} ) | \\
\leq & ~ \cdots \\
\leq & ~ \tau \cdot \max_{t_{[j-1]} \in [0,1]^{j-1} } \max_{t_j \in [0,1]} \max_{t_{[d]\backslash [j]} \in [0,1]^{d-j}} |f^{(N+1)}_{t_j} ( t_{[j-1]}, t_j, t_{ [d] \backslash [j] } ) |
\end{align*}
where the last step follows by re-using $\omega_i >0 $ and $\sum_{i = 1}^N \omega_i = 1$.

We have
\begin{align*}
& ~ \left| \int_{ [0,1]^j } g_j ( t_{ [j] } ) \mathrm{d} t_{[j]} - \int_{ [0,1]^{j-1} } g_{j-1} ( t_{ [j-1] } ) \mathrm{d} t_{ [j-1] } \right| \\
= & ~ \left| \int_{ [0,1]^j } g_j ( t_{ [j] } ) \mathrm{d} t_{[j]} - \int_{ [0,1]^{j-1} } \sum_{ i_j = 1 }^N \omega_{i_j} g_{j} (t_{ [j-1] }, s_{i_j}) \mathrm{d} t_{[j-1]} \right| \\
= & ~ \left| \int_{ [0,1]^{j-1} } \left( g_j ( t_{[j-1]}, t_j ) \mathrm{d} t_j - \sum_{i_j=1}^N \omega_{i_j} g_j ( t_{[j-1]}, s_{i_j} )  \right) \mathrm{d} t_{ [j-1] } \right| \\
\leq & ~ \int_{ [0,1]^{j-1} } \left| g_j ( t_{[j-1]}, t_j ) \mathrm{d} t_j - \sum_{i_j=1}^N \omega_{i_j} g_j ( t_{[j-1]}, s_{i_j} ) \right| \mathrm{d} t_{ [j-1] } \\
\leq & ~ \tau \cdot \max_{t_{[j-1]} \in [0,1]^{j-1} } \max_{t_j \in [0,1]} \max_{t_{[d]\backslash [j]} \in [0,1]^{d-j}} |f^{(2N)}_{t_j} ( t_{[j-1]}, t_j, t_{ [d] \backslash [j] } ) | \\
= & ~ \tau \cdot \max_{ t \in [0,1]^d } \left| \frac{ \partial^{(2N)} f}{ \partial t_j^{(2N)} (t) }  \right|.
\end{align*}
\end{proof}

\newpage 

\section{Perturbed Volumetric Center Cutting Plane Method}\label{sec:vaidya}

Recall that the feasibility problem is defined as follows.

\begin{itemize}
\item[] \textbf{Feasibility Problem}: Given a separation oracle for a set
$K$ contained in a box of radius $R$ either find a point $x\in K$
or prove that $K$ does not contain a ball of radius $\epsilon$.
\end{itemize}
Let $K\subset\mathbb{R}^{n}$ be a convex set that is contained in
a box of radius $R$ and equipped with a separation oracle. To solve the feasibility problem, a cutting plane method iteratively
refines a feasible region $\Omega$ and a query point $z$. In each
iteration, $z$ is the input to the separation oracle which
then either certifies $z\in K$ or produces a hyperplane $a^{\top}x = b$
separating $z$ from $K$. The algorithm then refines $\Omega$ and
computes the next query point $z^{\text{(new)}}$ from the new separating hyperplane and any past information.

At a high level, Vaidya's method heavily employs \emph{leverage scores} to decide which hyperplane to add or to drop from the feasible region. The main innovation behind our faster implementation is a data structure that maintains an estimate
of the leverage scores in amortized $O(n^{2})$. This removes the bottleneck in Vaidya's algorithm and yields the following result.

\begin{restatable}{theorem}{mainvaidya} \label{thm:mainvaidya}
Given a separation oracle for a convex set $K \subset \mathbb{R}^n$ that is contained in a box of radius $R$ and a parameter $\epsilon > 0$, 
there is a cutting plane method that either computes a point in $K$ or proves that $K$ does not contain a ball of radius $\epsilon$ in $O((n\SO+n^{3})\log(\kappa))$ time, where $\SO$ is the complexity of the separation oracle and $\kappa = nR / \epsilon$. 

%Assume the leverage scores $\sigma$ are replaced by estimates $\tilde{\sigma}$ in Vaidya's method, where $\|\tilde{\sigma}-\sigma\|_{2}\leq1/\log^{O(1)}n$. Then we can implement Vaidya's cutting plane method in $O((n\SO+n^{3})\log(\kappa))$ time, where $\SO$ is the complexity of the separation oracle.
\end{restatable}

Since the proof of the validity of Vaidya's method for our result is mostly a perturbed version of his analysis, we defer the proof of Theorem~\ref{thm:mainvaidya} to Section~\ref{sec:perturb}.
 %%% Section 4
\newpage
\section{Main Data Structure for Leverage Score Maintenance}\label{sec:main_leverage_score}

In this section, we present our main data structure for leverage score maintenance that achieves an amortized $O(n^2)$ time per update. 
Combined with Theorem~\ref{thm:mainvaidya} in Section~\ref{sec:vaidya}, our main data structure implies a faster $O(n\SO\log(\kappa)+n^{3}\log(\kappa))$ time cutting plane method. 
Our main data structure further uses the simple leverage score maintenance data structure in Section~\ref{sec:simple_leverage_score} and the complicated leverage score maintenance data structure in Section~\ref{sec:complicated_leverage_score}.

\begin{theorem}[Leverage score maintenance] \label{thm:leverage_score_maintain_main}
Given an initial matrix $A \in \R^{m \times n}$ with $m = O(n)$, initial weight $w \in \R_+^m$.
Then for any constant $c > 0$, there is a randomized data structure (Algorithm~\ref{alg:maintain_leverage_score_main}) that approximately maintains the  leverage scores
\begin{align*}
\sigma_i(w) = (\sqrt{W} A (A^{\top} W A)^{-1} A^{\top} \sqrt{W})_{i,i}, 
\end{align*}
for positive diagonal matrices $W = \diag(w)$ through the following operations:
\begin{enumerate}
	\item $\textsc{Init}(A,w)$: takes $O(n^{\omega + o(1)})$ time to intialize the data structure.
	\item $\textsc{Update}(\mathsf{act})$: updates the data structure for the single update $\mathsf{act}$. 
	\item $\textsc{Query}( )$: takes $O(n)$ time to output a vector $\wt{\sigma} \in \R^m$.
\end{enumerate}
Moreover, if the sequence of $K$ updates $(A^{(k)} \in \R^{m^{(k)} \times n}, w^{(k)} \in \R_+^{m^{(k)}} )$ with $m^{(k)} = O(n)$ satisfies that each update $k$ is one of the following:
\begin{enumerate}
	\item Insertion (resp. deletion) of row $a$ with weight $w_a$ into (resp. from) $(A^{(k-1)}, w^{(k-1)})$ that satisfies
	\begin{align*}
	w_a a a^\top \preceq 0.01 (A^{(k-1)})^\top W^{(k-1)} A^{(k-1)} .
	\end{align*}
	\item Update $w^{(k-1)} \in \R_+^{m^{(k-1)}}$ to $w^{(k)} \in \R_+^{m^{(k)}}$ such that $m^{(k-1)} = m^{(k)}$ and that
	\begin{align*}
	\norm{\log (w^{(k)}) - \log (w^{(k-1)})}_2 \leq 0.01 ,
	\end{align*}
\end{enumerate}
%and the number of insertions/deletions per $w$ update is $O(1)$.
then the function $\textsc{Update}(\mathsf{act})$ takes an amortized $O(n^2)$ time, and the vector $\wt{\sigma}^{(k)}$ output by $\textsc{Query}( )$ at each step $k \in [K]$ satisfies that
\begin{align} \label{eqn:l2_error_guarantee_main}
	 \norm{\wt{\sigma}^{(k)} - \sigma(w^{(k)})}_2 \leq O(1/\log^c(n)) .
\end{align}

\end{theorem}

The proof of Theorem~\ref{thm:leverage_score_maintain_main} relies on the data structures in Sections~\ref{sec:batched},~\ref{sec:simple_leverage_score} and~\ref{sec:complicated_leverage_score}.

\begin{proof}[Proof of Theorem~\ref{thm:leverage_score_maintain_main}]
We first prove the running time upper bound of the different functions. 
Notice that the running time upper bound for $\textsc{Init}(A,w)$, and $\textsc{Query}( )$ are immediate corollaries of Theorem~\ref{thm:maintain_leverage_score_simple} and~\ref{thm:maintain_leverage_score_complicated}. 
We prove running time upper bound for the function $\textsc{Update}(\mathsf{act})$ in the following. 
In particular, we show that the inner, middle and outer phases all run in amortized $O(n^2)$ time, and the amortized time to restart the data structure in Step~\ref{step:restart_data_structure_main} is $O(n^2)$.

We start by analyzing the running time of the inner phase. 
Since each call to $\textsc{Update}(\mathsf{act})$ makes one call of $\textsc{Update}(\mathsf{act})$ with $\mathsf{act}$ having a single action $\mathsf{act}$, it follows from Theorem~\ref{thm:maintain_leverage_score_simple} that the inner phase makes one call to $\textsc{pm}.\textsc{update}$ and takes extra time at most $O(n^2)$.
Since we choose our parameter as $\epsilon_{\inn} = 1/\log^{25}(n)$ (see Table~\ref{table:table_of_constants_main}), it then follows from Theorem~\ref{thm:maintain_projection_inn} with $C = O(1)$ that the amortized time per call to $\textsc{pm}.\textsc{update}$ is $O(n^2)$. 
Therefore, the inner phase runs in amortized time $O(n^2)$. 

Next we analyze the running time of the middle phase.
Notice that each middle update happens once every $T_{\inn} = \log^{10}(n) \leq n^{0.08}$ calls to $\textsc{Update}(\mathsf{act})$, so Assumption~\ref{assumption:short_sequence_small_update} holds for the middle update in Step~\ref{step:middle_update_main} of Algorithm~\ref{alg:maintain_leverage_score_main}. 
It then follows from Theorem~\ref{thm:maintain_leverage_score_simple} that for every $T_{\inn}$ calls to the function $\textsc{Update}(\mathsf{act})$, the middle phase makes one call to $\textsc{pm}.\textsc{update}$ and takes extra time at most $O(n^2 T_{\inn})$. 
Since we choose $\epsilon_{\mid} = 1/n^{0.08}$, it follows from Theorem~\ref{thm:maintain_projection_mid} with $C = T_{\inn} = \log^{10}(n)$ that the time for one call to $\textsc{pm}.\textsc{update}$ is $O(n^2 \log^2(n))$.
Amortized over the $T_{\inn} = \log^{10}(n)$ calls to the function $\textsc{Update}(\mathsf{act})$, the amortized running time for the middle phase is thus $O(n^2)$.

Finally we analyze the running time of the outer phase. 
Each outer update happens once every $T_{\inn} \cdot T_{\mid} = n^{0.01}\log^{10}(n)$  calls to $\textsc{Update}(\mathsf{act})$, so Assumption~\ref{assumption:short_sequence_small_update} holds for the outer update in Step~\ref{step:outer_update_main} of Algorithm~\ref{alg:maintain_leverage_score_main}. 
It then follows from Theorem~\ref{thm:maintain_leverage_score_complicated} that for every $T_{\inn} \cdot T_{\mid}$ calls to the function $\textsc{Update}(\mathsf{act})$, the outer phase makes $O(N) = O(\log^2 (n))$ calls to $\textsc{pm}.\textsc{update}$ and takes extra time at most 
\begin{align*}
O(n^2 \cdot T_{\inn} \cdot T_{\mid} + T(n,n,r_{\out}) \cdot N^3 \cdot \log(n)) 
= O(n^2 \cdot T_{\inn} \cdot T_{\mid} + n^{2 + o(1)} \log^7 (n)) .
\end{align*}
Since we choose $\epsilon_{\out} = 1/n^{0.1}$, it follows from Theorem~\ref{thm:maintain_projection_mid} with $C = T_{\inn} \cdot T_{\mid}$ that the time for $O(N)$ calls to $\textsc{pm}.\textsc{update}$ is $O(n^{2 + o(1)})$.  
Amortized over the $T_{\inn} \cdot T_{\mid} = n^{0.01}\log^{10}(n)$ calls to the function $\textsc{Update}(\mathsf{act})$, the amortized running time for the outer phase is again $O(n^2)$.

It remains to upper bound the amortized time to restart the data structure due to Step~\ref{step:restart_data_structure_main}.
We note that restarting takes time $O(n^{\omega + o(1)})$ every $T_{\inn} \cdot T_{\mid} \cdot T_{\out} = \Omega(n^{\omega - 1.99})$ calls to the function $\textsc{Update}(\mathsf{act})$. 
It follows that the amortized time to restart the data structure is $O(n^{\omega + o(1)}) / (T_{\inn} \cdot T_{\mid} \cdot T_{\out}) = O(n^2)$. 
This finishes the running time analysis.

Now we proceed to prove the error guarantee given in the second part of the theorem.
We only prove~(\ref{eqn:l2_error_guarantee_main}) for $c = 15$ as it's easy to adjust the parameter settings in Table~\ref{table:table_of_constants_main} to achieve the guarantee for any constant $c > 0$.
By Theorem~\ref{thm:maintain_leverage_score_simple}, the $\ell_2$-error of the inner phase is $O(\epsilon_{\inn}) = O(\log^{-25}(n))$ per call to $\textsc{Update}(\mathsf{act})$ and the inner phase runs for at most $T_{\inn} = \log^{10}(n)$ steps before its estimate gets refined by the estimate of the middle or outer phases (Step~\ref{step:refine_estimate_middle_to_inner_main} and~\ref{step:refine_estimate_outer_to_inner_main}). 
Therefore, the error introduced by the inner phase is at most $O(\epsilon_{\inn}) \cdot T_{\inn} = O(\log^{-15}(n))$. 
For the middle phase, it follows from Theorem~\ref{thm:maintain_leverage_score_simple} that the $\ell_2$-error introduced by each middle step is at most $O( (\epsilon_{\mid} + T_{\inn} \cdot n^{-0.08}) \cdot T_{\inn} ) \leq n^{-0.07}$.
Since the middle phase run for at most $T_{\mid} = n^{0.01}$ steps before its estimate gets refined by the estimate of the outer phase (Step~\ref{step:refine_estimate_outer_to_middle_main}), 
the error of the middle phase is at most $n^{-0.07} \cdot n^{0.01} = n^{-0.06}$.
For the outer phase, the error in Theorem~\ref{thm:maintain_leverage_score_complicated} is negligible so it suffices to bound the variance.
From Theorem~\ref{thm:maintain_leverage_score_complicated}, the variance in each outer step is $O(N^3/r_{\out}) \cdot \epsilon_{\out}^2 \cdot (T_{\inn} \cdot T_{\mid})^2 \leq O(n^{-0.48})$. 
Since the outer phase run for at most $T_{\out} = n^{\omega - 2}$ steps before the data structure gets restarted, the variance accumulated in the $T_{\out}$ steps is at most $O(n^{-0.1})$. 
Therefore, we have $\norm{\wt{\sigma}^{(k)} - \sigma(w^{(k)})}_2 \leq O(\log^{-15}(n))$ for each $k \in [K]$. 
This completes the proof of Theorem~\ref{thm:leverage_score_maintain_main}.

\end{proof}

\begin{table}[htp!]  
\centering
\begin{tabular}{| l | l | l |}
	\hline
	{\bf Parameter} & {\bf Value} & {\bf Description} \\ \hline
	$T_{\inn}$ & $\log^{10} n$   & Number of inner steps per middle step \\ \hline
	$T_{\mid}$ & $n^{0.01}$       & Number of middle step per outer step \\ \hline
	$T_{\out}$ & $n^{\omega - 2}$ & Number of outer step before re-computing\\ \hline
%	$r_{\inn}$ & $\log^{50} n$  & Matrix multiplication dimension for inner step \\ \hline
%	$r_{\mid}$ & $n^{0.17}$       & Matrix multiplication dimension for middle step  \\ \hline
	$r_{\out}$ & $n^{0.31}$       & Matrix multiplication dimension for outer step  \\ \hline
	$\epsilon_{\inn}$ & $1/\log^{25} n$ & Inner step projection maintenance error \\ \hline
	$\epsilon_{\mid}$ & $1/n^{0.08}$ & Middle step projection maintenance error \\ \hline
	$\epsilon_{\out}$ & $1/n^{0.1}$ &  Outer step projection maintenance error \\ \hline
	$N$   & $100 \log^2 n$ & Number of discrete sampling points in an outer step \\ \hline
\end{tabular}
\caption{Values of parameters used in the main data structure for leverage score maintenance in Algorithm~\ref{alg:maintain_leverage_score_main}.}\label{table:table_of_constants_main}
\end{table}

\begin{algorithm}[htp!]\caption{}\label{alg:maintain_leverage_score_main}
\begin{algorithmic}[1]
\State {\bf data structure} \textsc{MaintainLeverageScoreMain} \Comment{Theorem~\ref{thm:leverage_score_maintain_main}}
\State
\State {\bf members}
%\State \hspace{4mm} $A \in \R^{m \times n}$  \Comment{$m = \Theta(n)$}
\State \hspace{4mm} $\ctr_{\inn}, \ctr_{\mid},\ctr_{\out} \in \mathbb{Z}_+$ \Comment{Step counters for different phases} 
\State \hspace{4mm} $\mathsf{acts}_{\mid}, \mathsf{acts}_{\out} \in \textsc{Vector} \langle \textsc{Action} \rangle$  \Comment{Delayed sequence of updates}
\State \hspace{4mm} \Comment{\textsc{Action} includes insertion, deletion and $w$-update}
\State \hspace{4mm} \textsc{MaintainLeverageScoreSimple} $\textsc{simp}_{\inn}$, $\textsc{simp}_{\mid}$ \Comment{For inner and middle phases, Theorem~\ref{thm:maintain_leverage_score_simple}, Algorithm~\ref{alg:maintain_leverage_score_simple}}
\State \hspace{4mm} \textsc{MaintainLeverageScoreComplicated} $\textsc{comp}_{\out}$ \Comment{For outer phase, Theorem~\ref{thm:maintain_leverage_score_complicated}, Algorithm~\ref{alg:maintain_leverage_score_complicated}}
\State {\bf end members}
\State
\Procedure{\textsc{Init}}{$A \in \R^{m \times n}, w \in \R_+^m$}  \Comment{Initialization}
	\State $\ctr_{\inn}, \ctr_{\mid},\ctr_{\out} \leftarrow 0$ \Comment{Reset counters}
	\State $\mathsf{acts}_{\mid}, \mathsf{acts}_{\out} \leftarrow \textsc{NULL}$ 
	\State $\textsc{simp}_{\inn}$.\textsc{Init}($A,w,\epsilon_{\inn}$)
	\State $\textsc{simp}_{\mid}$.\textsc{Init}($A,w,\epsilon_{\mid}$)
	\State $\textsc{comp}_{\out}$.\textsc{Init}($A,w,r_{\out}, N ,\epsilon_{\out}$) \Comment{Initialize member data structures}
\EndProcedure
\State
\Procedure{\textsc{Update}}{$\mathsf{act} \in \textsc{Action}$}  \Comment{Update estimate for a single update}
	\State Add $\mathsf{act}$ to $\mathsf{acts}_{\mid}$ and $\mathsf{acts}_{\out}$
	\State $\textsc{simp}_{\inn}$.\textsc{Update}($\mathsf{act}$), $\ctr_{\inn} \leftarrow \ctr_{\inn} + 1$   \Comment{Inner update}
	\If {$\ctr_{\inn} = T_{\inn}$}  \Comment{Number of inner steps reaches limit}
		\State $\textsc{simp}_{\mid}$.\textsc{Update}($\mathsf{acts}_{\mid}$), $\ctr_{\mid} \leftarrow \ctr_{\mid} + 1$  \Comment{Middle update}
		\label{step:middle_update_main}
		\State $ \mathsf{acts}_{\mid} \leftarrow \textsc{NULL}$
		\If {$\ctr_{\mid} = T_{\mid}$}  \Comment{Number of middle steps reaches limit}
			\State $\textsc{comp}_{\out}$.\textsc{Update}($\mathsf{acts}_{\out}$), $\ctr_{\out} \leftarrow \ctr_{\out} + 1$ 
			\Comment{Outer update}
			\label{step:outer_update_main}
			\State $ \mathsf{acts}_{\out} \leftarrow \textsc{NULL}$
			\If {$\ctr_{\out} = T_{\out}$} \Comment{Number of outer steps reaches limit}
				\State \textsc{Self}.\textsc{Init}($\textsc{simp}_{\inn}.A,w^{\new}$)   \Comment{Restart the data structure}
				\label{step:restart_data_structure_main}
				\State \textbf{end procedure}
			\EndIf
			\State $\ctr_{\inn} \leftarrow 0, \ctr_{\mid} \leftarrow 0$ \Comment{Restart counters for inner/middle phases}
			\State $\textsc{simp}_{\inn}.\textsc{RefineEstimate}(\textsc{comp}_{\out}.\textsc{Query}())$ 
			\label{step:refine_estimate_outer_to_inner_main}
			\State $\textsc{simp}_{\mid}.\textsc{RefineEstimate}(\textsc{comp}_{\out}.\textsc{Query}())$ \Comment{Refine $\wt{\sigma}_{\inn}$ and $\wt{\sigma}_{\mid}$}
			\label{step:refine_estimate_outer_to_middle_main}
		\Else
			\State $\ctr_{\inn} \leftarrow 0$ \Comment{Restart counter for inner phase}
			\State $\textsc{simp}_{\inn}.\textsc{RefineEstimate}(\textsc{simp}_{\mid}.\textsc{Query}())$  \Comment{Refine $\wt{\sigma}_{\inn}$}
			\label{step:refine_estimate_middle_to_inner_main}
		\EndIf
	\EndIf	
\EndProcedure
\State
\Procedure{\textsc{Query}}{$ $}  
	\State \Return $\textsc{simp}_{\inn}$.\textsc{Query}() \Comment{Return the estimate maintained by $\textsc{simp}_{\inn}$}
\EndProcedure
\State
\State {\bf end data structure}
\end{algorithmic}
\end{algorithm} %%% Section 5
\newpage 
\section{Batched Low Rank Update}\label{sec:batched}

\subsection{Simple low rank update}

In this section, we show how to maintain leverage score under low
rank update. First, we start with a preconditioning lemma showing
that given two spectrally similar matrix, we can invert one matrix faster 
by knowing the inverse of the other matrix.
\begin{lemma}
\label{lem:precondition}Given $n\times n$ PSD matrices $A$ and
$M$ such that $0\preceq M\preceq A\preceq\kappa\cdot M$. For any integer $t \geq 1$ such that $\kappa(1 - 1/ \kappa)^{t + 1}<1$, we have
\begin{align*}
f(M,t) \preceq A^{-1}\preceq\frac{1}{1 - \kappa( 1 - 1 / \kappa )^{t+1}} \cdot f(M,t)
\end{align*}
where $f : \R^{n\times n} \times \mathbb{N} \rightarrow \R^{n \times n}$ is defined as
\begin{align*}
f(M,t) = \frac{1}{\kappa} M^{-1} \sum_{i=0}^{t} (I-\frac{1}{\kappa}AM^{-1})^{i}.
\end{align*}
Furthermore, if $M^{-1}$ and $A$ are given explicitly, for any $V\in\R^{n\times r}$,
then we can compute $f(M,t) \cdot V$ in $O(\T_{\mat}(n,n,r) \cdot t )$ time.
\end{lemma}

\begin{proof}
Note that
\begin{align}\label{eq:A_taylor}
A^{-1} 
 = & ~ M^{-\frac{1}{2}}(M^{-\frac{1}{2}}AM^{-\frac{1}{2}})^{-1}M^{-\frac{1}{2}}\notag \\
 = & ~ \frac{1}{\kappa}M^{-\frac{1}{2}}(I-(I-\frac{1}{\kappa}M^{-\frac{1}{2}}AM^{-\frac{1}{2}}))^{-1}M^{-\frac{1}{2}}.
\end{align}
Using
\begin{align*}
0\preceq I- (1/\kappa) \cdot M^{-\frac{1}{2}}AM^{-\frac{1}{2}} \preceq ( 1 - 1 / \kappa ) \cdot I
\end{align*}
and 
\begin{align*}
\sum_{i=0}^{t} x^{i} \leq(1-x)^{-1} \leq \sum_{i=0}^{t} x^{i} + \kappa ( 1 - 1 / \kappa )^{t+1} ~~~~ \forall x \in [0, 1 - 1 / \kappa ],
\end{align*}
we have
\begin{align}
\sum_{i=0}^{t} \Big( I-\frac{1}{\kappa}M^{-\frac{1}{2}}AM^{-\frac{1}{2}} \Big)^{i} 
& \preceq \Big( I-(I-\frac{1}{\kappa}M^{-\frac{1}{2}}AM^{-\frac{1}{2}}) \Big)^{-1} \nonumber \\
& \preceq \sum_{i=0}^{N} \Big( I-\frac{1}{\kappa}M^{-\frac{1}{2}}AM^{-\frac{1}{2}} \Big)^{i} + \kappa \cdot ( 1 - 1 / \kappa )^{t+1} I \label{eq:MAM_taylor}
\end{align}
Multiplying $\frac{1}{\kappa} M^{-1}$ on both sides of the above equation with a formula $A^{-1}$ (Eq.~\eqref{eq:A_taylor}), we get
\begin{align*}
\frac{1}{\kappa} M^{-1} \sum_{i=0}^{t} \Big( I-\frac{1}{\kappa}M^{-\frac{1}{2}}AM^{-\frac{1}{2}} \Big)^{i} \preceq A^{-1} \preceq \frac{1}{\kappa} M^{-1} \sum_{i=0}^{t} \Big( I-\frac{1}{\kappa}M^{-\frac{1}{2}}AM^{-\frac{1}{2}} \Big)^{i} +  ( 1 - 1 / \kappa )^{t+1} M^{-1}
\end{align*}

Using the definition of $f(M,t)$, we have that
\begin{align*}
f(M,t)\preceq A^{-1} \preceq f(M,t)+( 1 - 1 / \kappa )^{t+1} M^{-1}.
\end{align*}
Using $M^{-1} \preceq \kappa A^{-1}$, we have
\begin{align*}
f(M,t)\preceq A^{-1} \preceq f(M,t)+\kappa \cdot ( 1 - 1 / \kappa )^{t+1} A^{-1}.
\end{align*}
The result follows from some rearranging.
\end{proof}
Now, we show how to do leverage score update under monotone updates.
\begin{lemma}\label{lem:simple_low_rank}
Given a matrix $A\in\R^{m\times n}$,
non-negative weights $w,w^{\new}\in\R^{m}$ such that $w^{\new}\geq w$
with $m=O(n)$ and $\| w^{\new} - w \|_0 = k$. Let $\beta > 1$ denote the parameter such that $A^{\top} W^{\new} A \preceq \beta \cdot A^{\top} W A$. Given some explicit matrix $U \in \R^{n \times n}$ and $U^{\new} \in \R^{n \times n}$ such
that 
\begin{align*}
U^{-1} \preceq & ~ A^{\top}WA \preceq ( 1 + 1/\log n ) U^{-1} , \\
(U^{\new})^{-1}\preceq & ~ A^{\top}W^{\new}A\preceq ( 1 + 1 / \log n ) (U^{\new})^{-1}.
\end{align*}
 Then, for any $\epsilon \in (0,1/2)$, we can output a vector $c \in \R^m$ such that
\begin{align*}
\|c-(\sigma(w^{\new})-\sigma(w))\|_{2}\leq\epsilon
\end{align*}
in time 
\begin{align*}
O( \T_{\mat} (n,n,k) \cdot (1+\log_{\log (n)}( \beta k / \epsilon ) ) ).
\end{align*}

The same statement holds for decreasing $w$, namely $w^{\new}\leq w$, with
$A^{\top}W^{\new}A\succeq\beta^{-1}A^{\top}WA$.
\end{lemma}

\begin{proof}
We note that the decreasing case follows from the increasing case
by swapping $w$ and $w^{\new}$. Hence, we focus on the increasing
case. Using Lemma \ref{lem:precondition}, we can construct $U_{\epsilon}$
such that $U_{\epsilon}$ is low-degree polynomial of $U$ and $A^{\top}WA$
and that
\begin{equation}
U_{\epsilon}\preceq(A^{\top}WA)^{-1}\preceq(1+\epsilon)U_{\epsilon}.\label{eq:U_def}
\end{equation}
We define $U_{\epsilon}^{\new}$ similarly. We will show how to use
$U_{\epsilon}$ and $U_{\epsilon}^{\new}$ to approximate $\sigma(w^{\new})-\sigma(w)$.
First, we note that
\begin{align*}
\sigma(w^{\new})-\sigma(w)=(w^{\new}-w)\cdot\tau(w^{\new})+w\cdot(\tau(w^{\new})-\tau(w)).
\end{align*}
We define $c_i$, $c_{1,i}$, $c_{2,i}$ as follows:
\begin{align*}
c_i = & ~ \sigma(w^{\new})_i - \sigma(w)_i \\
= & ~ \underbrace{ (w^{\new}-w)_i \cdot \tau(w^{\new})_i }_{c_{1,i}} + \underbrace{ w_i \cdot (\tau(w^{\new})-\tau(w))_i }_{c_{2,i}} 
\end{align*}
We estimate the right hand side by
\begin{equation}\label{eq:low_rank_c}
\wt{c}_{i}:= \underbrace{ (w^{\new}-w)_{i}\cdot(AU_{\wt{\epsilon}}^{\new}A^{\top})_{i,i} }_{ \wt{c}_{1,i} }+ \underbrace{ w_{i}\cdot(A((U_{\wt{\epsilon}}^{-1}+\Delta)^{-1}-U_{\wt{\epsilon}})A^{\top})_{i,i} }_{ \wt{c}_{2,i} } 
\end{equation}
where $\Delta=A^{\top}W^{\new}A-A^{\top}WA$ and $\wt{\epsilon}=\frac{\epsilon}{3\beta\sqrt{k}}$.

\paragraph{Cost of computing $\wt{c}$}

For the first term in \eqref{eq:low_rank_c}. We note that $w^{\new}-w$
has only $k$ non-zeros. Hence, we only need to compute $AU_{\wt{\epsilon}}^{\new}A^{\top}$
on $k$ of the diagonals. Let $A_{S}\in\R^{k\times n}$ where each
rows of $A$ that $w^{\new}\neq w$. Then, we can compute the $k$
of the diagonals via the whole matrix
\begin{align*}
A_{S}U_{\wt{\epsilon}}^{\new}A_{S}^{\top}.
\end{align*}
Lemma \ref{lem:precondition} shows that $U_{\wt{\epsilon}}^{\new}$ has $O(1+\frac{\log\wt{\epsilon}}{\log\log (n)})$
terms and that it takes 
\begin{align*}
O\left( \T_{\mat}(n,n,k) \cdot \Big(1+\frac{\log\wt{\epsilon}}{\log\log (n)} \Big) \right)
\end{align*}
to compute $U_{\wt{\epsilon}}^{\new}A_{S}^{\top}$. Finally, it takes
extra $\T_{\mat}(k,n,k)=O(\T_{\mat}(n,n,k))$ to do the left multiplication on $A_{S}$.

For the second term in \eqref{eq:low_rank_c}. Woodbury matrix identity
shows that
\begin{align*}
(U_{\wt{\epsilon}}^{-1}+\Delta)^{-1}-U_{\wt{\epsilon}}=U_{\wt{\epsilon}}A_{S}^{\top}(\Delta_{W}^{-1}+A_{S}U_{\wt{\epsilon}}A_{S}^{\top})^{-1}A_{S}U_{\wt{\epsilon}}
\end{align*}
where $\Delta_{W}\in\R^{k\times k}$ be $W^{\new}-W$ restricted on
non-zeros. By same argument above, we can compute $A_{S}U_{\wt{\epsilon}}A_{S}^{\top}$
in 
\begin{align*}
O \left( \T_{\mat}(n,n,k) \cdot \Big(1+\frac{\log\wt{\epsilon}}{\log\log (n)} \Big) \right)
\end{align*}
time. 

Then,
we can compute $(\Delta_{W}^{-1}+A_{S}U_{\wt{\epsilon}}A_{S}^{\top})^{-1}$
in $\T_{\mat}(k,k,k)=O(\T_{\mat}(n,n,k))$ time. Then, we compute
\begin{align*}
AU_{\wt{\epsilon}}A_{S}^{\top}(\Delta_{W}^{-1}+A_{S}U_{\wt{\epsilon}}A_{S}^{\top})^{-1}\qquad\text{in time }\T_{\mat}(m,n,k)=O(\T_{\mat}(n,n,k))
\end{align*}
and
\begin{align*}
A_{S}U_{\wt{\epsilon}}=(U_{\wt{\epsilon}}A_{S}^{\top})^{\top}\qquad\text{in time }O \left(\T_{\mat}(n,n,k) \cdot \Big( 1 + \frac{\log\wt{\epsilon}}{\log\log (n)} \Big) \right).
\end{align*}
Now, we multiple the two terms above together in time $\T_{\mat}(n,k,n)=O(\T_{\mat}(n,n,k))$.
Hence, the total time is 
\begin{align*}
O \left( \T_{\mat} (n,n,k) \cdot \Big( 1+\frac{\log\wt{\epsilon}}{\log\log (n)} \Big) \right) = O( \T_{\mat} (n,n,k)(1+\log_{\log (n)}(  \beta k / \epsilon))).
\end{align*}

\paragraph{The accuracy of $\wt{c}$}

For the first term in $\wt{c}$, we have
\begin{align*}
A (A^{\top}W^{\new}A)^{-1} A^\top \preceq AU_{\wt{\epsilon}}^{\new}A^{\top} \preceq (1+\wt{\epsilon}) \cdot A (A^{\top}W^{\new}A)^{-1} A^\top .
\end{align*}
Hence, we have
\begin{align*}
|\wt{c}_{1,i} - c_{1,i}|
= & ~ \left|(w^{\new}-w)_{i}\left(AU_{\wt{\epsilon}}^{\new}A^{\top}-A(A^{\top}W^{\new}A)^{-1}A^{\top}\right)_{i,i}\right| \\ 
\leq & ~ \wt{\epsilon} \cdot | w^{\new}_i - w_i | \cdot ( A(A^{\top}W^{\new}A)^{-1}A^{\top} )_{i,i}\\
\leq & ~ \wt{\epsilon} \cdot w_{i}^{\new} \cdot ( A (A^{\top}W^{\new}A)^{-1} A^{\top} )_{i,i}\\
\leq & ~ \wt{\epsilon}
\end{align*}
where we used that $w^{\new}\geq w$ in the second inequality and
leverage score is upper bounded by $1$ in the last equality. Hence
the $\ell_{\infty}$ norm error of the first term is given by $\wt{\epsilon}$.
Hence, the $\ell_{2}$ norm error is given by 
\begin{equation}
%\textrm{Err}_{\text{1stTerm}} \text{~in~Eq.~\eqref{eq:low_rank_c}} 
\| \wt{c}_1 - c_1 \|_2
\leq\sqrt{k}\cdot\wt{\epsilon}.\label{eq:c_err_1}
\end{equation}

For the second term in $\wt{c}$, we have that the error is given by
\begin{align*}
\wt{c}_{2,i} - c_{2,i} = w_{i}\cdot(A((U_{\wt{\epsilon}}^{-1}+\Delta)^{-1}-U_{\wt{\epsilon}})A^{\top})_{i,i}-w_{i}\cdot(A((M+\Delta)^{-1}-M^{-1})A^{\top})_{i,i}
\end{align*}
where $M=A^{\top}WA$. To simplify the notation, we define $M_{t,s} = M + t \cdot \Delta_{2} + s \cdot \Delta$ where $\Delta_{2} = U_{\wt{\epsilon}}^{-1} - M$. Then, the error term becomes
\begin{equation}
(\sqrt{W}A(M_{1,1}^{-1}-M_{1,0}^{-1}-M_{0,1}^{-1}+M_{0,0}^{-1})A^{\top}\sqrt{W})_{i,i}\label{eq:c_2_err_1}
\end{equation}
Note that
\begin{align}
 & M_{1,1}^{-1}-M_{1,0}^{-1}-M_{0,1}^{-1}+M_{0,0}^{-1} \nonumber \\
= & \int_{0}^{1}\int_{0}^{1}\frac{ \d^{2} }{ \d t \d s}M_{t,s}^{-1} \d t \d s \nonumber \\
= & \int_{0}^{1}\int_{0}^{1}M_{t,s}^{-1}\Delta_{2}M_{t,s}^{-1}\Delta M_{t,s}^{-1}+M_{t,s}^{-1}\Delta M_{t,s}^{-1}\Delta_{2}M_{t,s}^{-1} \d t \d s\label{eq:c_2_err_2}
\end{align}
Combining \eqref{eq:c_2_err_1} and \eqref{eq:c_2_err_2}, we have that the $\ell_{2}$ norm error for the second term is bounded by
\begin{align*}
\| \wt{c}_2 - c_2 \|_2 %& ~ \textrm{Err}_{\text{2ndTerm}} \text{~in~Eq.~\eqref{eq:low_rank_c}}\\
\leq & ~ \left\| \int_{0}^{1}\int_{0}^{1}\sqrt{W}A(M_{t,s}^{-1}\Delta_{2}M_{t,s}^{-1}\Delta M_{t,s}^{-1}+M_{t,s}^{-1}\Delta M_{t,s}^{-1}\Delta_{2}M_{t,s}^{-1})A^{\top}\sqrt{W} \d t \d s \right\|_{F}\\
\leq & ~ 2 \max_{t \in[0,1], s \in[0,1]} \left\| \sqrt{W}AM_{t,s}^{-1}\Delta_{2}M_{t,s}^{-1}\Delta M_{t,s}^{-1}A^{\top}\sqrt{W} \right\|_{F}.
\end{align*}
To bound the Frobenius norm, we note that $\Delta \succeq 0$ by the
assumption and that $\Delta_{2}=U_{\wt{\epsilon}}^{-1}-M \succeq 0$ \eqref{eq:U_def}.
Hence, we have that $M_{t,s}\succeq M$ and
\begin{align*}
 & \|\sqrt{W}AM_{t,s}^{-1}\Delta_{2}M_{t,s}^{-1}\Delta M_{t,s}^{-1}A^{\top}\sqrt{W}\|_{F}^{2}\\
= & \tr \Big[ A^{\top}WAM_{t,s}^{-1}\Delta_{2}M_{t,s}^{-1}\Delta M_{t,s}^{-1}A^{\top}WAM_{t,s}^{-1}\Delta M_{t,s}^{-1}\Delta_{2}M_{t,s}^{-1} \Big] \\
= & \tr \Big[ M^{\frac{1}{2}}M_{t,s}^{-1}\Delta_{2}M_{t,s}^{-1}\Delta M_{t,s}^{-1}MM_{t,s}^{-1}\Delta M_{t,s}^{-1}\Delta_{2}M_{t,s}^{-1}M^{\frac{1}{2}} \Big] \\
\leq & \tr \Big[ M^{\frac{1}{2}}M_{t,s}^{-1}\Delta_{2}M_{t,s}^{-1}\Delta M_{t,s}^{-1}\Delta M_{t,s}^{-1}\Delta_{2}M_{t,s}^{-1}M^{\frac{1}{2}} \Big] \\
= & \tr \Big[ M_{t,s}^{-\frac{1}{2}}\Delta M_{t,s}^{-1}\Delta_{2}M_{t,s}^{-1}M^{\frac{1}{2}}M^{\frac{1}{2}}M_{t,s}^{-1}\Delta_{2}M_{t,s}^{-1}\Delta M_{t,s}^{-\frac{1}{2}} \Big] \\
\leq & \tr \Big[ M_{t,s}^{-\frac{1}{2}}\Delta M_{t,s}^{-1}\Delta_{2}M^{-1}\Delta_{2}M_{t,s}^{-1}\Delta M_{t,s}^{-\frac{1}{2}} \Big] .
\end{align*}
\eqref{eq:U_def} shows that $\Delta_{2} \preceq \wt{\epsilon} \cdot M$ and hence
$(M^{-\frac{1}{2}} \Delta_{2} M^{-\frac{1}{2}})^{2} \preceq \wt{\epsilon}^{2} \cdot I$.
Continuing the last equation, we have
\begin{align*}
\|\sqrt{W}AM_{t,s}^{-1}\Delta_{2}M_{t,s}^{-1}\Delta M_{t,s}^{-1}A^{\top}\sqrt{W}\|_{F}^{2} & \leq \wt{\epsilon}^{2} \cdot \tr \Big[ M_{t,s}^{-\frac{1}{2}}\Delta M_{t,s}^{-1}MM_{t,s}^{-1}\Delta M_{t,s}^{-\frac{1}{2}} \Big] \\
 & \leq \wt{\epsilon}^{2} \cdot  \tr \Big[ M_{t,s}^{-\frac{1}{2}}\Delta M_{t,s}^{-1}\Delta M_{t,s}^{-\frac{1}{2}} \Big] \\
 & \leq \wt{\epsilon}^{2} \cdot \|M^{-\frac{1}{2}}\Delta M^{-\frac{1}{2}}\|_{F}^{2}.
\end{align*}
Finally, using $\Delta$ has rank $k$, we have
\begin{align}
%\textrm{Err}_{\text{2ndTerm}} \text{~in~Eq.~\eqref{eq:low_rank_c}}
\| \wt{c}_{2} - c_2 \|_2
\leq & ~ 2\wt{\epsilon} \cdot \|M^{-\frac{1}{2}}\Delta M^{-\frac{1}{2}}\|_{F} \notag \\
\leq & ~ 2\wt{\epsilon} \cdot \sqrt{k}\|M^{-\frac{1}{2}}\Delta M^{-\frac{1}{2}}\| \notag \\
\leq & ~ 2\wt{\epsilon} \cdot \beta\sqrt{k}.\label{eq:c_err_2}
\end{align}

Combining \eqref{eq:c_err_1} and \eqref{eq:c_err_2}, the total $\ell_{2}$
error is bounded by $2\wt{\epsilon}\beta\sqrt{k}+\sqrt{k}\cdot\wt{\epsilon}\leq3\wt{\epsilon}\beta\sqrt{k}$.
This explains the choice of $\wt{\epsilon}$.
\end{proof}

\subsection{Batched low rank update}

When we are given a sequence of update and we need to compute the
change of leverage score, we cannot apply Lemma \ref{lem:simple_low_rank}
directly for the following reasons:
\begin{itemize}
\item The update in general involves both increasing some weight and decreasing
some weight. This can be fixed by splitting the update into positive
update and negative update. (See Step 1 in Algorithm \ref{alg:leverage_batch})
\item To get error $\frac{1}{n^{O(1)}}$, Lemma \ref{lem:simple_low_rank}
takes at least $n^{2}\log (n)$ time. Hence, it is too costly to do
a rank 1 update, which can happen if we have alternating positive
and negative updates (since we can only apply Lemma \ref{lem:simple_low_rank}
for a monotone change). This can be fixed by moving all positive/update
and insert in the front of the update sequence. (See Step 2 in Algorithm
\ref{alg:leverage_batch})
\item Our $\ell_{2}$ update are generally dense. This can be fixed by ignoring
all small multiplicative update. (See Line \ref{line:sparse1} and
\ref{line:sparse2} of Step 3 in Algorithm \ref{alg:leverage_batch})
\item In general, the matrices $A^{\top}WA$ can change by a factor of $e^{\Theta(T)}$ after $T$ iterations. So, Lemma \ref{lem:simple_low_rank}
takes at least $\Omega(\T_{\mat}(n,n,k)\frac{T}{\log\log (n)})$ time. This
forces us to do rank $k$ changes where $\T_{\mat}(n,n,k)\leq n^{2}\log\log (n)$.
Hence, for $\ell_{2}$ update, we need to handle coordinates
with not-so-small multiplicative updates separately. This can be fixed by cutting
the update sequence into phases of length $L$ and noticing there
are not too many coordinate that is changed by $e^{\Theta(L)}$ factor.
(See the rest of Step 3 in Algorithm \ref{alg:leverage_batch})
\end{itemize}
Given the above intuitions, Theorem \ref{thm:leverage_batch} follows from multiple use of
Lemma \ref{lem:simple_low_rank}.
\begin{theorem}[Batched Update of Leverage Score]
\label{thm:leverage_batch}
Given an initial matrix $A\in\R^{m\times n}$,
initial weight $w\in\R_{+}^{m}$ and a sequence of $T$ updates that
belong to one of the following forms:
\begin{itemize}
\item {\bf Update:} Update the weight to $w$
\item {\bf Insert:} Insert a row $a \in \R^n$ to the matrix $A \in \R^{m \times n}$ and weight $w_{a}$ to the vector $w$
\item {\bf Delete:} Delete a row $a \in \R^n$ from the matrix $A \in \R^{m \times n}$ and weight $w_{a}$ from the vector $w$
\end{itemize}
Let the sequence of matrices and vectors be $A^{(0)},A^{(1)},\cdots,A^{(T)} \in \R^{m \times n}$ and $w^{(0)},w^{(1)},\cdots,w^{(T)} \in \R^m$. We assume that $T\leq n^{0.08}$, the maximum number of rows is bounded by $O(n)$, and that the changes satisfy
\begin{itemize}
\item {\bf Update:} $\|\log w^{(k)}-\log w^{(k-1)}\|_{2} \leq 0.01$
\item {\bf Insert/Delete:} $w_{a}aa^{\top} \preceq 0.01 \cdot A^{(k-1)} W^{(k-1)} A^{(k-1)}$
\end{itemize}
Then, there is an algorithm (\textsc{BatchedUpdate} in Algorithm~\ref{alg:leverage_batch} and \ref{alg:leverage_batch2}) that finds a vector $v \in \R^n$ and a vector $c \in \R^n$ such that 
\begin{align*}
\|\log v-\log w^{(T)}\|_{\infty} \leq T \cdot n^{-0.08}, \|\log v-\log w^{(T)}\|_{2} \leq 0.01 \cdot T
\end{align*}
and that 
\begin{align*}
\| c - ( \sigma_{A^{(T)}} ( v ) - \sigma_{A^{(0)}} ( w^{(0)} ) ) \|_{2} \leq 1 / n^{100} .
\end{align*}
in time $O(n^{2}T)$.
\end{theorem}

\begin{proof}
$\ $\linebreak{}

\paragraph{Correctness}

Note that the algorithm \ref{alg:leverage_batch} involves reducing
the sequence $\{M^{(0,k)}\}$ to $\{M^{(1,k)}\}$ to $\{M^{(2,k)}\}$
and finally to $\{M^{(3,k)}\}$. All steps maintains
the matrix at $k=0$. The first two steps maintains the matrix at
the last step. For the last step, we only ignore the rows with multiplicative
changes less than $\eta=n^{-0.08}$. Since there are at most $T/L=n^{0.08}/\log^{3}(n)$
phases, the total accumulated multiplicative changes for one row is
$Tn^{-0.08}$. Therefore, we have
\begin{align*}
M^{(3,\textsc{End})} = (A^{(T)})^{\top} V A^{(T)}
\end{align*}
for some $v$ such that $\| \log v - \log w^{(T)}\|_{\infty} \leq Tn^{-0.08}$.
Furthermore, we output the good enough approximation of $\sigma_{A^{(T)}}(v)-\sigma_{A^{(0)}}(w^{(0)})=\sigma(M^{(3,\textsc{End})})-\sigma(M^{(3,0)})$
where we abused the notation to use $\sigma(M)$ to denote the leverage
score of the rows insides $M$. This is simply because we split the
difference into
\begin{align*}
\sigma(M^{(3,\textsc{End})})-\sigma(M^{(3,0)}) = \sum_{k = 1}^{T} \sigma(M^{(3,k)})-\sigma(M^{(3,k-1)})
\end{align*}
and estimate each step with $\ell_{2}$ error $n^{-1000}$.
There are at most $8T$ many terms, so the total error is less than
$n^{-100}$. (The number of terms can increases by at most a factor of 8 due to step 1, step 2 and step 4.)

\paragraph{Ratio on $M$ within each phase}

By the assumptions on insert/delete and update, we see that $0.8M^{(0,k-1)}\preceq M^{(0,k)}\preceq1.2M^{(0,k-1)}$
for all $k$. The first step clearly maintains this relation. We claim
that the second step also maintains this relation. To see this, we
let $\Delta^{(1,k)} = M^{(1,k+1)}-M^{(1,k)}$. By the relation, we have
that
\begin{equation}\label{eq:M_relation}
-0.2M^{(1,k)} \preceq \Delta^{(k)} \preceq 0.2M^{(1,k)}.
\end{equation}
Note that step 2 simply permutes $\Delta^{(1,k)}$, namely, for each $k$ there is an unique $k'$ such that $\Delta^{(1,k)}=\Delta^{(2,k')}$. Finally, we note that $M^{(1,k)}\preceq M^{(2,k')}$ for all $k$ because $M^{(2,k')}$ is simply $M^{(2,k')}$ plus some PSD matrix (which due to some insert/positive update move into the sum or some delete/negative update removed from the sum). Hence, we still have the relation \eqref{eq:M_relation}.

For the step 3, we should not expect the relation \eqref{eq:M_relation}.
However, we claim that within each phase, let $M^{(3,k)}$ and $M^{(3,\overline{k})}$
be two matrices in the phase. Then, we have that
\begin{equation}
2^{-O(L)}M^{(3,\overline{k})}\preceq M^{(3,k)}\preceq2^{O(L)}M^{(3,\overline{k})}.\label{eq:M_relation_2}
\end{equation}
The proof for this is same for the phase with increasing $M$ and
the phase with decreasing phase. Hence, we only discuss the first
case. Say $k_{1}$ is the first matrix in the phase and $k_{2}$ is
the last matrix in the phase. Note that the step $3$ does not change
the matrix $M^{(3,k_{1})}$ and $M^{(3,k_{2})}$ because we only swapping
operation order within a phase. Hence by \eqref{eq:M_relation}, we
have
\begin{align*}
M^{(3,k_{2})} = M^{(2,k_{2})} \preceq 2^{O(L)}M^{(2,k_{1})} = 2^{O(L)}M^{(3,k_{1})}.
\end{align*}
Since all matrix in the phase is sandwiched between $M^{(3,k_{1})}$
and $M^{(3,k_{2})}$ (due to the monotonicity of the sequence of $M$). Hence, we have
\eqref{eq:M_relation_2}.

\paragraph{Cost of Insert/Delete (Line \ref{line:insert_delete_lem1})}

Now, we bound the runtime. We will show the each phase takes $O(n^{2}L)$
time. Since there are $O(T/L)$ phases, this shows that the total
cost is $O(n^{2}T)$.

Each phase contains at most $L$ insert or delete. Hence, Line \ref{line:insert_delete_lem1}
involves estimating leverage score up to rank $L$ update. By \eqref{eq:M_relation_2},
we know the matrix is changed by a $2^{O(L)}$ factor. Theorem \ref{thm:leverage_batch}
shows the cost is
\begin{align*}
O( \T_{\mat} (n,n,k)(1+\log_{\log (n)}( \beta k / \epsilon )))
\end{align*}
where $\epsilon=n^{-1000}$, $\beta=2^{O(L)}$ and $k=O(L)$. By the
choice of $L=\log^{3}(n)$, we have $\T_{\mat} (n,n,k)=n^{2}$ and hence the
cost is $O(n^{2}L)$.

\paragraph{Cost of Large Update (Line \ref{line:sparse_update_lem1})}

Each phase contains at most $L$ updates. By the assumption, each
update change the $\ell_{2}$ norm of $\log w$ by at most $O(1)$.
Note that this is maintains during all steps. Hence, after merging
all updates into $1$, we have $\|\log w^{(k)}-\log w^{(k-1)}\|_{2}=O(L)$.
By the construction of $w^{(k-\frac{1}{2})}$, we have that $\|\log w^{(k)}-\log w^{(k-\frac{1}{2})}\|_{2}=O(L)$
and that $|\log w_{i}^{(k-\frac{1}{2})}-\log w_{i}^{(k)}|$ is either
$0$ or at least $\log2$ for all $i$. Hence the number of coordinates
that is non-zero in $\log w^{(k)}-\log w^{(k-\frac{1}{2})}$ is at most
$O(L^{2})$. Again, by \eqref{eq:M_relation_2}, we know the matrix
is changed by a $2^{O(L)}$ factor. Hence, Theorem \ref{thm:leverage_batch}
shows the cost is
\begin{align*}
O( \T_{\mat} (n,n,k)(1+\log_{\log (n)}( \beta k / \epsilon )))
\end{align*}
where $\epsilon=n^{-1000}$, $\beta=2^{O(L)}$ and $k=O(L^{2})$.
By the choice of $L=\log^{3}(n)$, we have $\T_{\mat} (n,n,k)=n^{2}$ and hence
the cost is $O(n^{2}L)$.

\paragraph{Cost of Small Update (Line \ref{line:less_sparse_update_lem1})}

Similar to above, we have that $\|\log w^{(k-\frac{1}{2})}-\log w^{(k-1)}\|_{2}=O(L)$.
Furthermore, we have that $\|\log w^{(k-\frac{1}{2})}-\log w^{(k-1)}\|_{\infty}=O(1)$.
Due to Line \ref{line:sparse1} and \ref{line:sparse2}, we removed
all rows with multiplicative changes less than $\eta$. Hence, the
number of coordinates that is non-zero in $\log w^{(k-\frac{1}{2})}-\log w^{(k-1)}$
is at most $O(L^{2}/\eta^{2})$. Hence, Theorem \ref{thm:leverage_batch}
shows the cost is
\begin{align*}
O( \T_{\mat} (n,n,k)(1+\log_{\log (n)}( \beta k / \epsilon )))
\end{align*}
where $\epsilon=n^{-1000}$, $\beta=O(1)$ and $k=O(L^{2}/\eta^{2})$.
By the choice of $L=\log^{3}(n)$ and $\eta=n^{-0.08}$, we have $\T_{\mat} (n,n,k)=n^{2}\log^{2}(n)$
and hence the cost is $O(n^{2}\log^{3}(n))=O(n^{2}L)$.

\paragraph{Summary}

For the cost of the rest of the algorithm, we note that step 1, 2
and 3 simply involves rearrangements and it takes $O(nT)$ for
all of them, much faster than $O(n^{2}T)$. Finally, we note that
Theorem \ref{thm:leverage_batch} requires having an approximate explicit
matrix inverse. This can be done by matrix maintenance (Theorem \ref{thm:maintain_projection_inn})
which takes $n^{2}L$ per phase. To see this, we note that Theorem
\ref{thm:maintain_projection_inn} takes $O(n^{2})$ time per rank 1 update
and it takes $O(n^{2})$ time per unit multiplicative $\ell_{2}$
changes in the weight (if we pick $\epsilon=0.01$ in Theorem \ref{thm:maintain_projection_inn}).
Hence, each phase, the cost is just $O(n^{2}L)$.
\end{proof}

\begin{algorithm}[htp!]
\caption{Batched Update of Leverage Score (Theorem \ref{thm:leverage_batch})}
\label{alg:leverage_batch}

\algnewcommand{\LineComment}[1]{\State \(\triangleright\) #1}
\begin{algorithmic}[1]

%\Procedure{$\textsc{BatchUpdate}$}{$\{A^{(k)}\in\R^{m^{(k)}\times n},w^{(k)}\in \R^{m^{(k)}}\}_{k=0}^{T}$}
%\State {\bf Input:}
\State {\bf Procedure} \textsc{BatchUpdate}($\{A^{(k)}\in\R^{m^{(k)}\times n},w^{(k)}\in \R^{m^{(k)}}\}_{k=0}^{T}$)

\State \ 

\State $T_{0}\leftarrow T$ and $M^{(0,k)}\leftarrow A^{(k)\top}W^{(k)}A^{(k)}$\Comment{Easier
to program via the $n\times n$ matrices $A^{\top}WA$}

\State \ 

\LineComment{{\bf Step 1}: Split the update into positive and negative
update}

\State $M^{(1,0)}\leftarrow M^{(0,0)}$

\For{$k=1,\cdots,T_{0}$}

\If{$\rank(M^{(0,k)}-M^{(0,k-1)})=1$}\Comment{Insert/Delete Step}

\State $M^{(1,\textsc{End}+1)}\leftarrow M^{(0,k)}$\Comment{Append
$M^{(0,k)}$ to the end of the list of matrix $\{M^{(0,:)}\}$}

\Else\Comment{Update Step}

\State Write $M^{(0,k)}=A^{\top}W^{(k)}A$ and $M^{(0,k-1)}=A^{\top}W^{(k-1)}A$

\State $M^{(1,\textsc{End}+1)}\leftarrow A^{\top}\max(W^{(k)},W^{(k-1)})A$\Comment{positive
update: $M^{(1,\textsc{End}+1)}\succeq M^{(1,\textsc{End})}$}

\State $M^{(1,\textsc{End}+1)}\leftarrow M^{(0,k)}$\Comment{negative
update: $M^{(1,\textsc{End}+1)}\preceq M^{(1,\textsc{End})}$}

\EndIf
	\State $\textsc{End}\leftarrow \textsc{End}+1$
\EndFor

\State \ 

\LineComment{{\bf Step 2}: Move insert / positive update into front}

\State $M^{(2,0)}\leftarrow M^{(1,0)}$, $L\leftarrow\log^{3}(n)$\Comment{$L$
is the size of the phase}

\For{$k=1,\cdots,T_{1}$ where $M^{(1,T_{1})}$ is the last element
in $\{M^{(1,:)}\}$}

\If{$M^{(1,k)}\succeq M^{(1,k-1)}$}\Comment{Insert / Positive
Update Step}

\State $M^{(2,\textsc{End}+1)}\leftarrow M^{(2,\textsc{End})}+M^{(1,k)}-M^{(1,k-1)}$\Comment{Put
it into front}
\State $\textsc{End}\leftarrow \textsc{End}+1$

\EndIf

\EndFor
\State \ 

\State Append enough $M^{(2,\textsc{End})}$ into $\{M^{(2,:)}\}$
such that the last index is a multiple of $L$\label{label:appendL1}

\For{$k=1,\cdots,T_{1}$}

\If{$M^{(1,k)}\prec M^{(1,k-1)}$}\Comment{Delete / Negative Update
Step}

\State $M^{(2,\textsc{End}+1)}\leftarrow M^{(2,\textsc{End})}+M^{(1,k)}-M^{(1,k-1)}$\Comment{Put
it into end}
\State $\textsc{End}\leftarrow \textsc{End}+1$

\EndIf

\EndFor

\State Append enough $M^{(2,\textsc{End})}$ into $\{M^{(2,:)}\}$
such that the last index is a multiple of $L$\label{label:appendL2}
\State \ 
%\State Continue in Algorithm \ref{alg:leverage_batch2}

\algstore{leverage_b}

\end{algorithmic}
\end{algorithm}

\begin{algorithm}[htp!]
\caption{Batched Update of Leverage Score (Theorem \ref{thm:leverage_batch}, Continuation of Algorithm~\ref{alg:leverage_batch})}
\label{alg:leverage_batch2}

\algnewcommand{\LineComment}[1]{\State \(\triangleright\) #1}

\begin{algorithmic}[1]

\algrestore{leverage_b}
\LineComment{{\bf Step 3}: Within each phase, merge rank-1 update together
and sparsify non rank-1 update}

\State $M^{(3,0)}\leftarrow M^{(2,0)}$, $\eta\leftarrow n^{-0.08}$\Comment{Ignore
all updates with multiplicative change less than $\eta$}

\For{$\ell=1,\cdots,T_{2}/L$}\Comment{Due to line \ref{label:appendL1}
and \ref{label:appendL2}, each phase is monotone}

\If{Phase $\ell$ involves only insert / positive update}

\State $\Delta_{\text{Insert}}\leftarrow\sum_{k}(M^{(2,k)}-M^{(2,k-1)})$
where it sums rank $1$ update in phase $\ell$

\State $\Delta_{\text{PosUpdate}}\leftarrow\sum_{k}(M^{(2,k)}-M^{(2,k-1)})$
where it sums rank $>1$ update in phase $\ell$

\State $M^{(3,\textsc{End}+1)}\leftarrow M^{(3,\textsc{End})}+\Delta_{\text{Insert}}$

\State Write $M^{(3,\textsc{End})}=A^{\top}WA$ and $\Delta_{\text{PosUpdate}}=A^{\top}\diag(\delta)A$

\State $M^{(3,\textsc{End}+1)}\leftarrow M^{(3,\textsc{End})}+A^{\top}\diag(\delta_{\text{sparse}})A$
where $\delta_{\text{sparse}}=\delta1_{|\delta|\geq\eta|w|}$.\label{line:sparse1}

\Else

\State $I\leftarrow$ the list of rows deleted in phase $\ell$

\State $\Delta_{\text{NegUpdate}}\leftarrow\sum_{k}(M^{(2,k)}-M^{(2,k-1)})$
where it sums rank $>1$ update in phase $\ell$

\State Write $M^{(3,\textsc{End})}=A^{\top}WA$ and $\Delta_{\text{NegUpdate}}=A^{\top}\diag(\delta)A$

\State $M^{(3,\textsc{End}+1)}\leftarrow M^{(3,\textsc{End})}+A^{\top}\diag(\delta_{\text{sparse}})A$
where $\delta_{\text{sparse}}=\delta1_{|\delta|\geq\eta|w|}$.\label{line:sparse2}

\State $M^{(3,\textsc{End}+1)}\leftarrow$ $M^{(3,\textsc{End})}$
with rows in $I$ removed

\EndIf
\State $v\leftarrow w + \delta_{\text{sparse}}$
\State $\textsc{End}\leftarrow \textsc{End}+1$
\EndFor
\State \ 

\LineComment{{\bf Step 4}: Add up the estimate change of leverage scores}

\State $c\leftarrow0$

\For{$k=1,\cdots,T_{3}$}

\State Write $M^{(3,k)}=A^{\top}W^{(k)}A$ and $M^{(3,k-1)}=A^{\top}W^{(k-1)}A$

\If{There is no insert/delete between $M^{(3,k)}$ and $M^{(3,k-1)}$}\Comment{Update}

\State $w^{(k-\frac{1}{2})}\leftarrow\text{median}(\frac{1}{2}w^{(k-1)},w^{(k)},2w^{(k-1)})$\Comment{Entry-wise
median}

\State $\delta_{c}\leftarrow$ Estimate of $\sigma(w^{(k-\frac{1}{2})})-\sigma(w^{(k-1)})$
using Lemma \ref{lem:simple_low_rank} with error $\frac{1}{n^{1000}}$\label{line:less_sparse_update_lem1}

\State $\delta_{c}\leftarrow\delta_{c}+$ Estimate of $\sigma(w^{(k)})-\sigma(w^{(k-\frac{1}{2})})$
using Lemma \ref{lem:simple_low_rank} with error $\frac{1}{n^{1000}}$\label{line:sparse_update_lem1}

\Else\Comment{Insert/Delete}

\State $\delta_{c}\leftarrow$ Estimate of $\sigma(w^{(k)})-\sigma(w^{(k-1)})$
using Lemma \ref{lem:simple_low_rank} with error $\frac{1}{n^{1000}}$\label{line:insert_delete_lem1}

\EndIf

\State $c\leftarrow c+\delta_{c}$

\EndFor

\State \Return $v,c$
\State {\bf EndProcedure}
%\EndProcedure

\end{algorithmic}
\end{algorithm}
 %%% Section 6
\newpage
\section{Simple Deterministic Leverage Score Maintenance}\label{sec:simple_leverage_score}

In this section, we give a simple deterministic leverage score maintenance data-structure which is used by both the inner and middle phases in our main data structure in Section~\ref{sec:main_leverage_score}.
As a sub-procedure, we make use of the batched low-rank update procedure in Section~\ref{sec:batched} which requires the following assumption for a sequence of updates $\textsf{acts}$:

\begin{assumption}[Short sequence and small updates] \label{assumption:short_sequence_small_update}
Let the sequence of matrices and vectors in $\mathsf{acts}$ be $A^{(0)},A^{(1)},\cdots,A^{(T)} \in \R^{m \times n}$ and $w^{(0)},w^{(1)},\cdots,w^{(T)} \in \R^m$. We assume that $T \leq n^{0.08}$, the maximum number of rows is bounded by $O(n)$, and that the changes satisfy
\begin{itemize}
\item {\bf Update:} $\|\log w^{(k)}-\log w^{(k-1)}\|_{2} \leq 0.01$.
\item {\bf Insert/Delete:} $w_{a}aa^{\top} \preceq 0.01(A^{(k-1)})^\top W^{(k-1)}A^{(k-1)}$ .
\end{itemize}
\end{assumption}

\subsection{Main result}

\begin{theorem}[Simple leverage score maintenance] \label{thm:maintain_leverage_score_simple}
Given an initial matrix $A \in \R^{m \times n}$ with $m = O(n)$, initial weight $w \in \R_+^m$, and error parameter $\epsilon_{\simp} \leq 0.01$.
There is a deterministic data structure (Algorithm~\ref{alg:maintain_leverage_score_simple}) that approximately maintains the  leverage scores
\begin{align*}
\sigma_i(w) = (\sqrt{W} A (A^{\top} W A)^{-1} A^{\top} \sqrt{W})_{i,i}, 
\end{align*}
for positive diagonal matrices $W = \diag(w)$ through the following operations:
\begin{enumerate}
	\item $\textsc{Init}(A,w,\epsilon_{\simp})$: takes $O(n^{\omega + o(1)})$ time to initialize the data structure. 

	\item $\textsc{Update}(\mathsf{acts})$: update the matrix $A$, vector $w$ and the approximation of the leverage scores for the sequence of updates in $\mathsf{acts}$. 
	Moreover, if $\mathsf{acts}$ satisfies Assumption~\ref{assumption:short_sequence_small_update}, then the function $\textsc{Update}(\mathsf{acts})$ makes one call to $\textsc{pm}.\textsc{update}$ \footnote{The running time of procedure \textsc{Update} depends on $\epsilon_{\mathrm{simp}}$, see details in Theorem~\ref{thm:maintain_projection}} and takes extra time $O(n^2 \cdot |\mathsf{acts}|)$ to compute a vector $\Delta \wt{\sigma} \in \R^m$ such that 
	\begin{align*}
	\norm{\Delta \wt{\sigma} - \Delta \sigma}_2 
	\leq & ~ \frac{1}{n^{100}} + O\left(\epsilon_{\simp} + |\mathsf{acts}| \cdot n^{-0.08} \right) \cdot |\mathsf{acts}| .
	\end{align*}

	\item $\textsc{RefineEstimate}(\wt{\sigma}^{\new})$: takes $O(n)$ time to update the approximation of $\sigma(w)$ to $\wt{\sigma}^{\new}$.

	\item $\textsc{Query}()$: takes $O(n)$ time to output the approximation $\wt{\sigma} \in \R^m$ maintained by the data structure. 
\end{enumerate}
\end{theorem}

\begin{proof}[Proof of Theorem~\ref{thm:maintain_leverage_score_simple}]
We only need to prove the guarantee for the function $\textsc{Update}(\mathsf{acts})$, as the guarantees for other functions are straightforward. 
Notice that each call to the function $\textsc{Update}(\mathsf{acts})$ makes at most one call to $\textsc{pm}.\textsc{update}$, and the extra running time of $O(n^2 \cdot |\mathsf{acts}|)$ follows directly from Theorem~\ref{thm:leverage_batch}. 
To prove the error guarantee, we let the sequence of matrices and vectors in $\mathsf{acts}$ be $A^{(0)},A^{(1)},\cdots,A^{(T)}$ and $w^{(0)},w^{(1)},\cdots,w^{(T)}$.
notice that by Theorem~\ref{thm:leverage_batch}, we have
\begin{align*}
\norm{ \sigma_{A^{(T)}}(w^{\mid}) - \sigma_{A^{(0)}}(w^{(0)}) - \Delta \wt{\sigma}^{\mid} }_2 \leq \frac{1}{n^{100}}. 
\end{align*}
It follows from Lemma~\ref{lem:error_bound_inner_w_update_simple} that
\begin{align*}
\norm{\sigma_{A^{(T)}}(w^{(T)}) - \sigma_{A^{(T)}}(w^{\mid}) - \Delta \wt{\sigma}^{\new} }_2 
\leq & ~ O\left(\epsilon_{\simp} + |\mathsf{acts}| \cdot n^{-0.08} \right) \cdot |\mathsf{acts}| .
\end{align*}
The error guarantee then follows by summing up the two error upper bounds above. 
\end{proof}

The following corollary is an immediate consequence of Theorem~\ref{thm:maintain_leverage_score_simple} and~\ref{thm:maintain_projection}. 
It states that if matrix multiplication can be performed fast enough, then the simple leverage score maintenance would imply a deterministic $O(n^3 \log(n/\epsilon))$ algorithm for the cutting plane method.

\begin{corollary} \label{cor:det_cutting_plane}
If $\T_{\mat}(n,n,r) = O(n^2 \log^{O(1)}(n))$ for $r = n^{\alpha}$ with $\alpha > 2/3$, then there's a deterministic $O(n^3 \log(n/\epsilon))$ time algorithm for the cutting plane method.
\end{corollary}
\begin{proof}
This algorithm only uses the inner and middle phases. The inner phase is run as is our main algorithm. For the middle phase, we notice that, in this case, the exponent of matrix multiplication time can be bounded as $\omega < 3 - 2/3 = 7/3$. Therefore, the data structure is restarted after every $n^{1/3}$ calls to the middle phase. For the middle phase, we use the simple leverage score maintenance data structure in Theorem~\ref{thm:maintain_leverage_score_simple} with parameter $\epsilon_{\simp} = n^{-\alpha/2} / \log^{O(1)}(n)$. It follows from Theorem~\ref{thm:maintain_leverage_score_simple} that the error accumulated in the $n^{0.2}$ steps is $O(n^{-\alpha/2} n^{1/3} \log^{O(1)}(n)) = n^{-\Omega(1)}$. By Theorem~\ref{thm:maintain_leverage_score_simple} and~\ref{thm:leverage_score_maintain_main}, the running time is $O(n^2 \log^{O(1)}(n))$ per step for middle phase which is fast enough.
This finishes the proof of the corollary.
\end{proof}

\begin{algorithm}[htp!]\caption{}\label{alg:maintain_leverage_score_simple}
\begin{algorithmic}[1]
\State {\bf data structure} \textsc{MaintainLeverageScoreSimple} \Comment{Theorem~\ref{thm:maintain_leverage_score_simple}}
\State
\State {\bf members}
\State \hspace{4mm} $m \in \mathbb{Z}_+$ \Comment{$m = \Theta(n)$}
\State \hspace{4mm} $A \in \R^{m \times n}$  
\State \hspace{4mm} $w \in \R^m$ \Comment{Target vector}
\State \hspace{4mm} $\wt{\sigma} \in \R^m$ \Comment{The approximate leverage scores}
\State \hspace{4mm} \textsc{ProjectionMaintain} \textsc{pm} \Comment{Theorem~\ref{thm:maintain_projection_mid} or \ref{thm:maintain_projection_inn}}
\State {\bf end members}
\State
\Procedure{\textsc{Init}}{$A\in \R^{m \times n}, w\in \R^m_+, \epsilon_{\simp} \in \R_+$}  \Comment{Initialization}
	\State $A \leftarrow A$, $m \leftarrow m$, $w \leftarrow w$
	\State $\wt{\sigma} \leftarrow \sigma(w)$
	\State $\textsc{pm}.\textsc{init}(A, w,\epsilon_{\simp})$ 
\EndProcedure
\State
\Procedure{\textsc{Update}}{$\mathsf{acts} \in \textsc{Vector} \langle \textsc{Action} \rangle$} \Comment{Update estimate for a sequence of updates}
	\State Let $A$ and $w^{\new}$ be the final matrix and vector in $\mathsf{acts}$
	\State $w^{\mid}, \Delta \wt{\sigma}^{\mid} \leftarrow \textsc{BatchedUpdate}$($\mathsf{acts}$) \Comment{Let $w^{\mid}, \Delta \wt{\sigma}^{\mid}$ be computed as in Theorem~\ref{thm:leverage_batch} (Algorithm~\ref{alg:leverage_batch} and \ref{alg:leverage_batch2}) with the sequence of updates in $\mathsf{acts}$}
	\State $v^{\mid}, M(v^{\mid})^{-1}, Q(v^{\mid}) \leftarrow$ \textsc{pm}.\textsc{update}$(w^{\mid})$
	\State Compute $\Delta \wt{\sigma}^{\new}$ as in Lemma~\ref{lem:compute_sigma_change_simple}
	\State $\wt{\sigma} \leftarrow \wt{\sigma} + \Delta \wt{\sigma}^{\mid} + \Delta \wt{\sigma}^{\new}$
	\State $w \leftarrow w^{\new}$
\EndProcedure
\State
\Procedure{\textsc{RefineEstimate}}{$\wt{\sigma}^{\new} \in \R_+^m$}   \Comment{Refine the estimation $\wt{\sigma}$ to $\wt{\sigma}^{\new}$}
	\State $\wt{\sigma} \leftarrow \wt{\sigma}^{\new}$
\EndProcedure
\State
\Procedure{\textsc{Query}}{$ $}
	\State \Return $\wt{\sigma}$
\EndProcedure
\State
\State {\bf end data structure}
\end{algorithmic}
\end{algorithm}

\subsection{Approximate leverage score's moving}

\begin{lemma}[Leverage score's moving] \label{lem:leverage_score_moving_simple}
Given any matrix $A \in \R^{m \times n}$, vectors $w^{\mid}, w^{\new} \in \R_+^m$.
The change in leverage score from $w^{\mid}$ to $w^{\new}$ can be written as
\begin{align*}
\Delta \sigma^{\new}_i 
:= & ~\sigma(w^{\new} )_i - \sigma(w^{\mid})_i \\
= & ~ (w^{\new}_i - w^{\mid}_i) \cdot \tau(w^{\mid})_i + w^{\new}_i \cdot \int_0^1 e_i^\top Q(x_t) (W^{\mid}-W^{\new}) Q(x_t) e_i \mathrm{d} t ,
\end{align*}
where we define $x_t = w^{\mid} + t(w^{\new} - w^{\mid})$ for each $t \in [0,1]$.
\end{lemma}

\begin{proof}
We can rewrite $\sigma(w^{\new})_i - \sigma(w^{\mid})_i$ as 
\begin{align*}
\sigma(w^{\new})_i - \sigma(w^{\mid})_i = (w^{\new}_i - w^{\mid}_i) \cdot \tau(w^{\mid})_i + w^{\new}_i (\tau(w^{\new})_i - \tau(w^{\mid})_i) . 
\end{align*}
Then the second term is computed as
\begin{align*}
& ~ w^{\new}_i(\tau( w^{\new} )_i - \tau(w^{\mid})_i) \\
 = & ~ w^{\new}_i \cdot ((A (A^\top W^{\new} A )^{-1} A^\top )_{i,i} - ( A ( A^\top W ^{\mid}A )^{-1} A^\top )_{i,i} ) \\
= & ~ w^{\new}_i \cdot \int_0^1 e_i^\top Q(x_t) (W^{\mid}-W^{\new}) Q(x_t) e_i \mathrm{d} t .
\end{align*}
where the last step follows by the definition of $Q(v) = A (A^\top V A)^{-1} A^\top \in \R^{m \times m}$. 
\end{proof}

\begin{lemma}[Computing $\Delta \wt{\sigma}^{\new}$] \label{lem:compute_sigma_change_simple}
Given any matrix $A \in \R^{m \times n}$, vectors $w^{\mid}, w^{\new} \in \R_+^m$, where $m = O(n)$.
Define an approximation $\Delta \wt{\sigma}^{\new}_i$ to the change in leverage score to be
\begin{align*}
\Delta \wt{\sigma}^{\new}_i 
= & ~ (w^{\new}_i - w^{\mid}_i) \cdot e_i^\top Q(v^{\mid}) e_i + w^{\new}_i e_i^\top Q(v^{\mid}) (W^{\mid}-W^{\new}) Q(v^{\mid}) e_i . 
\end{align*}
Then given the matrix $Q(v^{\mid})$, we can compute $\Delta \wt{\sigma}^{\new}_i$ for all $i \in [m]$ in time $O(n^2)$. 
\end{lemma}

\begin{proof}
Since we know the matrix $Q(v^{\mid})$, the first term is essentially reading all the diagonal entries of matrix $Q(v^{\mid})$, which can be computed for all $i \in [m]$ in time $O(n)$.
For the second term, we denote $Q = Q(x_t)$ and $Q^{(2)}$ the entry-wise square of $Q$. Then we have
\begin{align*}
  & ~ w_i^{\new} e_i^\top Q ( W^{\mid} - W^{\new} ) Q e_i \\
= & ~ w_i^{\new} \cdot \sum_{j \in [m]} (Q_{i,j})^2 ( w^{\mid} - w^{\new} )_j \\
= & ~ w_i^{\new} \cdot e_i^\top Q^{(2)} ( w^{\mid} - w^{\new} ) .
\end{align*}
Therefore, in order to compute the second term for all $i \in [m]$, it suffices to compute the matrix-vector product $Q^{(2)} ( w^{\mid} - w^{\new} )$, which takes time $O(n^2)$. 
\end{proof}

\subsection{$\ell_2$-error bound}

Let $A^{(T)}$ and $A^{(0)}$ be the final and initial matrix in the sequence of updates in $\mathsf{acts}$, and $w^{(0)}$ be the initial vector in $\mathsf{acts}$.
Denote $\Delta \sigma^{\mid} = \sigma_{A^{(T)}}(w^{\mid}) - \sigma_{A^{(0)}}(w^{(0)})$. 
It follows from Theorem~\ref{thm:leverage_batch} that if Assumption~\ref{assumption:short_sequence_small_update} holds, then we have 
\begin{align*}
\norm{ \Delta \wt{\sigma}^{\mid} - \Delta \sigma^{\mid} }_2 \leq \frac{1}{n^{100}}.
\end{align*}
It therefore suffices to bound the $\ell_2$-error $\norm{ \Delta \wt{\sigma}^{\new} - \Delta \sigma^{\new} }_2$. 
For simplicity, we use $A$ to denote $A^{(T)}$. 
Recall from Lemma~\ref{lem:leverage_score_moving_simple} that
\begin{align*}
\Delta \sigma^{\new}_i 
:= & ~\sigma(w^{\new} )_i - \sigma(w^{\mid})_i \\
= & ~ (w^{\new}_i - w^{\mid}_i) \cdot \tau(w^{\mid})_i + w^{\new}_i \cdot \int_0^1 e_i^\top Q(x_t) (W^{\mid}-W^{\new}) Q(x_t) e_i \mathrm{d} t ,
\end{align*}
where $x_t = w^{\mid} + t(w^{\new} - w^{\mid})$ for each $t \in [0,1]$.
Also recall that in Lemma~\ref{lem:compute_sigma_change_simple}, our approximation $\Delta \wt{\sigma}^{\new}$ is computed as
\begin{align*}
\Delta \wt{\sigma}^{\new}_i 
= & ~ (w^{\new}_i - w^{\mid}_i) \cdot e_i^\top Q(v^{\mid}) e_i + w^{\new}_i e_i^\top Q(v^{\mid}) (W^{\mid}-W^{\new}) Q(v^{\mid}) e_i . 
\end{align*}
We use $\Delta_1$ and $\Delta_2$ to denote the two corresponding error terms  
\begin{align*}
(\Delta_1)_i = & ~ (w^{\new}_i - w^{\mid}_i) \cdot e_i^\top Q(v^{\mid}) e_i - (w^{\new}_i - w^{\mid}_i) \cdot \tau(w^{\mid})_i \\
(\Delta_2)_i = & ~ w^{\new}_i e_i^\top Q(v^{\mid}) (W^{\mid}-W^{\new}) Q(v^{\mid}) e_i - w^{\new}_i \cdot \int_0^1 e_i^\top Q(x_t) (W^{\mid}-W^{\new}) Q(x_t) e_i \mathrm{d} t ,
\end{align*}
where $x_t = w^{\mid} + t(w^{\new} - w^{\mid})$.
It follows that
\begin{align*}
\norm{\Delta \wt{\sigma}^{\new} - \Delta \sigma^{\new}}_2 \leq \norm{\Delta_1}_2 + \norm{\Delta_2}_2 .
\end{align*}
We prove the following bound on the $\ell_2$-error.

\begin{lemma}[$\ell_2$-error bound] \label{lem:error_bound_inner_w_update_simple}
Assume Assumption~\ref{assumption:short_sequence_small_update} holds and that $\epsilon_{\simp} \leq 0.01$, where $\epsilon_{\simp}$ is the error parameter in Algorithm~\ref{alg:maintain_leverage_score_simple}.
Then the $\ell_2$-error is upper bounded as 
\begin{align*}
\norm{\Delta \wt{\sigma}^{\new} - \Delta \sigma^{\new} }_2 
\leq & ~ O\left(\epsilon_{\simp} + |\mathsf{acts}| \cdot n^{-0.08} \right) \cdot |\mathsf{acts}| .
\end{align*}
\end{lemma}
\begin{proof}
For simplicity, we assume here that $w^{\new} - w^{\mid} \geq 0$  and that $w^{\mid} - v^{\mid} \geq 0$. 
This is without loss of generality, since in the case where these assumptions don't hold, we can break the corresponding terms into positive and negative parts and bound each one of them. 
The lemma then follows immeditely from the following Lemma~\ref{lem:error_first_two_terms_simple} and~\ref{lem:error_fourth_term_simple}. 
\end{proof}

\begin{lemma}[$\norm{\Delta_1}_2$] \label{lem:error_first_two_terms_simple}
Assume Assumption~\ref{assumption:short_sequence_small_update} holds and that $\epsilon_{\simp} \leq 0.01$.
Then we have the following upper bound on $\norm{\Delta_1}_2$
\begin{align*}
\norm{\Delta_1}_2 \leq & ~ O(\epsilon_{\simp}) \cdot |\mathsf{acts}| .
\end{align*}
\end{lemma}

\begin{proof}
Recall the definition of the error term $\Delta_1$ as
\begin{align*}
(\Delta_1)_i = & ~ (w^{\new}_i - w^{\mid}_i) \cdot e_i^\top Q(v^{\mid}) e_i - (w^{\new}_i - w^{\mid}_i) \cdot \tau(w^{\mid})_i .
\end{align*}
We define $y_s = w^{\mid} + s(v^{\mid} - w^{\mid})$ for $s \in [0,1]$.
The term $(\Delta_1)_i$ can be rewritten as
\begin{align*}
(\Delta_1)_i 
= & ~ (w^{\new}_i - w^{\mid}_i) \cdot e_i^\top Q(v^{\mid}) e_i - (w^{\new}_i - w^{\mid}_i) \cdot \tau(w^{\mid})_i \\
= & ~ (w^{\new}_i - w^{\mid}_i) \cdot e_i^\top \left(\int_0^1 Q(y_s) (W^{\mid} - V^{\mid}) Q(y_s) \mathrm{d} s \right) e_i \\
= & ~ e_i^\top \sqrt{W^{\new} - W^{\mid}}   Q(y_{\xi}) (W^{\mid} - V^{\mid}) Q(y_{\xi}) \sqrt{W^{\new} - W^{\mid}} e_i ,
\end{align*}
for some $\xi \in [0,1]$. It follows that 
\begin{align*}
\norm{\Delta_1}_2^2 
\leq & ~ \norm{ \sqrt{ \frac{ W^{\new} - W^{\mid} } {Y_{\xi} } }  P(y_{\xi}) \frac{W^{\mid} - V^{\mid}}{Y_{\xi}} P(y_{\xi}) \sqrt{ \frac{ W^{\new} - W^{\mid} } {Y_{\xi} } } }_F^2 \\
= & ~ \tr [ (\sqrt{C} P B P \sqrt{C})^2 ] ,
\end{align*}
where in the last step we define diagonal matrices $B, C \in \R^{m \times m}$ as
\begin{align*}
B =   \frac{ W^{\mid} - V^{\mid} }{ Y_{\xi} }    ,
C =   \frac{ W^{\new} - W^{\mid}  }{ Y_{\xi} } ,
\end{align*}
and we use $P$ to denote $P(y_{\xi})$.
Then we have
\begin{align*}
  & ~\tr[ (\sqrt{C} P B P \sqrt{C})^2 ]  \\
\leq & ~ \norm{b}_{\infty}^2 \cdot \tr[ C^2 ]  & \text{$PBP \preceq \norm{b}_\infty \cdot I$} \\
= & ~ O(1) \cdot \norm{\log(w^{\mid}) - \log(v^{\mid}) }_{\infty}^2 \cdot \norm{\log(w^{\new}) - \log( w^{\mid}) }_2^2 \\
\leq & ~ O(\epsilon_{\simp}^{2}) \cdot \norm{\log(w^{\new}) - \log( w^{\mid}) }_2^2 \\
\leq & ~ O(\epsilon_{\simp}^{2}) \cdot |\mathsf{acts}|^2 ,
\end{align*}
where the last step follows from Theorem~\ref{thm:leverage_batch}.
\end{proof}

\begin{lemma}[$\norm{\Delta_2}_2$] \label{lem:error_fourth_term_simple}
Assume Assumption~\ref{assumption:short_sequence_small_update} holds and that $\epsilon_{\simp} \leq 0.01$.
Then we have the following upper bound on $\norm{\Delta_2}_2$
\begin{align*}
\norm{\Delta_2}_2 \leq O\left(\epsilon_{\simp} + |\mathsf{acts}| \cdot n^{-0.08} \right) \cdot |\mathsf{acts}| .
\end{align*}
\end{lemma}

\begin{proof}
Recall the definition of the error term $\Delta_2$ as
\begin{align*}
(\Delta_2)_i = w^{\new}_i e_i^\top Q(v^{\mid}) (W^{\mid}-W^{\new}) Q(v^{\mid}) e_i - w^{\new}_i \cdot \int_0^1 e_i^\top Q(x_t) (W^{\mid}-W^{\new}) Q(x_t) e_i \mathrm{d} t ,
\end{align*}
where $x_t = w^{\mid} + t(w^{\new} - w^{\mid})$. 
It follows that there exists $t \in [0,1]$ such that
\begin{align*}
(\Delta_2)_i 
= & ~ w^{\new}_i e_i^\top Q(v^{\mid}) (W^{\mid}-W^{\new}) Q(v^{\mid}) e_i - w^{\new}_i e_i^\top Q(x_t) (W^{\mid}-W^{\new}) Q(x_t) e_i \\
= & ~ \int_0^1 \frac{\mathrm{d}}{\mathrm{d} s} \left( e_i^\top \sqrt{W^{\new}} Q(v_{t,s}) (W^{\mid}-W^{\new}) Q(v_{t,s}) \sqrt{W^{\new}} e_i \right) \mathrm{d} s ,
\end{align*}
where in the last line we define $v_{t,s} = x_t + s(v^{\mid} - x_t)$ for $s \in [0,1]$. 
Therefore, in order to bound the $\norm{\Delta_2}_i$, it suffices to bound the $\ell_2$-norm of the derivative:
\begin{align*}
\norm{\Delta_2}_2 
\leq & ~\int_0^1 \norm{ \left( \frac{\mathrm{d}}{\mathrm{d} s} \left( e_i^\top \sqrt{W^{\new}} Q(v_{t,s}) (W^{\mid}-W^{\new}) Q(v_{t,s}) \sqrt{W^{\new}} e_i \right) \right )_{i=1}^m }_2 \mathrm{d} s \\
\leq & ~ \sup_{s \in [0,1]} \norm{ \left( \frac{\mathrm{d}}{\mathrm{d} s} \left( e_i^\top \sqrt{W^{\new}} Q(v_{t,s}) (W^{\mid}-W^{\new}) Q(v_{t,s}) \sqrt{W^{\new}} e_i \right) \right )_{i=1}^m }_2 .
\end{align*}
Notice that
\begin{align*}
& ~ \frac{\mathrm{d}}{\mathrm{d} s} \left( e_i^\top \sqrt{W^{\new}} Q(v_{t,s}) (W^{\mid}-W^{\new}) Q(v_{t,s}) \sqrt{W^{\new}} e_i \right) \\
= & ~ -2 e_i^\top \sqrt{W^{\new}} Q(v_{t,s}) (V^{\mid} - X_t) Q(v_{t,s}) (W^{\mid}-W^{\new}) Q(v_{t,s}) \sqrt{W^{\new}} e_i .
\end{align*}
It follows that for any $s \in [0,1]$, we have 
\begin{align*}
& ~ \norm{ \left( \frac{\mathrm{d}}{\mathrm{d} s} \left( e_i^\top \sqrt{W^{\new}} Q(v_{t,s}) (W^{\mid}-W^{\new}) Q(v_{t,s}) \sqrt{W^{\new}} e_i \right) \right )_{i=1}^m }_2^2 \\
\leq & ~ 4 \cdot \norm{\sqrt{W^{\new}} Q(v_{t,s}) (V^{\mid} - X_t) Q(v_{t,s}) (W^{\mid}-W^{\new}) Q(v_{t,s}) \sqrt{W^{\new}}}_F^2 \\
= & ~ 4 \tr [ \sqrt{D} P C P B P D P B P C P \sqrt{D} ] ,
\end{align*}
where in the last step we define diagonal matrices $B, C \in \R^{m \times m}$ as
\begin{align*}
B =   \frac{ W^{\mid} - W^{\new} }{ V_{t,s} }    ,
C =   \frac{ V^{\mid} - X_t }{ V_{t,s} }  ,
D =   \frac{ W^{\new} }{ V_{t,s} } ,
\end{align*}
and we use $P$ to denote $P(v_{t,s})$.
By Assumption~\ref{assumption:short_sequence_small_update}, we have
\begin{align*}
\norm{\log(w^{\new}) - \log(w^{\mid})}_\infty \leq 0.01.
\end{align*}
This together with the assumption that $\epsilon_{\simp} \leq 0.01$ further implies that
\begin{align*}
\norm{ \log(w^{\new}) - \log(v_{s,t}) }_{\infty} \leq 0.1 .
\end{align*}
Therefore, we have
\begin{align*}
& ~ \tr [ \sqrt{D} P C P B P D P B P C P \sqrt{D} ] \\
\leq & ~ O(1) \cdot \tr [ \sqrt{D} P C P B^2 P C P \sqrt{D} ]  & \text{$PDP \preceq O(1) \cdot I$} \\
= & ~ O(1) \cdot \tr [ B P C P D P C P B ] \\
\leq & ~ O(1) \cdot \tr [ B P C^2 P B ]  & \text{$PDP \preceq O(1) \cdot I$} \\
\leq & ~ O(1) \cdot \norm{c}_\infty^2 \cdot \norm{b}_2^2 & \text{$PC^2P \preceq \norm{c}_\infty^2 \cdot I$} \\
\leq & ~ O(1) \cdot \norm{\log(v^{\mid}) - \log(x_t)}_\infty^2 \cdot \norm{\log(w^{\new}) - \log(w) }_2^2 \\
\leq & ~ O(1) \cdot \left(\epsilon_{\simp} + \norm{ \log(w^{\new}) - \log(w^{\mid})}_\infty \right)^2 \cdot \norm{\log(w^{\new}) - \log(w) }_2^2 .
\end{align*}
The lemma then follows from Theorem~\ref{thm:leverage_batch} which gives $\norm{\log(w^{\new}) - \log(w^{\mid})}_\infty \leq |\mathsf{acts}| \cdot n^{-0.08}$ and $\norm{\log(w^{\new}) - \log(w^{\mid})}_2 \leq 0.01 |\mathsf{acts}|$.
\end{proof}

 %%% Section 7
\newpage
\section{Complicated Randomized Leverage Score Maintenance}\label{sec:complicated_leverage_score}

In this section, we present a complicated randomized leverage score maintenance data structure which is used by the outer phase in our main data structure in Section~\ref{sec:main_leverage_score}.
We again make crucial use of the batched low-rank update procedure in Section~\ref{sec:batched}.

\subsection{Main result}

\begin{theorem}[Complicated leverage score maintenance] \label{thm:maintain_leverage_score_complicated}
Given an initial matrix $A \in \R^{m \times n}$ with $m = O(n)$, initial weight $w \in \R_+^m$, JL dimension $r > 0$, discrete sampling parameter $N \in \mathbb{Z}_+$, and error parameter $\epsilon_{\comp} < 0.01$.
There is a randomized data structure (Algorithm~\ref{alg:maintain_leverage_score_complicated}) that approximately maintains the  leverage scores
\begin{align*}
\sigma_i(w) = (\sqrt{W} A (A^{\top} W A)^{-1} A^{\top} \sqrt{W})_{i,i}, 
\end{align*}
for positive diagonal matrices $W = \diag(w)$ through the following operations:
\begin{enumerate}
	\item $\textsc{Init}(A,w,r,N,\epsilon_{\comp})$: takes $O(n^{\omega + o(1)})$ time to initialize the data structure. 

	\item $\textsc{Update}(\mathsf{acts})$: update the matrix $A$, vector $w$ and the approximation of the leverage scores for the sequence of updates in $\mathsf{acts}$. 
	Moreover, if $\mathsf{acts}$ satisfies Assumption~\ref{assumption:short_sequence_small_update}, then the function $\textsc{Update}(\mathsf{acts})$ makes $O(N)$ calls to $\textsc{pm}.\textsc{update}$ and takes extra time $O(n^2 \cdot |\mathsf{acts}| + \T_{\mat}(n,n,r) \cdot N^3 \cdot \log(n))$ to compute a random vector $\Delta \widetilde{\sigma} \in \R^m$ with error upper bounded by
	\begin{align*}
		\norm{ \Delta \sigma - \E[\Delta \wt{\sigma}] }_2 \leq \frac{1}{n^{100}} + O\left(\poly(n) / 2^{2N} \right) ,
	\end{align*}
	and variance upper bounded by
	\begin{align*}
		\E [ \norm{\sigma(w^{\new}) - \sigma(w) - \Delta \wt{\sigma} }_2^2]
		\leq & ~ O(N^3 /r) \cdot \epsilon_{\comp}^2 \cdot |\mathsf{acts}|^2 .
	\end{align*} 

	\item $\textsc{RefineEstimate}(\wt{\sigma}^{\new})$: takes $O(n)$ time to update the approximation of $\sigma(w)$ to $\wt{\sigma}^{\new}$.

	\item $\textsc{Query}()$: takes $O(n)$ time to output the approximation $\wt{\sigma} \in \R^m$ maintained by the data structure. 
\end{enumerate}
\end{theorem}

\begin{proof}[Proof of Theorem~\ref{thm:maintain_leverage_score_complicated}]

We only need to prove the guarantee for the function $\textsc{Update}(\mathsf{acts})$, as the guarantees for other functions are straightforward.
The running time upper bound for the function $\textsc{Update}(\mathsf{acts})$ is given in Lemma~\ref{lem:update_time_complicated}.
The upper bound on the $\ell_2$-error $\norm{ \Delta \sigma - \E[\Delta \wt{\sigma}] }_2$ follows from Theorem~\ref{thm:leverage_batch} and the lemmas in Section~\ref{subsec:error_bound_discrete_to_continuous}.
The upper bound on the variance follows from the lemmas in Section~\ref{subsec:upper_bound_eta_alpha_beta_gamma} and~\ref{subsec:variance_bound_random_gaussian_matrix}.
\end{proof}

\begin{figure}[htp!]
\centering
\includegraphics[width=0.6\textwidth]{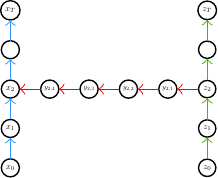}
\caption{This figure is related to Algorithm~\ref{alg:maintain_leverage_score_complicated}.}\label{fig:sequence}
\end{figure}

\begin{algorithm}[htp!]\caption{}\label{alg:maintain_leverage_score_complicated}
\begin{algorithmic}[1]
\State {\bf data structure} \textsc{MaintainLeverageScoreComplicated} \Comment{Theorem~\ref{thm:maintain_leverage_score_complicated}}
\State
\State {\bf members}
\State \hspace{4mm} $m \in \mathbb{Z}_+$  \Comment{$m = \Theta(n)$}
\State \hspace{4mm} $A \in \R^{m \times n}$ 
\State \hspace{4mm} $\wt{\sigma} \in \R^m$ \Comment{The approximate leverage scores}
\State \hspace{4mm} $w \in \R^m$ \Comment{Target vector}
%\State \hspace{4mm} $v \in \R^m$, $M(v)^{-1} \in \R^{n \times n}$, $Q(v) \in \R^{m \times m}$ \Comment{$v \approx w$, $M(v) \approx M(w)$, $Q(v) \approx Q(w)$}
\State \hspace{4mm} $r \in \mathbb{Z}_+$  \Comment{Dimension of JL matrices}
\State \hspace{4mm} $N \in \mathbb{Z}_+$, ${ \cal S }, { \cal T } \subseteq [0,1]$ \Comment{Discrete sampling points}
\State \hspace{4mm} \textsc{ProjectionMaintain} \textsc{pm} \Comment{Theorem~\ref{thm:maintain_projection}}
\State {\bf end members}
\State
\Procedure{\textsc{Init}}{$A \in \R^{m \times n}, w \in \R^m_+, r \in \mathbb{Z}_+, N \in \mathbb{Z}_+, \epsilon_{\comp} \in R_+$}  \Comment{Initialization}
	\State $A \leftarrow A$, $m \leftarrow m$, $w \leftarrow w$, $r \leftarrow r$, $N \leftarrow N$, $\wt{\sigma} \leftarrow \sigma(w)$
	\State Set ${\cal S}, {\cal T}$ with size $N$ and corresponding weights $\omega$ as in Theorem~\ref{thm:multiple_variable_integral} 
	\State $\textsc{pm}.\textsc{init}(A,w,\epsilon_{\comp})$ %\Comment{$\wt{P} \approx P$}
\EndProcedure
\State
\Procedure{\textsc{Update}}{$\mathsf{acts} \in \textsc{Vector} \langle \textsc{Action} \rangle$} \Comment{Update estimate for a sequence of updates}
	\State Let $A$ and $w^{\new}$ be the final matrix and vector in $\mathsf{acts}$
	\State $w^{\mid},\Delta \wt{\sigma}^{\mid} \leftarrow \textsc{BatchedUpdate}(\mathsf{acts})$ \Comment{ Let $w^{\mid}, \Delta \wt{\sigma}^{\mid}$ be computed as in Theorem~\ref{thm:leverage_batch} (Algorithm~\ref{alg:leverage_batch}, Algorithm~\ref{alg:leverage_batch2}) with the sequence of updates in $\mathsf{acts}$ }
	\State $v^{\mid}, M(v^{\mid})^{-1}, Q(v^{\mid}) \leftarrow$ \textsc{pm}.\textsc{update}$(w^{\mid})$
%	\For{$t = 1 \to | { \cal T } | $}
%		\State $x_t = w^{\mid} + {\cal T}_t \cdot ( w^{\new} - w^{\mid} )$ \Comment{$w^{\mid} = x_0 \rightarrow x_1 \rightarrow x_2 \rightarrow \cdots \rightarrow x_{|{\cal T}|} = w^{\new}$}
%		\State $M(z_t)^{-1}, Q(z_t) , z_t \leftarrow$ \textsc{pm}.\textsc{update}$(x_t)$ \Comment{$v^{\mid} = z_0 \rightarrow z_1 \rightarrow z_2 \rightarrow \cdots \rightarrow z_{|{\cal T}|} = v^{\new}$}
%		\For{$s = 1 \to | {\cal S} | $}
%			\State $y_{s,t} \leftarrow z_t + {\cal S}_s \cdot (x_t - z_t)$ 
%		\EndFor
%	\EndFor
%	\State Sample random matrix $R \in \R^{r \times m}$ where $R_{i,j} \sim \mathcal{N}(0,1/r)$
	\For{$t$ \textbf{in} ${\cal T}$}
		\State $x_t = w^{\mid} + t \cdot ( w^{\new} - w^{\mid} )$ \Comment{$w^{\mid} = x_0 = x_{{\cal T}_1} \rightarrow x_{{\cal T}_2} \rightarrow \cdots \rightarrow x_{{\cal T}_{|{\cal T}|}} = w^{\new}$}
		\State $M(z_t)^{-1}, Q(z_t) , z_t \leftarrow$ \textsc{pm}.\textsc{update}$(x_t)$ \Comment{$v^{\mid} = z_0 = z_{{\cal T}_1} \rightarrow z_{{\cal T}_2} \rightarrow \cdots \rightarrow z_{{\cal T}_{|{\cal T}|}} = v^{\new}$}
		\State Compute $\theta_{i,t}^\top \theta_{i,t}$ for all $i \in [m]$ as in Lemma~\ref{lem:compute_theta_theta}
		\For{$s$ \textbf{in} ${ \cal S } $}
			\State $y_{s,t} \leftarrow z_t + s \cdot (x_t - z_t)$ 
			\State Sample independent $R_{\alpha \beta,s,t}, R_{\eta,s} \in \R^{r \times m}$ where each entry $\sim \mathcal{N}(0,1/r)$
			\State Compute $R_{\alpha \beta,s,t} \alpha_{i,s,t}$, $R_{\alpha \beta,s,t} \beta_{i,s,t}$, $R_{\eta,s} \eta_{i,s}$ for all $i \in [m]$ as in Lemma~\ref{lem:compute_alpha_beta_gamma_eta}
			\For {$s'$ \textbf{in} ${ \cal S } $}
				\State Sample $R_{\gamma,s,s',t} \in \R^{r \times m}$ where each entry $\sim \mathcal{N}(0,1/r)$
				\State Compute $R_{\gamma,s,s',t} \gamma_{i,s,t}$ and $R_{\gamma,s,s',t} \gamma_{i,s',t}$ for all $i \in [m]$ as in Lemma~\ref{lem:compute_alpha_beta_gamma_eta}
			\EndFor
		\EndFor
	\EndFor
	\State Compute $\Delta \wt{\sigma}^{\new}$ as in Lemma~\ref{lem:compute_sigma_change}
	\State $\wt{\sigma} \leftarrow \wt{\sigma} + \Delta \wt{\sigma}^{\mid} + \Delta \wt{\sigma}^{\new}$, $w \leftarrow w^{\new}$
	%, $v \leftarrow z_{|{\cal T}|}$, $M(v)^{-1} \leftarrow M(z_{|{\cal T}|})^{-1}$, $Q(v) \leftarrow Q(z_{|{\cal T}|})$
\EndProcedure
\State
\Procedure{\textsc{RefineEstimate}}{$\wt{\sigma}^{\new} \in \R_+^m$}   \Comment{Refine the estimation $\wt{\sigma}$ to $\wt{\sigma}^{\new}$}
	\State $\wt{\sigma} \leftarrow \wt{\sigma}^{\new}$
\EndProcedure
\State
\Procedure{\textsc{Query}}{$ $}
	\State \Return $\wt{\sigma}$
\EndProcedure
\State
\State {\bf end data structure}
\end{algorithmic}
\end{algorithm}

\subsection{Leverage score's moving}

\begin{definition}[$x_t$, $z_t$ and $y_{s,t}$]
Given vectors $w^{\mid}, w^{\new} \in \mathbb{R}_+^m$ and $v^{\mid} \in \mathbb{R}_+^m$. 
For all $t \in [0,1]$, we define $x_t$ to be $ w^{\mid} + t(w^{\new} - w^{\mid})$. 

To define $z_t$ for $t \in [0,1]$, we first define $z_t$ for all $t \in {\cal T}$ as in Algorithm~\ref{alg:maintain_leverage_score_complicated}.
We then extend the definition of $z_t$ to the entire interval $[0,1]$ by connecting consecutive points with segments.

For any $s \in [0,1]$, we define $y_{s,t} = z_t + s ( x_t - z_t )$. 
\end{definition}

\begin{definition}[$\eta$, $\theta$, $\alpha$, $\beta$ and $\gamma$]
\label{defn:eta_theta_alpha_beta_gamma}
We define
\begin{align*}
\eta_{i,s} = &~ \sqrt{Z_0 - X_0} \cdot Q(y_{s,0}) \cdot \sqrt{W^{\mid} - W^{\new}} \cdot e_i \\
\theta_{i,t} = & ~ \sqrt{W^{\mid} - W^{\new}} \cdot Q(z_t) \cdot \sqrt{ W^{\new} } \cdot e_i, \\
\alpha_{i,s,t} = & ~ \sqrt{  Z_t - X_t } \cdot Q(y_{s,t}) \cdot \sqrt{ W^{\new} } \cdot e_i , \\
\beta_{i,s,t} = & ~ \sqrt{ Z_t - X_t } \cdot Q(y_{s,t}) \cdot ( W^{\mid} - W^{\new} ) \cdot Q( z_t ) \cdot \sqrt{ W^{\new} } \cdot e_i , \\
\gamma_{i,s,t} = & ~ \sqrt{ W^{\mid} - W^{\new}} \cdot Q(y_{s,t}) \cdot (Z_t - X_t) \cdot Q(y_{s,t}) \cdot \sqrt{ W^{\new} } \cdot e_i. 
\end{align*}
\end{definition}

Here we make the simplifying assumption that $(Z_t - X_t)$ is PSD for every $t \in [0,1]$ and that $(W^{\mid} - W^{\new})$ is PSD.
In the case where these might not hold, we can simply break these matrices into positive and negative parts and bound both parts using essentially the same calculations. 

\begin{lemma}[Leverage score's moving] \label{lem:leverage_score_moving_complicated}
The change in leverage score can be written as
\begin{align}\label{eq:rewrite_sigma_into_two_terms}
\sigma(w^{\new} )_i - \sigma(w^{\mid})_i  
= & ~ \int_0^1 \theta_{i,t}^\top \theta_{i,t} \d t 
 + 2 \int_0^1 \int_0^1 \alpha_{i,s,t}^\top \beta_{i,s,t} \d s \d t + \int_0^1 \int_0^1 \int_0^1 \gamma_{i,s,t}^\top \gamma_{i,s',t} \d s \d s' \d t \nonumber  \\
 & ~  + \tau(v^{\mid})_i \cdot (w_i^{\new} - w^{\mid}_i) - \int_0^1 \eta_{i,s}^{\top} \eta_{i,s} \mathrm{d} s .
\end{align}
\end{lemma}
\begin{proof}

We first rewrite $\sigma(w^{\new})_i - \sigma(w^{\mid})_i$ into two terms
\begin{align*}
\sigma(w^{\new})_i - \sigma(w^{\mid})_i 
= & ~ w^{\new}_i \tau(w^{\new})_i - w^{\mid}_i \tau(w^{\mid})_i \notag \\
%= & ~ w^{\new}_i \tau(w^{\new})_i - w^{\new}_i \tau(w)_i + w^{\new}_i \tau(w)_i - w_i \tau(w)_i \notag \\
= & ~ w^{\new}_i \cdot ( \tau(w^{\new})_i - \tau(w^{\mid})_i ) + \tau(w^{\mid})_i \cdot ( w^{\new}_i - w^{\mid}_i ) .
\end{align*}
The lemma then follows immediately from the following Lemma~\ref{lem:first_term_leverage_score_moving_complicated} and~\ref{lem:second_term_leverage_score_moving_complicated}.
\end{proof}

\begin{lemma}[First term in leverage score's moving] \label{lem:first_term_leverage_score_moving_complicated}
For the first term in Lemma~\ref{lem:leverage_score_moving_complicated}, we have
\begin{align*}
w^{\new}_i \cdot ( \tau(w^{\new})_i - \tau(w^{\mid})_i )
= & ~ \int_0^1 \theta_{i,t}^\top \theta_{i,t} \d t 
 + 2 \int_0^1 \int_0^1 \alpha_{i,s,t}^\top \beta_{i,s,t} \d s \d t + \int_0^1 \int_0^1 \int_0^1 \gamma_{i,s,t}^\top \gamma_{i,s',t} \d s \d s' \d t .
\end{align*}
\end{lemma}

\begin{proof}

We can compute $\tau( w^{\new} )_i - \tau(w^{\mid})_i$ as
\begin{align*}
& ~ \tau( w^{\new} )_i - \tau(w^{\mid})_i \\
= & ~ (A (A^\top W^{\new} A )^{-1} A^\top )_{i,i} - ( A ( A^\top W^{\mid} A )^{-1} A^\top )_{i,i} \\
= & ~ \left[ \int_0^1 A (A^\top (W^{\mid} + t (W^{\new} - W^{\mid}) ) A )^{-1} A^\top  (W^{\mid} - W^{\new}) A (A^\top (W^{\mid} + t (W^{\new} - W^{\mid}) ) A )^{-1} A^\top \mathrm{d} t \right]_{i,i} \\
= & ~ \left[ \int_0^1 Q(x_t) (W^{\mid} - W^{\new}) Q(x_t) \mathrm{d} t \right]_{i,i} ,
\end{align*}
where the last step follows by the definition of $Q(v) = A (A^\top V A)^{-1} A^\top \in \R^{m \times m}$ and $x_t = w^{\mid} + t(w^{\new} - w^{\mid})$.
Now we write $Q(x_t ) - Q( z_t )$ in a similar way as follows:
\begin{align*}
  Q( x_t ) - Q( z_t ) 
= & ~ A (A^\top X_t A )^{-1} A^\top - A ( A^\top Z_t A )^{-1} A^\top \\
%= & ~ - ( A ( A^\top Z_t A )^{-1} A^\top - A (A^\top X_t A )^{-1} A^\top ) \\
%= & ~ \int_0^1 A ( A^\top (Z_t + s (X_t - Z_t) ) A )^{-1} A^\top (Z_t - X_t)  A ( A^\top (Z_t + s (X_t - Z_t) ) A )^{-1} A^\top \mathrm{d} s\\
= & ~ \int_0^1 Q(y_{s,t}) \cdot (Z_t - X_t) \cdot Q( y_{s,t} ) \mathrm{d} s ,
\end{align*}
where recall the definition of $y_{s,t} = z_t + s(x_t - z_t)$. 
Using the above expressions, we can rewrite $w^{\new}_i \cdot (\tau(w^{\new})_i - \tau(w^{\mid})_i)$ as follows:
\begin{align*}
& ~ w^{\new}_i \cdot (\tau(w^{\new})_i - \tau(w^{\mid})_i) \\
= & ~ w^{\new}_i \cdot \left[ \int_0^1 \left( Q(z_t) +  \int_0^1 Q(y_{s,t}) \cdot (Z_t - X_t) \cdot Q(y_{s,t}) \mathrm{d} s \right) \cdot ( W^{\mid} - W^{\new} ) \right. \\
& ~ \left. \cdot \left( Q(z_t) + \int_0^1 Q(y_{s',t}) \cdot (Z_t - X_t) \cdot Q(y_{s',t}) \mathrm{d} s' \right) \mathrm{d} t \right]_{i,i} \\
= & ~ \int_0^1 e_i^\top \sqrt{ W^{\new} } \cdot Q(z_t) \cdot ( W^{\mid} - W^{\new}) \cdot Q(z_t) \cdot \sqrt{ W^{\new} } e_i \d t \\
& ~ + 2 \int_0^1 e_i^\top \sqrt{ W^{\new} } \left( \int_0^1 Q(y_{s,t}) \cdot (Z_t - X_t) \cdot Q(y_{s,t}) \mathrm{d} s \right) ( W^{\mid} - W^{\new}) Q(z_t) \sqrt{ W^{\new} } e_i \mathrm{d} t \\
& ~ + \int_0^1 e_i^\top \sqrt{ W^{\new} } \left( \int_0^1 Q(y_{s,t}) \cdot (Z_t - X_t) \cdot Q(y_{s,t}) \mathrm{d} s \right) \cdot ( W^{\mid} - W^{\new}) \\
& ~ \cdot \left( \int_0^1 Q(y_{s',t}) \cdot (Z_t - X_t) \cdot Q(y_{s',t}) \mathrm{d} s' \right) \sqrt{ W^{\new} } e_i \mathrm{d} t \\
= & ~ \int_0^1 \theta_{i,t}^\top \theta_{i,t} \d t + 2 \int_0^1 \int_0^1 \alpha_{i,s,t}^\top \beta_{i,s,t} \d s \d t + \int_0^1 \int_0^1 \int_0^1 \gamma_{i,s,t}^\top \gamma_{i,s',t} \d s \d s' \d t .
\end{align*}

\end{proof}

\begin{lemma}[Second term in leverage score's moving] \label{lem:second_term_leverage_score_moving_complicated}
For the second term in Lemma~\ref{lem:leverage_score_moving_complicated}, we have
\begin{align*}
 &~ \tau(w^{\mid})_i \cdot (w_i^{\new} - w^{\mid}_i) =  \tau(v^{\mid})_i \cdot (w_i^{\new} - w^{\mid}_i) - \int_0^1 \eta_{i,s}^{\top} \eta_{i,s} \mathrm{d} s .
\end{align*}
\end{lemma}
\begin{proof}

We can rewrite $\tau(w^{\mid})_i \cdot (w_i^{\new} - w^{\mid}_i)$ as
\begin{align*}
\tau(w^{\mid})_i \cdot (w_i^{\new} - w^{\mid}_i)  
%= & ~ ( \tau(v)_i + ( \tau(w)_i - \tau(v)_i ) ) \cdot (w_i^{\new} - w_i) \\
= & ~ \tau(v^{\mid})_i \cdot (w_i^{\new} - w^{\mid}_i) + ( \tau(w^{\mid})_i - \tau(v^{\mid})_i ) \cdot (w_i^{\new} - w^{\mid}_i) .
\end{align*}
The second term can be computed as
\begin{align*}
 &~ (w_i^{\new} - w^{\mid}_i) \cdot (\tau(w^{\mid})_i - \tau(v^{\mid})_i)   \\
= &~ (w_i^{\new} - w^{\mid}_i) \cdot (\tau(x_0)_i - \tau(z_0)_i ) \\
= &~ (w_i^{\new} - w^{\mid}_i) \cdot \left[ \int_0^1 Q(z_0 + s(x_0- z_0)) \cdot (Z_0 - X_0) \cdot Q(z_0 + s(x_0 - z_0)) \mathrm{d} s \right]_{i,i} \\
= &~ -\int_0^1 e_i^{\top}  \cdot \sqrt{ W^{\mid} - W^{\new} } \cdot Q(y_{s,0}) \cdot (Z_0 - X_0) \cdot Q(y_{s,0}) \cdot \sqrt{ W^{\mid} - W^{\new} } \cdot e_i \mathrm{d} s \\
= & ~ - \int_0^1 \eta_{i,s}^{\top} \eta_{i,s} \mathrm{d} s .
\end{align*}
\end{proof}
 
\subsection{Running time anlaysis}

\begin{lemma}[Update time] \label{lem:update_time_complicated}
Assume the sequence $\mathsf{acts}$ satisfies Assumption~\ref{assumption:short_sequence_small_update} and $\epsilon_{\comp} \leq 0.01$. 
Then each call to the function $\textsc{Update}(\mathsf{acts})$ makes $O(N)$ calls to $\textsc{pm}.\textsc{update}$ and takes an extra $O(n^2 \cdot |\mathsf{acts}| + \T_{\mat}(n,n,r) \cdot N^3 \cdot \log(n))$ time.
\end{lemma}
\begin{proof}
The number of calls to $\textsc{pm}.\textsc{update}$ directly follows from Algorithm~\ref{alg:maintain_leverage_score_complicated}.
The time to compute $\Delta \wt{\sigma}^{\mid}$ is $O(n^2 \cdot |\mathsf{acts}|)$ by applying Theorem~\ref{thm:leverage_batch}. 
The remaining part of the running time is given by the following Lemma~\ref{lem:compute_theta_theta},~\ref{lem:compute_alpha_beta_gamma_eta} and~\ref{lem:compute_sigma_change}.
\end{proof}

\begin{lemma}[Computing $\theta^\top \theta$] \label{lem:compute_theta_theta}
Assume $m = O(n)$.
Then for any $t \in {\cal T}$, the time to compute $\theta_{i,t}^\top \theta_{i,t}$ for all $ i \in [m]$ is $O(n^2)$.
\end{lemma}
\begin{proof}
Recall from Definition~\ref{defn:eta_theta_alpha_beta_gamma} that
\begin{align*}
 \theta_{i,t} = \sqrt{ W^{\mid} - W^{\new} } \cdot Q(z_t) \cdot \sqrt{ W^{\new} } \cdot e_i.
\end{align*}
Denote $Q = Q(z_t)$ and $Q^{(2)}$ the entry-wise square of $Q$. 
We have
\begin{align*}
\theta_{i,t}^\top \theta_{i,t} 
= & ~ w_i^{\new} e_i^\top Q ( W^{\mid} - W^{\new} ) Q e_i \\
= & ~ w_i^{\new} \cdot \sum_{j \in [m]} (Q_{i,j})^2 ( w^{\mid} - w^{\new} )_j \\
= & ~ w_i^{\new} \cdot e_i^\top Q^{(2)} ( w^{\mid} - w^{\new} ) .
\end{align*}
Therefore, in order to compute $\theta_{i,t}^\top \theta_{i,t}$ for all $i \in [m]$, it suffices to compute the matrix-vector product $Q^{(2)} ( w^{\mid} - w^{\new} )$, which takes time $O(n^2)$. 

\end{proof}

\begin{lemma}[Computing $\alpha,\beta,\gamma,\eta$] \label{lem:compute_alpha_beta_gamma_eta}
Assume $m = O(n)$ and $\epsilon_{\comp} \leq 0.01$.
For any $r \geq 0$, $t \in {\cal T}$, $s, s' \in {\cal S}$ and matrices $R_{\alpha \beta,s,t}, R_{\gamma,s,s',t}, R_{\eta,s} \in \R^{r \times m}$, the time to compute $R_{\alpha \beta,s,t} \alpha_{i,s,t}$, $R_{\alpha \beta,s,t} \beta_{i,s,t}$, $R_{\gamma,s,s',t} \gamma_{i,s,t}$, $R_{\gamma,s,s',t} \gamma_{i,s',t}$, $R_{\eta,s} \eta_{i,s}$ for all $i \in [m]$ (up to negligible error) is $O(\T_{\mat}(n,n,r) \cdot \log(n))$ using fast matrix multiplication.
\end{lemma}
\begin{proof}
We only describe how to compute $R_{\alpha \beta,s,t} \alpha_{i,s,t}$ for all $i \in [m]$ in time $O(\T_{\mat}(n,n,r) \cdot \log(n))$ as the rest of the calculations are similar. 
Notice that that computing $R_{\alpha \beta,s,t} \alpha_{i,s,t}$ for all $i \in [m]$ is essentially computing the matrix
\begin{align*}
\{ R_{\alpha \beta,s,t} \alpha_{i,s,t} \}_{i \in [m]} =  R_{\alpha \beta,s,t} \sqrt{  Z_t - X_t } \cdot A \cdot M(y_{s,t})^{-1} \cdot A^\top \cdot \sqrt{ W^{\new} } .
\end{align*}
Assume we know the matrix $M(y_{s,t})^{-1}$, then we can perform the computation from left to right, and it follows that each matrix multiplication here can be done in time $O(\T_{\mat}(n,n,r))$.
To remove the assumption that we know $M(y_{s,t})^{-1}$, we pre-condition on the matrix $M(z_t)^{-1}$ and apply Lemma~\ref{lem:precondition} to compute $R_{\alpha \beta,s,t} \alpha_{i,s,t}$ for all $i \in [m]$ (up to negligible error) in time $O( \T_{\mat} (n,n,r) \cdot \log (n))$.

\end{proof}

\begin{lemma}[Computing $\Delta \wt{\sigma}^{\new}$] \label{lem:compute_sigma_change}
Assume $m = O(n)$ and $\epsilon_{\comp} \leq 0.01$. 
Define an approximation $\Delta \wt{\sigma}^{\new}_i$ to the change in leverage score to be
\begin{align*}
\Delta \wt{\sigma}^{\new}_i 
=&~ \tau(v^{\mid})_i \cdot (w_i^{\new} - w^{\mid}_i) + \sum_{t \in \mathcal{T}} \omega_t \theta_{i,t}^\top \theta_{i,t} - \sum_{s \in \mathcal{S}}  \omega_s \eta_{i,s}^{\top} R_{\eta,s}^\top R_{\eta,s} \eta_{i,s} \\
&~ \quad + 2 \sum_{t \in \mathcal{T}} \sum_{s \in \mathcal{S}} \omega_t \omega_s \alpha_{i,s,t}^\top R_{\alpha \beta,s,t}^\top R_{\alpha \beta,s,t} \beta_{i,s,t} + \sum_{t \in \mathcal{T}} \sum_{s \in \mathcal{S}} \sum_{s' \in \mathcal{S}} \omega_t \omega_s \omega_{s'} \gamma_{i,s,t}^\top R_{\gamma,s,s',t}^\top R_{\gamma,s,s',t} \gamma_{i,s', t}.
\end{align*}
Then the time to compute $\Delta \wt{\sigma}^{\new}_i$ for all $i \in [m]$ (up to negligible error) is $O( \T_{\mat} (n,n,r) \cdot N^3 \cdot \log(n))$.
\end{lemma}

\begin{proof}
Since we know $\tau(v^{\mid})_i = Q(v^{\mid})_{i,i}$ and $w_i^{\new} - w^{\mid}_i$, we can compute the first term for all $i \in [m]$ in $O(n)$ time.
By Lemma~\ref{lem:compute_theta_theta}, the time to compute the second term for all $i \in [m]$ is $O(n^2 \cdot N)$.
It follows from Lemma~\ref{lem:compute_alpha_beta_gamma_eta} that the time to compute the rest of the terms for all $i \in [m]$ is $O( \T_{\mat}(n,n,r) \cdot N^3 \cdot \log(n))$.
\end{proof}

\subsection{Upper bounding $\eta$, $\alpha$, $\beta$ and $\gamma$}
\label{subsec:upper_bound_eta_alpha_beta_gamma}

From the computation in Section~\ref{subsec:variance_bound_random_gaussian_matrix}, the variance of the three terms where we applied JL matrices are bounded by $\sum_i \norm{\eta_i}_2^4$, $\sum_i \norm{\alpha_i}_2^2 \norm{\beta_i}_2^2$ and $\sum_i \norm{\gamma_i}_2^4$.
Our bounds for these terms are summarized in Table~\ref{table:upper_bound_eta_alpha_beta_gamma}.

\begin{table}[htp!]  
\centering
\begin{tabular}{| l | l | l |}
	\hline
	{\bf Quantity} & {\bf Bound} & {\bf Lemma} \\ \hline
	$\sum_{i=1}^m \| \eta_i \|_2^4$ & $\epsilon_{\comp}^2 \cdot \norm{ \log(w^{\new}) - \log(w^{\mid}) }_2^2$ & Lemma~\ref{lem:upper_bound_sum_i_eta_i_eta_i} \\ \hline
	$\sum_{i=1}^m \| \alpha_i \|_2^2 \| \beta_i \|_2^2$ & $\epsilon_{\comp}^2 \cdot \norm{ \log(w^{\new}) - \log(w^{\mid}) }_2^2$ & Lemma~\ref{lem:upper_bound_sum_i_alpha_i_beta_i} \\ \hline
	$\sum_{i=1}^m \| \gamma_i \|_2^4$ & $\epsilon_{\comp}^4 \cdot \norm{ \log(w^{\new}) - \log(w^{\mid}) }_2^2$ & Lemma~\ref{lem:upper_bound_sum_i_gamma_i_gamma_i} \\ \hline
\end{tabular}
\caption{Upper bounds for $\sum_i \norm{\eta_i}_2^4$, $\sum_i \norm{\alpha_i}_2^2 \norm{\beta_i}_2^2$ and $\sum_i \norm{\gamma_i}_2^4$, where $\eta$, $\alpha$, $\beta$ and $\gamma$ are defined in Definition~\ref{defn:eta_theta_alpha_beta_gamma}.}
\label{table:upper_bound_eta_alpha_beta_gamma}
\end{table}

\subsubsection{Upper bounding $\sum_{i=1}^m \| \eta_i \|_2^4$}

\begin{lemma}\label{lem:upper_bound_sum_i_eta_i_eta_i}
Assume $\epsilon_{\comp} \leq 0.01$, where $\epsilon_{\comp}$ is the error parameter in Algorithm~\ref{alg:maintain_leverage_score_complicated}.
Then for any $s \in {\cal S}$, we have
\begin{align*}
\sum_{i=1}^m \| \eta_{i,s} \|_2^4 = O(\epsilon_{\comp}^2) \cdot \norm{ \log(w^{\new}) - \log(w^{\mid}) }_2^2 .
\end{align*}
\end{lemma}

\begin{proof}
We recall the definition of vector $\eta_{i,s} \in \R^m$ from Definition~\ref{defn:eta_theta_alpha_beta_gamma} as follows:
\begin{align*}
\eta_{i,s} = & ~ \sqrt{ Z_0 - X_0} \cdot Q(y_{s,0}) \cdot \sqrt{ W^{\mid} - W^{\new} } \cdot e_i,
\end{align*}
where $y_{s,t} = z_t + s ( x_t - z_t ) $ and $Y_{s,t} = Z_t + s (X_t - Z_t) $.
Recall from Definition~\ref{defn:operator_Q_and_P} that $P(v) = V^{1/2} A (A^\top V A)^{-1} A^\top V^{1/2} = V^{1/2} \cdot Q(v) \cdot V^{1/2}$. 
We can rewrite $\eta_{i,s}$ as
\begin{align*}
\eta_{i,s} = & ~ \frac{ \sqrt{ Z_0 - X_0} }{ \sqrt{Y_{s,0}} } P(y_{s,0}) \frac{ \sqrt{ W^{\mid} - W^{\new} } }{ \sqrt{Y_{s,0}} } e_i .
\end{align*}
Therefore, we have
\begin{align*}
 & ~ \sum_{i=1}^m \| \eta_{i,s} \|_2^4 \\
= & ~ \sum_{i=1}^m \left( e_i^\top \frac{ \sqrt{ W^{\mid} - W^{\new} } }{ \sqrt{Y_{s,0}} } P(y_{s,0}) \frac{ \sqrt{ Z_0 - X_0} }{ \sqrt{Y_{s,0}} } \cdot \frac{ \sqrt{ Z_0 - X_0} }{ \sqrt{Y_{s,0}} } P(y_{s,0}) \frac{ \sqrt{ W^{\mid} - W^{\new} } }{ \sqrt{Y_{s,0}} }  e_i \right)^2 \\
\leq & ~ \left\| \frac{ \sqrt{ W^{\mid} - W^{\new} } }{ \sqrt{Y_{s,0}} } P(y_{s,0}) \frac{ Z_0 - X_0 }{ Y_{s,0} } P(y_{s,0}) \frac{ \sqrt{ W^{\mid} - W^{\new} } }{ \sqrt{Y_{s,0}} } \right\|_F^2 \\
= & ~ \tr \left[ \frac{ \sqrt{ W^{\mid} - W^{\new} } }{ \sqrt{Y_{s,0}} } P(y_{s,0}) \frac{ Z_0 - X_0 }{ Y_{s,0} } P(y_{s,0}) \frac{ \sqrt{ W^{\mid} - W^{\new} } }{ \sqrt{Y_{s,0}} } \right. \\
& ~ \cdot \left. \frac{ \sqrt{ W^{\mid} - W^{\new} } }{ \sqrt{Y_{s,0}} } P(y_{s,0}) \frac{ Z_0 - X_0 }{ Y_{s,0} } P(y_{s,0}) \frac{ \sqrt{ W^{\mid} - W^{\new} } }{ \sqrt{Y_{s,0}} } \right] \\
= & ~ \tr[ (\sqrt{ B } P C P \sqrt{B} )^2 ] ,
\end{align*}
where in the last step we define diagonal matrices $B, C \in \R^{m \times m}$ as
\begin{align*}
B =   \frac{ W^{\mid} - W^{\new} }{ Y_{s,0} }  ,
C =   \frac{ Z_0 - X_0 }{ Y_{s,0} } ,
%D =   \frac{ W^{\new} }{ Y_{s,t} } .
\end{align*}
and for simplicity, we use $P$ to denote the projection matrix $P(y_{s,0})$.
Since $\epsilon_{\comp} \leq 0.01$, it follows that 
\begin{align*}
\norm{ \log(y_{s,0}) - \log(x_0) }_\infty \leq 0.01 .
\end{align*}
It follows that
\begin{align*}
 & ~ \tr[ (\sqrt{ B } P C P \sqrt{B} )^2 ] \\
\leq & ~ \| c \|_\infty^2 \cdot \tr[B^2] & P \cdot C \cdot P \preceq  \| c \|_\infty \cdot I\\
%= & ~ \| c \|_\infty^2 \cdot \| b \|_2^2  \\
= & ~  \| c \|_\infty^2 \cdot \norm{b}_2^2 \\
= & ~ O(1) \cdot \norm{ \log(z_0) - \log(x_0) }_{\infty}^2 \cdot \norm{ \log(w^{\new} ) - \log( w^{\mid} )}_2^2 \\
= & ~ O(\epsilon_{\comp}^2) \cdot \norm{ \log(w^{\new}) - \log(w^{\mid}) }_2^2 .
\end{align*}
\end{proof}

\subsubsection{Upper bounding $\sum_{i=1}^m \| \alpha_i \|_2^2 \| \beta_i \|_2^2$} 

%We prove the following upper bound on $\sum_{i=1}^m \| \alpha_{i,s,t} \|_2^2 \| \beta_{i,s,t} \|_2^2$.

\begin{lemma}\label{lem:upper_bound_sum_i_alpha_i_beta_i}
Assume $\epsilon_{\comp} \leq 0.01$ and $\norm{ \log (w^{\new}) - \log(w^{\mid}) }_\infty \leq 0.01$.
Then we have
\begin{align*}
\sum_{i=1}^m \| \alpha_{i,s,t} \|_2^2 \| \beta_{i,s,t} \|_2^2 = O(\epsilon_{\comp}^2) \cdot \norm{ \log(w^{\new}) - \log(w^{\mid}) }_2^2 .
\end{align*}
\end{lemma}
\begin{proof}
We recall the definition of vector $\alpha_{i,s,t} \in \R^m$ and $\beta_{i,s,t} \in \R^m$ from Definition~\ref{defn:eta_theta_alpha_beta_gamma} as follows:
\begin{align*}
\alpha_{i,s,t} = & ~ \sqrt{  Z_t - X_t } \cdot Q(y_{s,t}) \cdot \sqrt{ W^{\new} } \cdot e_i , \\
\beta_{i,s,t} = & ~ \sqrt{ Z_t - X_t } \cdot Q(y_{s,t}) \cdot ( W^{\mid} - W^{\new} ) \cdot Q( z_t )\cdot \sqrt{ W^{\new} } \cdot e_i,
\end{align*}
where recall that $y_{s,t} = z_t + s ( x_t - z_t ) $ and $Y_{s,t} = Z_t + s (X_t - Z_t)$. Recall from Definition~\ref{defn:operator_Q_and_P} that $P(v) = V^{1/2} A (A^\top V A)^{-1} A^\top V^{1/2} = V^{1/2} \cdot Q(v) \cdot V^{1/2}$. 
We therefore can rewrite $\alpha_{i,s,t}$ and $\beta_{i,s,t}$ as
\begin{align*}
\alpha_{i,s,t} = & ~ \frac{ \sqrt{ Z_t - X_t } }{ \sqrt{ Y_{s,t} } } P(y_{s,t})  \frac{ \sqrt{ W^{\new} } }{ \sqrt{ Y_{s,t} } }  e_i , \\
\beta_{i,s,t} = & ~ \frac{ \sqrt{ Z_t - X_t } }{ \sqrt{ Y_{s,t} } } P(y_{s,t}) \frac{ W^{\mid} - W^{\new} }{ \sqrt{Z_t Y_{s,t}} }  P(z_t) \frac{ \sqrt{W^{\new}} } { \sqrt{Z_t} } e_i .
\end{align*}
Our assumptions $\epsilon_{\comp} \leq 0.01$ and $\norm{ \log (w^{\new}) - \log(w^{\mid}) }_\infty \leq 0.01$ imply that 
\begin{align*}
\norm{ \log(y_{s,t}) - \log(w^{\new}) }_\infty \leq 0.1 \qquad \text{and} \qquad \norm{ \log(z_t) - \log(w^{\new}) }_\infty \leq 0.1 .
\end{align*}
Therefore we can upper bound $\alpha_{i,s,t}$ as
\begin{align*}
  \max_{i,s,t} \| \alpha_{i,s,t} \|_2^2
= & ~ \max_{i,s,t} \left\| \frac{ \sqrt{ Z_t - X_t } }{ \sqrt{ Y_{s,t} } } P(y_{s,t}) \frac{ \sqrt{W^{\new}} }{ \sqrt{Y_{s,t}} } e_i \right\|_2^2 \\
\leq & ~ \max_{i,s,t} \norm{y_{s,t}^{-1} (z_t - x_t)}_\infty \cdot \norm{\frac{ \sqrt{W^{\new}} }{ \sqrt{Y_{s,t}} } e_i }_2^2 \\
= & ~ O(\epsilon_{\comp}).
\end{align*}
For $\beta_{i,s,t}$, we have
\begin{align*}
& ~ \sum_{i=1}^m \| \beta_{i,s,t} \|_2^2 \\
= & ~ \sum_{i=1}^m e_i^\top \frac{ \sqrt{W^{\new}} }{ \sqrt{Z_t} } P(z_t) \frac{W^{\mid} - W^{\new}}{ \sqrt{Z_t Y_{s,t}} } P(y_{s,t}) \frac{ Z_t-X_t }{ Y_{s,t} } P(y_{s,t}) \frac{ W^{\mid} - W^{\new}}{ \sqrt{Z_t Y_{s,t}} }  P(z_t) \frac{ \sqrt{W^{\new}} } { \sqrt{Z_t} } e_i \\
= & ~ \tr \left[ \frac{ \sqrt{W^{\new}} }{ \sqrt{Z_t} } P(z_t) \frac{W^{\mid} - W^{\new}}{ \sqrt{Z_t Y_{s,t}} } P(y_{s,t}) \frac{ Z_t-X_t }{ Y_{s,t} } P(y_{s,t}) \frac{ W^{\mid} - W^{\new}}{ \sqrt{Z_t Y_{s,t}} }  P(z_t) \frac{ \sqrt{W^{\new}} } { \sqrt{Z_t} } \right] \\
= & ~ \tr[ \sqrt{D} P(z_t) B P(y_{s,t}) C P(y_{s,t}) B P(z_t) \sqrt{D} ] ,
\end{align*} 
where in the last step we define diagonal matrices
\begin{align*}
B =  \frac{W^{\mid} - W^{\new}}{ \sqrt{ Z_t Y_{s,t} } } , C =  \frac{Z_t - X_t}{ Y_{s,t} } , D =  \frac{W^{\new}}{Z_t} .
\end{align*}
Thus, we have
\begin{align*}
 & ~ \tr[ \sqrt{D} P(z_t) B P(y_{s,t}) C P(y_{s,t}) B P(z_t) \sqrt{D} ]  \\
\leq & ~ \| c \|_\infty \cdot \tr[ \sqrt{D} P(z_t) B^2 P(z_t) \sqrt{D} ] & P(y_{s,t}) \cdot C \cdot P(y_{s,t}) \preceq \| c \|_\infty \cdot I \\
= & ~ \| c \|_\infty \cdot \tr[  B P(z_t) D P(z_t) B ] \\
= & ~ O(1) \cdot \| c \|_\infty \cdot \tr[  B^2 ] & P(z_t) \cdot D \cdot P(z_t) \preceq O(1) \cdot I \\
= & ~ O(1) \cdot \| c \|_\infty  \cdot \norm{b}_2^2 \\
= & ~ O(\epsilon_{\comp}) \cdot \norm{ \log( w^{\new} ) - \log( w^{\mid} )}_2^2 . 
\end{align*}

\end{proof}

\subsubsection{Upper bounding $\sum_{i=1}^m \| \gamma_i \|_2^4$}  

\begin{lemma}\label{lem:upper_bound_sum_i_gamma_i_gamma_i}
Assume $\epsilon_{\comp} \leq 0.01$ and $\norm{ \log (w^{\new}) - \log(w^{\mid}) }_\infty \leq 0.01$.
Then for any $s \in {\cal S}$ and $t\in {\cal T}$, we have
\begin{align*}
\sum_{i=1}^m \| \gamma_{i,s,t}  \|_2^4 = O(\epsilon_{\comp}^4) \cdot \norm{\log(w^{\new}) - \log(w^{\mid})}_2^2 .
\end{align*}
\end{lemma}

\begin{proof}
We recall the definition of vector $\gamma_{i,s,t} \in \R^m$ from Definition~\ref{defn:eta_theta_alpha_beta_gamma} as follows:
\begin{align*}
\gamma_{i,s,t} = & ~ \sqrt{ W^{\mid} - W^{\new}} \cdot Q(y_{s,t}) \cdot (Z_t - X_t) \cdot Q(y_{s,t}) \cdot \sqrt{ W^{\new} } \cdot e_i,
\end{align*}
where recall that $y_{s,t} = ( z_t + s ( x_t - z_t ) )$ and $Y_{s,t} = ( Z_t + s (X_t - Z_t) )$.
Recall from Definition~\ref{defn:operator_Q_and_P} that $P(v) = V^{1/2} A (A^\top V A)^{-1} A^\top V^{1/2} = V^{1/2} \cdot Q(v) \cdot V^{1/2}$. 
We therefore can rewrite $\gamma_{i,s,t}$ as
\begin{align*}
\gamma_{i,s,t} = & ~ \frac{ \sqrt{ W^{\mid} - W^{\new}} }{ \sqrt{ Y_{s,t} } } \cdot P(y_{s,t}) \cdot \frac{ Z_t - X_t }{ Y_{s,t} } \cdot P(y_{s,t}) \cdot \frac{ \sqrt{W^{\new}} }{ \sqrt{ Y_{s,t}  } } \cdot e_i .
\end{align*}
It follows that 
\begin{align*}
 & ~ \sum_{i=1}^m \| \gamma_{i,s,t} \|_2^4 \\
= & ~ \sum_{i=1}^m \left( e_i^\top \frac{ \sqrt{W^{\new}} }{ \sqrt{ Y_{s,t} } } P(y_{s,t}) \frac{Z_t - X_t}{ Y_{s,t} }P(y_{s,t}) \frac{ W^{\mid} - W^{\new} }{ Y_{s,t} } P(y_{s,t}) \frac{Z_t-X_t}{Y_{s,t}} P(y_{s,t}) \frac{ \sqrt{W^{\new}} }{ \sqrt{ Y_{s,t} } } e_i \right)^2 \\
\leq & ~ \left\| \frac{ \sqrt{W^{\new}} } { \sqrt{ Y_{s,t} } } P(y_{s,t}) \frac{Z_t - X_t}{ Y_{s,t} }P(y_{s,t}) \frac{ W^{\mid} - W^{\new} }{ Y_{s,t} } P(y_{s,t}) \frac{Z_t-X_t}{Y_{s,t}} P(y_{s,t}) \frac{ \sqrt{W^{\new}} }{ \sqrt{ Y_{s,t} } } \right\|_F^2 \\
= & ~ \tr \left[ \frac{ \sqrt{W^{\new}} }{ \sqrt{ Y_{s,t} } } P(y_{s,t}) \frac{Z_t - X_t}{ Y_{s,t}} P(y_{s,t}) \frac{ W^{\mid} - W^{\new}}{Y_{s,t}} P(y_{s,t}) \frac{Z_t - X_t}{Y_{s,t}} P(y_{s,t}) \frac{ \sqrt{W^{\new}} }{ \sqrt{ Y_{s,t} } } \right. \\
& ~ \cdot \left. \frac{ \sqrt{W^{\new}} }{ \sqrt{ Y_{s,t} } } P(y_{s,t}) \frac{Z_t - X_t}{ Y_{s,t}} P(y_{s,t}) \frac{ W^{\mid} - W^{\new}}{Y_{s,t}} P(y_{s,t}) \frac{Z_t - X_t}{Y_{s,t}} P(y_{s,t}) \frac{ \sqrt{W^{\new}} }{ \sqrt{ Y_{s,t} } } \right] \\
= & ~ \tr[ ( \sqrt{ D } P C P B P C P \sqrt{ D } )^2 ].
\end{align*}
where in the last step we define diagonal matrices $B, C, D \in \R^{m \times m}$, 
\begin{align*}
B =   \frac{ W^{\mid} - W^{\new}}{ Y_{s,t} }    ,
C =   \frac{ Z_t - X_t }{ Y_{s,t} }  ,
D =   \frac{ W^{\new} }{ Y_{s,t} } .
\end{align*}
and for simplicity, we use $P$ to denote $P(y_{s,t})$.
By our assumptions that $\epsilon_{\simp} \leq 0.01$ and 
\begin{align*}
\norm{\log(w^{\new}) - \log(w^{\mid})}_\infty \leq 0.01 ,
\end{align*}
we have 
\begin{align*}
\norm{ \log(y_{s,t}) - \log(w^{\new}) }_\infty \leq 0.1 .
\end{align*}
Therefore we have

\begin{align*}
  & ~ \tr[ ( \sqrt{ D } P C P B P C P \sqrt{ D } )^2 ] \\
= & ~ \tr[ ( \sqrt{B} P C P D P C P \sqrt{B} )^2 ] \\
\leq & ~ O(1) \cdot \tr[ ( \sqrt{B} P C^2 P \sqrt{B} )^2 ] & P \cdot D \cdot P \preceq O(1) \cdot I \\
\leq & ~ O(1) \cdot \| c \|_\infty^4 \cdot \tr[ B^2 ] & P \cdot C^2 \cdot P \preceq \| c \|_\infty^2 \cdot I \\
= & ~ O(1) \cdot \| c \|_\infty^4  \cdot \| b \|_2^2 \\
= & ~ O(\epsilon_{\comp}^4) \cdot \norm{\log(w^{\new}) - \log(w^{\mid})}_2^2  .
\end{align*}

\end{proof}

\subsection{Variance upper bound for random Gaussian matrices}
\label{subsec:variance_bound_random_gaussian_matrix}

\begin{table}[htp!]
\centering
\begin{tabular}{ | l | l | l | l |}
\hline
{\bf Notation} & {\bf Statement} & {\bf Quantity} & {\bf Upper bound} \\ \hline
$\eta^\top \eta$ & Lemma~\ref{lem:variance_bound_for_eta_eta_jl_to_discrete} & $ \Var[\sum \eta^\top \eta - \sum \eta^\top R^\top R \eta]$ & $ O(1/r) \cdot \sum \| \eta_{i,s} \|_2^4$ \\ \hline 
%$\theta^\top \theta$ & Lemma~\ref{lem:variance_bound_for_theta_theta_jl_to_discrete} & JL matrices & $| \sum \theta^\top \theta - \sum \theta^\top R^\top R \theta |$ \\ \hline 
$\alpha^\top \beta$ & Lemma~\ref{lem:variance_bound_for_alpha_beta_jl_to_discrete} & $ \Var[\sum \alpha^\top \beta - \sum \beta^\top R^\top R \alpha]$ & $O(1/r) \cdot \sum \| \beta_{i,s,t} \|_2^2  \| \alpha_{i,s,t} \|_2^2$ \\ \hline 
$\gamma^\top \gamma$ & Lemma~\ref{lem:variance_bound_for_gamma_gamma_jl_to_discrete} & $ \Var[\sum \gamma^\top \gamma - \sum \gamma^\top R^\top R \gamma ]$ & $O(1/r) \cdot \sum \| \gamma_{i,s,t} \|_2^4 $ \\ \hline
\end{tabular}
\caption{Variance upper bounds for $\eta^\top \eta$, $\alpha^\top \beta$ and $\gamma^\top \gamma$, where $\eta$, $\alpha$, $\beta$ and $\gamma$ are defined in Definition~\ref{defn:eta_theta_alpha_beta_gamma}.}
\label{table:variance_bounds_complicated}
\end{table}

Our variance bounds for the terms where we applied JL matrices are summarized in Table~\ref{table:variance_bounds_complicated}.
These results are consequences of the following Lemma~\ref{lem:variance_bound_general}.

\subsubsection{A variance bound for random Gaussian matrices}

\begin{lemma}[JL variance bound]\label{lem:variance_bound_general}
Given vectors $x, y \in \R^m$. Let $R \in \R^{r \times m}$ be a random Gaussian matrix  with $R_{i,j} \sim {\cal N}(0, 1/r) $. Then we have
\begin{align*}
\E_{R} [ x^\top R^\top R y ] = x^\top y , ~~~~~ \Var_{R} [ x^\top R^\top R y ] \leq \frac{3}{r} \cdot \norm{x}_2^2 \norm{y}_2^2.
\end{align*}
\end{lemma}

\begin{proof}
Since $\E[R^\top R] = I$, we have
\begin{align*}
\E_{R} [ x^\top R^\top R y ] = x^\top y .
\end{align*} 
Next we prove the bound on the variance. We use $R_i$ for each $i \in [r]$ to denote the column vector that corresponds to the $i$th row of $R$. We have

\begin{align*}
\Var_{R} [ x^\top R^\top R y ] 
= & ~ \Var_{R} \left[\sum_{i \in [r]} x^\top R_i R_i^\top y \right] \\
= & ~ r \Var_{R_1} \left[ x^\top R_1 R_1^\top y \right] \\
\leq & ~ r \E_{R_1} [ ( x^\top R_1 R_1^\top y )^2 ] \\
\leq & ~ r \left( \E_{R_1} [ (x^\top R_1)^4 ] \right)^{1/2} \cdot \left( \E_{R_1} [ (y^\top R_1)^4 ] \right)^{1/2} \\
\leq & ~ \frac{3}{r} \cdot \norm{x}_2^2 \norm{y}_2^2 ,
\end{align*}
where the last inequality is because $x^\top R_1$ is a Gaussian random variable with variance $\norm{x}_2^2 / r $.

\end{proof}

\subsubsection{Variance bound for $\eta^\top \eta$, $\alpha^\top \beta$ and $\gamma^\top \gamma$}
The lemmas below all follow immediately from Lemma~\ref{lem:variance_bound_general}.

% Variance bound for eta^\top eta
\begin{lemma}[Variance bound for $\eta^\top \eta$]\label{lem:variance_bound_for_eta_eta_jl_to_discrete}
Let $\sigma_{i,\dis}$ and $\sigma_{i,\jl}$ be defined as follows
\begin{align*}
\sigma_{i,\dis} = \sum_{s \in {\cal S}} \omega_s \eta_{i,s}^\top \eta_{i,s}, ~~~ \sigma_{i,\jl} = \sum_{s \in {\cal S}} \omega_s \eta_{i,s}^\top R_{\eta,s}^\top R_{\eta,s} \eta_{i,s}.
\end{align*}
where $R_{\eta,s} \in \R^{r \times m}$ for $s \in {\cal S}$ are independent random Gaussian matrices with $(R_{\eta,s})_{i,j} \sim {\cal N}(0, 1/r) $. Then, we have
\begin{align*}
\E_{R} [ \sigma_{i,\jl} ] = \sigma_{i,\dis} , ~~~~~ \Var_{R} [ \sigma_{i,\jl} ] \leq \frac{3}{r} \cdot \sum_{s \in {\cal S}} \omega_s^2 \| \eta_{i,s} \|_2^4 .
\end{align*}
\end{lemma}

% Variance bound for beta^\top alpha

\begin{lemma}[Variance bound for $\alpha^\top \beta$]\label{lem:variance_bound_for_alpha_beta_jl_to_discrete}
Let $\sigma_{i,\dis}$ and $\sigma_{i,\jl}$ be defined as follows
\begin{align*}
\sigma_{i,\dis} = \sum_{t \in {\cal T}} \sum_{s \in {\cal S}} \omega_t \omega_s \beta_{i,s,t}^\top \alpha_{i,s,t}, ~~~ \sigma_{i,\jl} = \sum_{t \in {\cal T}} \sum_{s \in {\cal S}} \omega_t \omega_s \beta_{i,s,t}^\top R_{\alpha \beta,s,t}^\top R_{\alpha \beta,s,t} \alpha_{i,s,t}.
\end{align*}
where $R_{\alpha \beta,s,t} \in \R^{r \times m}$ for $s \in {\cal S}, t \in {\cal T}$ are independent random Gaussian matrices with $(R_{\alpha \beta,s,t})_{i,j} \sim {\cal N}(0, 1/r) $.
Then, we have
\begin{align*}
\E_{R} [ \sigma_{i,\jl} ] = \sigma_{i,\dis} , ~~~~~ \Var_{R} [ \sigma_{i,\jl} ] \leq \frac{3}{r} \sum_{t \in {\cal T}} \sum_{s \in {\cal S}} \omega_t^2 \omega_s^2 \| \beta_{i,s,t} \|_2^2  \| \alpha_{i,s,t} \|_2^2 .
\end{align*}
\end{lemma}

% Variance bound for gamma^\top gamma

\begin{lemma}[Variance bound for $\gamma^\top \gamma$]\label{lem:variance_bound_for_gamma_gamma_jl_to_discrete}
Let $\sigma_{i,\dis}$ and $\sigma_{i,\jl}$ be defined as follows
\begin{align*}
\sigma_{i,\dis} = \sum_{t\in {\cal T}} \sum_{s \in {\cal S}} \sum_{s' \in {\cal S}} \omega_t \omega_s \omega_{s'} \gamma_{i,s,t}^\top \gamma_{i,s',t}, ~~~ \sigma_{i,\jl} = \sum_{t\in {\cal T}} \sum_{s \in {\cal S}} \sum_{s' \in {\cal S}} \omega_t \omega_s \omega_{s'} \gamma_{i,s,t}^\top R_{\gamma,s,s',t}^\top R_{\gamma,s,s',t} \gamma_{i,s',t}.
\end{align*}
where $R_{\gamma,s,s',t} \in \R^{r \times m}$ are independent random Gaussian matrices with $(R_{\gamma,s,s',t})_{i,j} \sim {\cal N}(0, 1/r) $.
Then, we have
\begin{align*}
\E_{R} [ \sigma_{i,\jl} ] = \sigma_{i,\dis} , ~~~~~ \Var_{R} [ \sigma_{i,\jl} ] \leq \frac{3}{r} \sum_{t \in {\cal T}} \sum_{s \in {\cal S}} \sum_{s' \in {\cal S}} \omega_t^2 \omega_s^2 \omega_{s'}^2 \| \gamma_{i,s,t} \|_2^4 .
\end{align*}
\end{lemma}

\subsection{Error upper bound for discrete sampling} 
\label{subsec:error_bound_discrete_to_continuous}

\begin{table}[htp!]
\centering
\begin{tabular}{ | l | l | l | l |}
\hline
{\bf Notation} & {\bf Statement} & {\bf Quantity} & {\bf Upper bound} \\ \hline
$\eta^\top \eta$ & Lemma~\ref{lem:variance_bound_for_eta_eta_discrete_to_continuous} & $| \int \eta^\top \eta - \sum \eta^\top \eta |$ & $O(\poly(n) / 2^{2N})$ \\ \hline
$\theta^\top \theta$ & Lemma~\ref{lem:variance_bound_for_theta_theta_discrete_to_continuous} & $| \int \theta^\top \theta - \sum \theta^\top \theta |$ & $O(\poly(n) / 2^{2N})$ \\ \hline
$\alpha^\top \beta$ & Lemma~\ref{lem:variance_bound_for_beta_alpha_discrete_to_continuous} & $| \int \alpha^\top \beta - \sum \alpha^\top \beta |$ & $O(\poly(n) / 2^{2N})$ \\ \hline
$\gamma^\top \gamma$ & Lemma~\ref{lem:variance_bound_for_gamma_gamma_discrete_to_continuous} & $| \int \gamma^\top \gamma - \sum \gamma^\top \gamma |$ & $O(\poly(n) / 2^{2N})$ \\ \hline
\end{tabular}
\caption{Error upper bounds for discrete sampling. Here, $\eta$, $\alpha$, $\beta$ and $\gamma$ are defined as in Definition~\ref{defn:eta_theta_alpha_beta_gamma}.}
\label{table:discrete_sampling_bounds_complicated}
\end{table}
Our error upper bounds for approximating integral by discrete sampling are summarized in Table~\ref{table:discrete_sampling_bounds_complicated}.
We only prove Lemma~\ref{lem:variance_bound_for_eta_eta_discrete_to_continuous} in the following. The rest of the lemmas follows from similar arguments.

\begin{lemma}[Error upper bound for $\eta^\top \eta$]\label{lem:variance_bound_for_eta_eta_discrete_to_continuous}
Assume $\epsilon_{\comp} \leq 0.01$.
Let $\sigma_{i,\cts}$ and $\sigma_{i,\dis}$ be defined as follows
\begin{align*}
\sigma_{i,\cts} = \int_0^1 \eta_{i,s}^\top \eta_{i,s} \d s, ~~~ \sigma_{i,\dis} = \sum_{s \in {\cal S}} \omega_s \eta_{i,s}^\top \eta_{i,s} ,
\end{align*}
where $|{\cal S}| = N$.
Then we have
\begin{align*}
| \sigma_{i,\dis} - \sigma_{i,\cts} | \leq O\left(\poly(n) / 2^{2 N } \right).
%\E_{{\cal T}, {\cal S}} [ \sigma_{i,\dis} ] = \sigma_{i,\cts}, ~~~ \E_{{\cal T}, {\cal S}} [ \sigma_{i,\dis}^2 ] \leq
\end{align*}
\end{lemma}

\begin{proof}

Define function $f: [0,1] \rightarrow \R$ to be $f(s) = \eta_{i,s}^\top \eta_{i,s}$.
Applying Theorem~\ref{thm:one_variable_integral} with $N = |{\cal S}|$, the error is bounded as
\begin{align*}
| \sigma_{i,\dis} - \sigma_{i,\cts} | \leq \frac{M_{2N}}{ (2N) ! \cdot 2^{2N} } ,
\end{align*}
where 
\begin{align*}
M_{2N} = \max_{s \in [0,1] }  \frac{ \partial^{(2N)} f } { \partial s^{ ( 2N ) } } ( s )  .
\end{align*}
In order to bound $M_{2N}$, we need the following Cauchy's estimates.

\begin{theorem}[Cauchy's Estimates] \label{thm:cauchy_estimates}
Suppose $f$ is holomorphic on a neighborhood of the ball $B = \{ z \in \C : | z - z_0 | \leq r \}$, then we have that 
\begin{align*}
|f^{(k)} (z_0)| \leq \frac{k !}{r^k} \cdot \sup_{z \in B} |f(z)|. 
\end{align*}
\end{theorem}

Since $f(x)$ is a rational polynomial, we can extend the definition of $f(s)$ to the complex plane and the resulting function, which we also denote as $f(s)$. 
Since $\norm{\log(z_0) - \log(x_0)}_\infty \leq \epsilon_{\comp} \leq 0.01$, $M(y_{s,0})$ is invertible for $|s| \leq 1$. 
Hence $f$ is holomorphic on the unit ball on the complex plane around $0$. 
Applying Theorem~\ref{thm:cauchy_estimates} with $r = 1$, we have that 
\begin{align*}
M_{2N} \leq (2N ) ! \cdot \poly(n) .
\end{align*}
Therefore, the error is bounded as
\begin{align*}
| \sigma_{i,\dis} - \sigma_{i,\cts} | \leq \frac{\poly(n)}{ 2^{2N} } ,
\end{align*}
which finishes the proof of the lemma.

\end{proof}

\begin{lemma}[Error upper bound for $\theta^\top \theta$]\label{lem:variance_bound_for_theta_theta_discrete_to_continuous}
Assume $\epsilon_{\comp} \leq 0.01$.
Let $\sigma_{i,\cts}$ and $\sigma_{i,\dis}$ be defined as follows
\begin{align*}
\sigma_{i,\cts} = \int_0^1 \theta_{i,t}^\top \theta_{i,t} \d t, ~~~ \sigma_{i,\dis} = \sum_{t \in {\cal T}} \omega_t \theta_{i,t}^\top \theta_{i,t} ,
\end{align*}
where $|{\cal T}| = N$.
Then we have
\begin{align*}
| \sigma_{i,\dis} - \sigma_{i,\cts} | \leq O\left(\poly(n) / 2^{2 N } \right) .
%\E_{{\cal T}, {\cal S}} [ \sigma_{i,\dis} ] = \sigma_{i,\cts}, ~~~ \E_{{\cal T}, {\cal S}} [ \sigma_{i,\dis}^2 ] \leq
\end{align*}
\end{lemma}

\begin{lemma}[Error upper bound for $\alpha^\top \beta$]\label{lem:variance_bound_for_beta_alpha_discrete_to_continuous}
Assume $\epsilon_{\comp} \leq 0.01$.
Let $\sigma_{i,\cts}$ and $\sigma_{i,\dis}$ be defined as follows
\begin{align*}
\sigma_{i,\cts} = \int_0^1 \int_0^1 \beta_{i,s,t}^\top \alpha_{i,s,t} \d s \d t, ~~~ \sigma_{i,\dis} = \sum_{t \in {\cal T}} \sum_{s \in {\cal S}} \omega_t \omega_s \beta_{i,s,t}^\top \alpha_{i,s,t} ,
\end{align*}
where $|{\cal S}| = |{\cal T}| = N$.
Then we have
\begin{align*}
| \sigma_{i,\dis} - \sigma_{i,\cts} | \leq O\left(\poly(n) / 2^{2 N } \right).
\end{align*}
\end{lemma}

\begin{lemma}[Error upper bound for $\gamma^\top \gamma$]\label{lem:variance_bound_for_gamma_gamma_discrete_to_continuous}
Assume $\epsilon_{\comp} \leq 0.01$.
Let $\sigma_{i,\cts}$ and $\sigma_{i,\dis}$ be defined as follows
\begin{align*}
\sigma_{i,\cts} = & ~ \int_0^1 \int_0^1 \int_0^1 \gamma_{i,s,t}^\top \gamma_{i,s',t} \d s \d s' \d t \\
\sigma_{i,\dis} = & ~ \sum_{t \in {\cal T}} \sum_{s \in {\cal S}} \sum_{s' \in {\cal S}} \omega_t \omega_s \omega_{s'} \gamma_{i,s,t}^\top \gamma_{i,s',t} ,
\end{align*}
where $|{\cal S}| = |{\cal T}| = N$.
Then we have
\begin{align*}
| \sigma_{i,\dis} - \sigma_{i,\cts} | \leq O\left(\poly(n) / 2^{2 N } \right).
%\E_{ {\cal T} , {\cal S} } [ \sigma_{i,\dis} ] = \sigma_{i,\cts}, ~~~ \E_{{\cal T}, {\cal S}} [ \sigma_{i,\dis}^2 ] \leq 
\end{align*}
\end{lemma}

 %%% Section 8
\newpage 
\appendix
\section*{Appendix}

\section{Perturbed Volumetric Center Cutting Plane Method}\label{sec:perturb}

In this section we present an overview of Vaidya's cutting plane method \cite{v89}
and illustrate how our leverage score maintenance data structure in Section~\ref{sec:main_leverage_score} implies
a faster implementation in $O(n\SO\log (\kappa)+n^{3}\log (\kappa))$ time.
%In this section we prove that Vaidya's cutting plane method remains valid in the presence of small changes to leverage scores. 
Formally, we prove the following Theorem~\ref{thm:mainvaidya} from Section~\ref{sec:vaidya}.

\mainvaidya*

%\begin{theorem*}
%We can implement Vaidya's cutting plane method in $O((n\SO+n^{3})\log (\kappa))$ time.
%\end{theorem*}

Our cutting plane method essentially replaces the leverage scores $\sigma$ in Vaidya's method by estimates $\tilde{\sigma}$ from our leverage score maintenance data structure in Section~\ref{sec:main_leverage_score}, which satisfies that $\|\tilde{\sigma}-\sigma\|_{2}\leq1/\log^{O(1)} (n)$. 
To justify the validility of such a replacement, we first give an overview of Vaidya's method. 

\paragraph*{Feasible region}

Vaidya's algorithm maintains a polytope $P^{(k)}=\{x \in \R^n :A^{(k)}x\geq b^{(k)}\} \subset \R^n$ with $m^{(k)} = O(n)$ constraints as the feasible region.

\paragraph*{Volmetric barrier function and volumetric center}

Vaidya's algorithm maintains an approximate minimizer $z^{(k)} \in \R^n$, known as the volumetric center,
of the volumetric barrier function: 
\begin{align*}
z^{(k)}\approx \arg\min_{x\in P^{(k)}}F^{(k)}(x)\qquad\text{with}\qquad F^{(k)}(x):=\frac{1}{2}\log(\det(A^{(k)\top}(S_{x}^{(k)})^{-2}A^{(k)}))
\end{align*}
where $s_{x}^{(k)}=A^{(k)}x-b^{(k)} \in \R^{m^{(k)}}$ is the slack and $S_{x}^{(k)} \in \R^{ m^{(k)} \times m^{(k)} }$
is the diagonal matrix for $s_{x}^{(k)}$. Here $z^{(k)} \in \R^n$
serves as the query point to the separation oracle.

\paragraph*{Leverage score}

Each constraint $i$ of $P^{(k)}$ is associated with a leverage score
\begin{align*}
\sigma_{i}(z)=\left( (S_{x}^{(k)})^{-1}A^{(k)}(A^{(k)\top}(S_{x}^{(k)})^{-2}A^{(k)})^{-1}A^{(k)\top}(S_{x}^{(k)})^{-1} \right)_{i,i}, \forall i \in [m^{(k)}]
\end{align*}
which measures its relative importance (see preliminary for the definition).
It is well-known that $0 \leq \sigma_{i} (z) \leq 1$, $\forall i \in [m^{(k)}]$ and $\sum_{i=1}^{ m^{(k)} } \sigma_{i}(z)=n$.
We denote by $\sigma$ and $\Sigma$ the vector and
 diagonal matrix of leverage scores respectively. 

\paragraph*{Updating $P^{(k)}$ according to leverage scores}

$P^{(k)}$ can be updated in two different ways. Parameters are chosen such that $F(z) - \min_{z \in \R^n} F(z)$ does not increased more than a multiplicative constant. 
\begin{itemize}
\item Whenever the leverage score $\sigma_{i}(z)$ is smaller than some
universal constant $c_{1}$, constraint $i$ is dropped and $z$ is
updated by the Newton method below. As $\sum_{i}\sigma_{i}(z)=n$, this implies that the number of constraints
is $\leq\frac{\sum_{i}\sigma_{i}}{c_{1}}=\frac{n}{c_{1}}=O(n)$. 

\item Otherwise, the separation oracle is queried at the volumetric center
$z^{(k)}$ and returns a new separating hyperplane $a_{k}^{\top}x\geq b_{k}$.
However, $P^{(k+1)}$ is \emph{not }the intersection of $P^{(k)}$
and $a_{k}^{\top}x\geq b_{k}$. Instead, $a_{k}^{\top}x\geq b_{k}'$ is added for
some $b_{k}'\leq b_{k}$ so that the leverage score of $a_{k}^{\top}x\geq b_{k}'$
is $0.5(\delta c_{1})^{1/2}$, where $\delta\geq10^{3}c_{1}$ is another
small universal constant.
\end{itemize}

\paragraph{Updating $z^{(k)}$ via Newton-type method}

Upon adding or removing a constraint the volumetric center $z^{(k)} \in \R^n$
must be recomputed. As $z^{(k)}$ should minimize $F^{(k)}(x)$, we
iteratively update $z^{(k)}$ via a Newton-type step:
\begin{align*}
z^{\new}\leftarrow z-\frac{1}{10}Q(z)^{-1}\nabla F(z)
\end{align*}
where $Q(z):=A^{\top}S_{z}^{-1}\Sigma S_{z}^{-1}A \in \R^{n \times n}$ can be shown to
be a constant spectral approximation to the Hessian $\nabla^{2}F(z) \in \R^{n \times n}$.
One can calculate $\nabla F(z)=A^{\top}S_{z}^{-1}\sigma \in \R^n$.

We perform this Newton step iteratively until $z^{\new}$
is a good approximate minimizer of $F^{(k)}(x)$, i.e. $F(z)-\min_{z \in \R^n} F(z)\leq c_{2}$
for some small universal constant $c_{2}$. Theorem 1 in full version
of \cite{v89} states that 
\begin{align*}
F( z^{\new} )-\min_{z\in \R^n} F(z)\leq(1-\frac{1}{100})(F(z)-\min_{z \in \R^n} F(z))
\end{align*}
Since the decrease is multiplicative, this can be accomplished in
only $O(1)$ many iterations.

\paragraph{Performance guarantee of Vaidya's method}\mbox{} \\

\noindent\emph{Number of iterations required}

Vaidya showed that after $T$ iterations, the volume of $P^{(k)}$
decreases by a factor of $c^{T-O(n\log (n))}$ for some constant $c$,
i.e.
\begin{align*}
\vol(P^{(k)})\leq c^{T-O(n\log (n))}\vol(P^{(0)})=c^{T-O(n\log (n))}R^{n}.
\end{align*}
Therefore in $T=O(n\log(nR/\epsilon))$ iterations, we have $\vol(P^{(k)})\leq\epsilon^{O(n)}$
showing that $P^{(k)}$ does not contain a ball of radius $\epsilon$
and hence solving the feasibility problem.

\smallskip

\noindent \emph{Running time per iteration}

As $\sum_{i}\sigma_{i}(z)=n$ and the leverage scores are maintained
so that $\sigma_{i}\geq c_{1}$ always holds, the number of constraints
is $\leq\frac{\sum_{i}\sigma_{i}}{c_{1}}=\frac{n}{c_{1}}=O(n)$. Thus
all vectors and matrices above have dimension $O(n)$ and $O(n)\times O(n)$. Moreover, recall that only $O(1)$ steps of Newton method are needed within a iteartion of cutting plane.

Therefore in one iteration, the running time of Vaidya is $O(n^{2})$
plus the time to compute $\sigma$ and to solve a linear system in
$Q(z)^{-1}$ (from the Newton step), which naively requires $O(n^{\omega})$ time. In the
rest of this section we explain speed up these two bottlenecks using
our leverage score maintenance data structure.

\subsection{Our faster implementation via leverage score maintenance}

We provide a faster implementation of Vaidya's method via our leverage
score maintenance data structure, which efficiently updates leverage scores. Specifically, We design a data structure which, upon updates to the volumetric center $z$, maintains an estimate $\tilde{\sigma}$
of the leverage scores $\sigma$ in amortized $O(n^{2})$
time (Theorem~\ref{thm:leverage_score_maintain_main}). Our error guarantee satisfies
\begin{align*}
\|\tilde{\sigma}-\sigma\|_{2}\leq1/\log^{O(1)} (n).
\end{align*}

To apply our data structure, we plug in $W=S_{z}^{-2}$
as the weight in Theorem~\ref{thm:leverage_score_maintain_main}. To establish the validity of our method, we show that conditions (1) and
(2) required for our data structure are satisfied for sufficiently small parameters $c_1,\delta,c_2$. We then prove that Vaidya's performance guarantee is preserved in the presence of a small perturbation to the leverage score.

\begin{lemma}[Condition 1 of data structure (Theorem~\ref{thm:leverage_score_maintain_main})]\label{lem:con1satisfied}
For any constraint $a^{\top} x\geq b$ added or removed, let $s=a^{\top}z-b$
be its slack. We have 
\begin{align*}
\frac{1}{s^2} aa^{\top} \preceq0.01A^{\top}S_{z}^{-2}A.
\end{align*}
\end{lemma}

\begin{proof}
Let $H(z)=A^{\top}S_{z}^{-2}A$. Our goal is to show
\begin{align*}
\frac{1}{s^2} aa^{\top} \preceq0.01H(z).
\end{align*}

Recall that a constraint is removed when its leverage score is smaller
than $c_{1}$ and added so that its leverage score is $(\delta c_{1})^{1/2}$.
In Vaidya's analysis, the only requirement on $c_{1}$ and $\delta$
is that $c_{1},\delta$ are sufficiently small constants and $\delta\geq10^{3}c_{1}$.
Hence in either case, we can make the leverage score of $a^{\top} x\geq b$
smaller than 0.01 by choosing $c_{1}$ and $\delta$ small enough, i.e. the leverage score of $a^{\top}x\geq b$ satisfies
\begin{align*}
\frac{1}{s^{2}} \cdot a^{\top}H(z)^{-1}a \leq0.01.
\end{align*}
Since $H(z)$ is PSD and the square root of a PSD matrix exists,
\begin{align*}
\frac{1}{s^{2}} \cdot a^{\top}H(z)^{-1}a =\frac{1}{s^{2}} \cdot (H(z)^{-1/2}a)^{\top}(H(z)^{-1/2}a) \leq0.01.
\end{align*}
Note that the spectral norm of $(H(z)^{-1/2}a)(H(z)^{-1/2}a)^{\top}$
is $(H(z)^{-1/2}a)^{\top}(H(z)^{-1/2}a)$. Thus
\begin{align*}
\frac{1}{s^{2}} H(z)^{-1/2}a^{\top}aH(z)^{-1/2} =\frac{1}{s^{2}} (H(z)^{-1/2}a)(H(z)^{-1/2}a)^{\top} \preceq0.01I.
\end{align*}
Multiplying by $H(z)^{1/2}$ on the both sides of the above equation, we have
\begin{align*}
\frac{1}{s^2} a^\top a \preceq 0.01 \cdot H(z),
\end{align*}
thus we complete the proof.
\end{proof}

\begin{lemma}[Condition 2 of data structure (Theorem~\ref{thm:leverage_score_maintain_main})]\label{lem:con2satisfied}
Whenever the volumetric center $z \in \R^n$ is updated to $z^{\new} \in \R^n$, we have
\begin{align*}
\left\| \log ( s_{z^{\new}} ) - \log ( s_{z} )\right \|_{2}\leq 0.01.
\end{align*}
\end{lemma}

\begin{proof}
First, we note that it suffices to show that
\begin{align*}
\left\Vert \frac{s_{z^{\new}}-s_{z}}{s_{z}}\right\Vert _{2}\leq0.00001.
\end{align*}
Indeed, this implies that 
\begin{align*}
(1-0.00001)s_{z^{\new}}\leq s_{z}\leq(1+0.00001)s_{z^{\new}}.
\end{align*}
Thus each coordinate of $\log ( s_{z^{\new}} ) - \log ( s_{z} ) $ is
bounded by $\log(1\pm0.00001)$. 

Now using $\log^{2}(1+t)\leq2t^{2}$ for $|t| \leq 0.01$, we have
\begin{align*}
\left\Vert \log ( s_{z^{\new}} ) - \log ( s_{z} ) \right\Vert _{2}\leq2\left\Vert \frac{s_{z^{\new}}-s_{z}}{s_{z}}\right\Vert _{2}\leq0.01.
\end{align*}
It then remains to prove $\left\Vert \frac{s_{z^{\new}}-s_{z}}{s_{z}}\right\Vert _{2}\leq0.00001$.

In Vaidya's work \cite{v89}, they showed that 
\begin{align*}
(z-z^{\new})^{\top}Q(z)(z-z^{\new})
= & ~ O(F(z)-\min_{z\in \R^n} F(z))\\
= & ~ O(\delta(\delta c_1)^{1/4}) \\
= & ~ O( \delta^{5/4} c_1^{1/4} ).
\end{align*}
Recall that leverage scores are at least $c_{1}$, and $Q(z)=A^{\top}S_{z}^{-1}\tilde{\Sigma}S_{z}^{-1}A$
is PSD. Thus
\begin{align*}
(z^{\new}-z)^{\top}Q(z)(z^{\new}-z)\geq c_{1}(z^{\new}-z)^{\top}(A^{\top}S_{z}^{-2}A)(z^{\new}-z).
\end{align*}
Moreover, note that $A(z^{\new}-z) = s_{z^{\new}}-s_{z}$ so 
\begin{align*}
(z^{\new}-z)^{\top}(A^{\top}S_{z}^{-2}A)(z^{\new}-z)=\left\Vert \frac{s_{z^{\new}}-s_{z}}{s_{z}}\right\Vert _{2}^2.
\end{align*}
Combining all, we obtain
\begin{align*}
\left\Vert \frac{s_{z^{\new}}-s_{z}}{s_{z}}\right\Vert_{2}^2=O(\delta^{5/4}/c_1^{3/4}).
\end{align*}
\end{proof}

Having established the conditions of our data structure which maintains perturbed leverage scores, we argue that Vaidya's method tolerates small additive perturbations
$o(1)$ in the leverage scores\footnote{In fact, Vaidya's method would survive even if the perturbation is
a sufficiently small constant.}. Note that in the presence of such perturbations, both the procedure for updating
$P^{(k)}$ and the Newton step are affected.

\begin{lemma}[Newton step]\label{lem:newtonPreserved}
Assume the leverage scores $\sigma$ are replaced by estimate $\tilde{\sigma}$,
where $\|\tilde{\sigma}-\sigma\|_{2}\leq1/\log^{O(1)} (n)$ in Vaidya's method.
Then Newton step still requires only $O(1)$ many iterations so that $F(z) - \min_{z\in \R^n} F(z)\leq c_{2}$ for some small constant $c_{2}$.
\end{lemma}

\begin{proof}
For the Newton step, let $\widetilde{\Sigma}$ be the diagonal matrix
of $\tilde{\sigma}$. Vaidya's Newton step is modified as
\begin{align*}
z^{\new}\leftarrow z-\frac{1}{10}(A^{\top}S_{z}^{-1}\widetilde{\Sigma}S_{z}^{-1}A)^{-1}A^{\top}S_{z}^{-1}\tilde{\sigma}
\end{align*}
As only $\Sigma$ and $\sigma$ are changed, this amounts to a small
difference in the convergence rate, i.e.
\begin{align*}
F(z^{\new})- \min_{z\in \R^n} F(z)\leq(1-\frac{1}{101})(F(z)-\min_{z \in \R^n}F(z))+O(\|\tilde{\sigma}-\sigma\|_{2}^{2}).
\end{align*}
Recall that our goal is $F(z)-\min_{z \in \R^n}F(z)\leq c_{2}$ for a small universal constant $c_{2}$.
Therefore our modified Newton step still requires only $O(1)$ many iterations
as long as $\|\tilde{\sigma}-\sigma\|_{2}^{2}=o(1)$.
\end{proof}

\begin{lemma}[Convergence rate]\label{lem:ratePreserved}
Assume the leverage scores $\sigma$ are replaced by estimate $\tilde{\sigma}$,
where $\|\tilde{\sigma}-\sigma\|_{2}\leq1/\log^{O(1)} (n)$ in Vaidya's method.

Then Vaidya's convergence guarantee still holds: after $T$ iterations, the volume of $P^{(k)}$
decreases by a factor of $c^{T-O(n\log (n))}$ for some constant $c$,
i.e.
\begin{align*}
\vol(P^{(k)})\leq c^{T-O(n\log (n))}\vol(P^{(0)})=c^{T-O(n\log (n))}R^{n}.
\end{align*}
\end{lemma}
\begin{proof}
The proof of Vaidya's convergence lemma essentially depends on the fact that leverage scores are at least $c_1$ and at most $0.5(\delta c_1)^{1/2}$. 

Note that $P^{(k)}$ is updated by dropping constraint $i$ if $\sigma_{i}(z)\geq c_{1}$, or adding constraint $i$ s.t. $\sigma_{i}(z) = 0.5(\delta c_1)^{1/2}$. 
The purpose of the constraint adding and dropping is to make sure the leverage
score of all constraints are $\Theta(1)$. Hence, we can use any constant
approximation to leverage score. In particular, an additive $o(1)$ pertubation in the leverage score can be absorbed by scaling $c_1,\delta$ slightly.
\end{proof}

Now we are ready to bound the running time of our modification of Vaidya's cutting plane method and complete the proof of Theorem~\ref{thm:mainvaidya}. 

%\begin{theorem}[Performance of our Cutting Plane Method]\label{thm:perfcutting}
%We can implement Vaidya's cutting plane method in $O((n\SO+n^{3})\log (\kappa))$ time.
%\end{theorem}

\begin{proof}[Proof of Theorem~\ref{thm:mainvaidya}]
By Lemma~\ref{lem:ratePreserved}, in $T=O(n\log(nR/\epsilon))=O(n\log (\kappa))$ iterations, we have $\vol(P^{(k)})\leq\epsilon^{O(n)}$
showing that $P^{(k)}$ does not contain a ball of radius $\epsilon$.
Thus the number of calls to the separation oracle is $O(n \log (\kappa))$.
We next analyze the runtime per iteration.

By Lemma~\ref{lem:newtonPreserved}, we still only need $O(1)$ Newton steps. Thus, as argued at the end of last subsection, the per-iteration running
time is $O(n^{2})$ plus the time to compute $\sigma$ and
to solve a linear system in $Q(z)^{-1}$. We argue that both of these
two tasks can be accomplished in amortized $O(n^{2})$ time.

For $\sigma$, we instead use its estimate $\tilde{\sigma}$ output
by our leverage score maintenance data structure (Theorem~\ref{thm:leverage_score_maintain_main}). By Lemmas~\ref{lem:con1satisfied} and~\ref{lem:con2satisfied},
the conditions of the data structure are satisfied. Hence we can update
$\tilde{\sigma}$ in amortized $O(n^{2})$ time.

Solving a linear system in $Q(z)^{-1}$, as pointed out in Theorem
31 of LSW~\cite{lsw15}, can be done by inverse maintenance. Using the inverse maintenance
procedure in~\cite{cls19}, this can also be done in amortized $O(n^{2})$
time (see Theorem~\ref{thm:maintain_projection_inn}).
\end{proof}
\newpage
\section{Modified Projection Maintenance}\label{sec:projection_maintenance}

\begin{theorem}[{\cite[Appendix]{cls19}}]
\label{thm:maintain_projection}
%Assume that in $O(t_{k})$ time, we
%can multiply a $m\times m$ and a $m\times k$ matrix, and we can
%multiply a $m\times k$ and a $k\times m$ matrix. Also, assume that
%$t_{k}/k$ is decreasing in $k$.
Let $t_k = \T_{\mat}(m,m,k)$ denote the time to multiply a $m\times m$ and a $m\times k$ matrix. 
Given a matrix $A\in\R^{m\times n}$ with $m\geq n$, and $k^{*}\in[m]$,
there is a deterministic data structure that approximately maintains
the projection matrices $\sqrt{W}A(A^{\top}WA)^{-1}A^{\top}\sqrt{W}$
and the inverse matrices $(A^{\top}WA)^{-1}$ for positive diagonal
matrices $W$ through the following operations: 
\begin{itemize}
\item $\textsc{Initialize}(A,w,\epsilon)$: Initialize the data structure
of the matrix $A$, the weight $w$ and the target accuracy $\epsilon\in(0,1/4)$
in $m^{\omega+o(1)}$ time.
\item $\textsc{Update}(w)$: Output a vector $v\in\R^{m}$ along with $(A^{\top}VA)^{-1}$
and $\sqrt{V}A(A^{\top}VA)^{-1}A^{\top}\sqrt{V}$ such that for all
$i$,
\[
(1-\epsilon)v_{i}\leq w_{i}\leq(1+\epsilon)v_{i}.
\]
\item $\textsc{Insert}(a,w_{a})$: Insert a column $a$ into $A$, a weight
$w_{a}$ into $w$ in $O(m^{2})$ time.
\item $\textsc{Delete}(a,w_{a})$: Delete a column $a$ from $A$ and its corresponding
weight $w_a$ from $w$ in $O(m^{2})$ time.
\end{itemize}
Suppose that the number of columns is $O(m)$ during the whole algorithm
and that for any call of $\textsc{Update}$, we have 
\[
\sum_{i=1}^{m}\left(\log w_{i}-\log w_{i}^{(\textrm{old})}\right)^{2}\leq C^{2}
\]
where $w$ is the input of call, $w^{(\textrm{old})}$ is the weight
before the call. Then, the amortized expected time per call of $\textsc{Update}(w)$
is
\[
O\left(t_{k^{*}}+(C/\epsilon)\cdot \left( \frac{t_{k^{*}}^2}{k^{*}}+\sum_{i=k^{*}}^{m}\frac{t_{i}^{2}}{i^{2}}\right)^{1/2}\cdot\log m \right)
\]
\end{theorem}

We will use this theorem with difference algorithms for rectangular
matrix multiplication. Recall that it takes $O(m^{2}\log^{2}m)$ time
to multiply an $m\times m$ and an $m\times m^{0.17}$ matrix. By splitting
the matrix into blocks (See e.g. \cite[Lemma A.5]{cls19}), one can check
it takes
\[
t_{k}:= O(m^{2}\log^{2}m+k^{\frac{\omega-2}{1-\alpha}}m^{2-\frac{\alpha(\omega-2)}{1-\alpha}}\log^{2}m)
\]
time to multiply a $m\times m$ and a $m\times k$ matrix with $\alpha=0.17$.
Using this and putting $k^{*}=m^{\alpha}$, we have
\begin{align*}
\left( \frac{t_{k^{*}}^{2}}{k^{*}}+\sum_{i=k^{*}}^{m}\frac{t_{i}^{2}}{i^{2}} \right)^{1/2} & \leq\frac{m^{2}\log^{2}m}{m^{\alpha/2}}+m^{2-\frac{\alpha(\omega-2)}{1-\alpha}}\log^{2}m \left( \sum_{i=k^{*}}^{m}\frac{i^{2\frac{\omega-2}{1-\alpha}}}{i^{2}} \right)^{1/2}\\
 & =m^{2-\alpha/2}\log^{2}m+m^{2-\frac{\alpha(\omega-2)}{1-\alpha}}\log^{2}m\cdot O(m^{-\frac{\alpha}{2}(1-2\frac{\omega-2}{1-\alpha})})\\
 & =O(m^{2-\alpha/2}\log^{2}m)
\end{align*}
Hence, applying Theorem \ref{thm:maintain_projection} with $k^*=m^{0.17}$, we have the following Theorem
\begin{theorem}\label{thm:maintain_projection_mid} %%% this is 2
There is a variant of the data structure
in Theorem \ref{thm:maintain_projection} where the amortized time per call
of $\textsc{Update}(w)$ is 
\begin{align*}
O(m^{2}\log^{2}m+ C m^2 \epsilon^{-1}m^{-0.085}).
\end{align*}
\end{theorem}

%\begin{proof}
%This follows from the same proof in \cite{cls19} and the fact that
%we can multiply $m\times m^{0.17227}$ and $m^{0.17227}\times m$
%matrix in $m^{2}\log^{2}m$ time. (Note that $0.086135=\frac{0.17227}{2}$.)
%\end{proof}
Unfortunately, this version still have extra $\log^{O(1)}m$ terms
in the runtime. To get the $m^{2}$ time, we use the following lemma:
\begin{lemma}
\label{lem:n2mul}For any $c>0$, it takes 
\begin{align*}
O \Big( m^{2} + \frac{m^{2}r^{0.4}}{\log^{c}m} \Big)
\end{align*}
time to multiply $m\times r$ and $r\times m$ matrices.
\end{lemma}

\begin{proof}
If $r>m^{\frac{0.38}{0.39}}$, we simply multiply it using a $m^{2.38}$
time square matrix multiplication algorithm. This is faster than $m^{2}r^{0.39}\leq O(\frac{m^{2}r^{0.4}}{\log^{c}m})$.
If $r<\log^{2c}m$, we simply use (3) in Theorem \ref{thm:fast_matrx_multiplication} which takes
$O(m^{2})$ time. Hence, we can assume $\log^{2c}m<r<m^{\frac{0.38}{0.39}}$.

Let $k=\frac{r}{\log^{c/0.4}m}$. We can view the problem as multiplying
a $\frac{m}{k}\times\frac{r}{k}$ and a $\frac{r}{k}\times\frac{m}{k}$
block matrices and each block has size $k\times k$ size. 

Note that
\begin{align*}
\frac{m}{k}=\frac{m\log^{c/0.4}m}{r}=m^{\Omega(1)} \text{~~~and~~~} \frac{r}{k}=\log^{O(1)}m.
\end{align*}
Hence,  (3) in Theorem \ref{thm:fast_matrx_multiplication} shows that the total cost is $O((\frac{m}{k})^{2})$
many block matrix multiplication and each takes $O(k^{2.4})$ time.
Therefore, the total cost is
\begin{align*}
O \Big( \left(\frac{m}{k}\right)^{2} \times k^{2.4} \Big) = O( m^{2} k^{0.4} ) = O \Big( \frac{m^{2}r^{0.4}}{\log^{c}m} \Big).
\end{align*}
\end{proof}

Now, applying Theorem \ref{thm:maintain_projection} with $k^*=\log^{O(1)}m$, we have
\begin{theorem}\label{thm:maintain_projection_inn} %%% this is 3
For any $c>0$, there is a variant of the
data structure in Theorem \ref{thm:maintain_projection} with the amortized
time per call of $\textsc{Update}(w)$ is 
\begin{align*}
O \Big( m^{2}+\frac{ C m^2 }{ \epsilon \log^{c} m  } \Big).
\end{align*}
\end{theorem}
%
%\begin{proof}
%One fact \cite{cls19} used to analyze their $\textsc{Update}(w)$ operation
%is that it takes 
%\begin{align*}
%m^{2+o(1)} + r^{\frac{\omega-2}{1-\alpha}} \cdot m^{2-\frac{\alpha(\omega-2)}{1-\alpha}+o(1)}
%\end{align*}
%time to multiply $m\times r$ and $r\times m$ matrices (Lemma A.5
%in \cite{cls19}). If we use Corollary \ref{cor:n2mul} instead and
%change the batch size for updating the inverse from $m^{\alpha}$
%to $\log^{O(1)}m$, a similar\footnote{The only difference is that we need to change the parameter $g_{i}$
%from $g_{i}=i^{\frac{\omega-2}{1-\alpha}-1}m^{-\frac{\alpha(\omega-2)}{1-\alpha}}$
%to $g_{i}=i^{0.4-1}/\log^{c}m$ and update all calculations about
%$g$.} proof in \cite{cls19} shows the result.
%\end{proof}
\newpage
\section{Cutting Plane Method for Convex Minimization and Saddle Point Problems}\label{sec:saddle_point}

We show in this section that cutting plane methods can be applied to not only convex minimization, but also the more general problem of computing a saddle point in a convex-concave game with essentially the same guarantee.

\subsection{Notations and definitions}
For a convex set $\mathcal{X} \subseteq \R^n$, the interior of $\mathcal{X}$, denoted as $\interior(\mathcal{X})$, is the subset of points in $\mathcal{X}$ that has a small neighborhood fully contained inside $\mathcal{X}$.
We denote by $\vol(\mathcal{X})$ the volume of $\mathcal{X}$. For any vector $x$ and $r>0$, we use $B(x,r)$ to denote the $\ell_2$ ball of radius $r$ centered at $x$, i.e. $B_\infty(x,r) = \{y : \|x-y\|_2 \leq r\}$. Similarly, $B_\infty(x,r) = \{y : \|x-y\|_\infty \leq r\}$.
For any set $K$, we denote $B(K,-\epsilon)$ to denote the set $\{x : B(x,\epsilon) \subset K \}$.

\subsection{Convex minimization}

Using a standard reduction of convex minimization to the feasiblity problem (\cite{nemi94} and Theorem 42 of~\cite{lsw15}), we can minimize a convex function with ${O}(n\log\kappa)$ subgradient oracle calls and $O(n^3\log\kappa)$ time. This improves over the previous best of $O(n^3\log^{O(1)}\kappa)$~\cite{lsw15}. Since $\kappa$ can be exponential in certain applications, this allows us to obtain significantly faster algorithms (see e.g. subsection~\ref{sec:walras}).

\begin{theorem}\label{thm:convex}
Let $f$ be a convex function on $\R^n$ and $S$ be a convex set that contains a minimizer of $f$. Suppose we have a subgradient oracle for $f$ with cost $\mathcal{T}$ and $S\subset B_\infty(0,R)$. Using $B_\infty(0,R)$ as the initial polytope for our Cutting Plane Method, for any $0<\alpha<1$, we can compute $x\in S$ such that $f(x)-\min_{y\in S}f(y)\leq \alpha\left(\max_{y\in S}f(y)-\min_{y\in S}f(y)\right),$ with high probability in $n$ and with a running time of 
\begin{align*}
O(\T \cdot n \log(\kappa)+n^{3}\log (\kappa)),
\end{align*} 
where $\kappa =  n\gamma / \alpha$ and $\gamma=R /\mathrm{minwidth}(S)$. Here the {\rm minwidth} of $S$ is defined by 
\begin{align*}
\mathrm{minwidth}(S):=\min_{a\in\R^n:\left\Vert a \right\Vert_2 =1}  \left( \max_{y\in S} a^\top y - \min_{y\in S} a^\top y \right).
\end{align*}
\end{theorem}

\subsection{Convex-concave games}

In this subsection we show that a similar guarantee holds for solving convex-concave games. Much of the materials in this section are modified from \cite[Lecture 5]{n95}. For completeness, we will explain both the standard theory and the various changes needed for our promised runtime.

In the convex-concave game, we are asked to solve 
\begin{align*}
\min_{x \in \mathcal{X}} \max_{y \in \mathcal{Y}} f(x,y),
\end{align*}
where $\mathcal{X} \in \R^n,\mathcal{Y} \in \R^m$ are convex sets and $f(x,y): \mathcal{X} \times \mathcal{Y} \rightarrow \R$ is a continuous function that is convex in $x \in \mathcal{X}$ and concave in $y \in \mathcal{Y}$. 
Von Neumann's minimax theorem states that 
\begin{align*}
\min_{x \in \mathcal{X}} \max_{y \in \mathcal{Y}} f(x,y) = \max_{y \in \mathcal{Y}} \min_{x \in \mathcal{X}} f(x,y),
\end{align*}
and that all solutions to the LHS are solutions to the RHS, and vice versa.
Any solution $(x^*,y^*)$ to this problem is called a {\em saddle point}, and satisifes $f(x,y^*) \geq f(x^*,y^*) \geq f(x^*,y)$ for any $(x,y) \in \mathcal{X} \times \mathcal{Y}$. 

We will be interested in computing an $\epsilon$-saddle point which we define in the following. Define $\overline{f}(x):= \max_{y \in \mathcal{Y}} f(x,y)$ and $\underline{f}(y) = \min_{x \in \mathcal{X}} f(x,y)$. The minimax theorem is equivalent to 
\begin{align}
\label{eqn:minimaxEquiv}
\min_{x \in \mathcal{X}} \overline{f}(x) = \min_{x \in \mathcal{X}} \max_{y \in \mathcal{Y}} f(x,y) =\max_{y \in \mathcal{Y}}\min_{x \in \mathcal{X}} f(x,y) = \max_{y \in \mathcal{Y}} \underline{f}(y),
\end{align}
and it is immediate that the set of saddle points $S^*(f)$ is exactly the direct product of the optimal solutions to $\overline{f} : {\cal X} \rightarrow \mathbb{R}$ and $\underline{f} : {\cal Y} \rightarrow \mathbb{R}$:
\begin{align*}
S^*(f) = \arg\max_{x \in \mathcal{X}} \overline{f}(x) \times \arg\min_{y \in \mathcal{Y}} \underline{f} (y).
\end{align*}

Since certifying the values of $f(x,y)$ around the boundary of $\mathcal{X}$ and $\mathcal{Y}$ is quite difficult under the black-box setting, our definition of $\epsilon$-saddle point ignores small portion of the domain around the boundary.

\begin{definition}[$\epsilon$-Saddle Point]
\label{defn:epsSaddlePoint}
Consider a convex-concave game with convex sets $\mathcal{X} \subseteq B(0,R) \subset \R^n, \mathcal{Y} \subseteq B(0,R) \subset \R^m$, and $L$-Lipschitz function $f(x,y): \mathcal{X} \times \mathcal{Y} \rightarrow \R$ that is convex in $x$ and concave in $y$.
Given $\epsilon \in (0,1)$, a pair $(x,y) \in \mathcal{X} \times \mathcal{Y}$ is called an $\epsilon$-saddle point, if
\begin{align}
\label{eqn:dualityGap}
\overline{f}(x) - \underline{f}(y) \leq \epsilon L R.
\end{align}
\end{definition}

It follows from~(\ref{eqn:minimaxEquiv}) that the LHS of~(\ref{eqn:dualityGap}) is
%%{\color{red}Zhao: I don't understand this sentence.}
\begin{align*}
\overline{f}(x) - \underline{f}(y) = \left( \overline{f}(x) - \min_{x' \in \mathcal{X}} \overline{f}(x') \right) + \left( \max_{y' \in \mathcal{Y}} \underline{f}(y') - \underline{f}(y) \right).
\end{align*}
%{\color{red}Zhao : I don't understand the above equation.}

In what follows we assume access to a first-order oracle which, 
%(1) give $(x,y) \notin \interior(\mathcal{X} \times \mathcal{Y})$, returns a separating hyperplane given by a vector $g(x,y)$ such that $(x,y)^\top g(x,y) \geq \sup_{z' \in \mathcal{X} \times \mathcal{Y}} z'^\top g(x,y)$, and (2) 
given $(x,y) \in \interior(\mathcal{X} \times \mathcal{Y})$, returns the subgradient vector
\begin{align}\label{eq:def_g_x_y}
g(x,y) = (\nabla_x f (x,y), - \nabla_y f(x,y)) \in \R^n \times \R^m.
\end{align}

The crucial property of this vector is as follows:
\begin{lemma}[Convex-Concave Property]
\label{lem:convexprop}
Let $f : {\cal X} \times {\cal Y} \rightarrow \mathbb{R}$ be convex in $x$ and concave in $y$. Let $g : {\cal X} \times {\cal Y} \rightarrow \R^{n+m}$ be defined as in Eq.~\eqref{eq:def_g_x_y}.  
For any $z = (x,y) \in \mathcal{X} \times \mathcal{Y}$ and $z'=(x',y') \in \interior(\mathcal{X} \times \mathcal{Y})$, we have
\begin{align*}
(z'-z)^\top g(z') \geq f(x',y) - f(x,y').
\end{align*}
In particular, if $z = (x,y)$ is a saddle point of $f$, then we have
\begin{align*}
(z' - z)^\top g(z') \geq 0.
\end{align*}
\end{lemma}
\begin{proof}
Since $f$ is convex in $x$ and concave in $y$, we have
\begin{align*}
f(x,y') - f(x',y') \geq (x-x')^\top \nabla_x f(x',y') \qquad \text{and} \qquad f(x',y) - f(x',y') \leq (y - y')^\top \nabla_y f(x',y').
\end{align*}
It follows that
\begin{align*}
f(x',y) - f(x,y') \leq (z'-z)^\top g(z'),
\end{align*}
which proves the first part of the lemma. For the second part, simply notice that if $z = (x,y)$ is a saddle point, we have
\begin{align*}
f(x',y) \geq f(x,y) \geq f(x,y').
\end{align*}
\end{proof}

\subsection{Applying cutting plane method to convex-concave games}
In the cutting plane framework, we maintain a polytope $P^{(k)} \subseteq \R^{n + m}$ that contains a saddle point $(x^*,y^*)$ of $f$. The algorithm initially starts with the polytope
$P^{(0)} := B_\infty(0,R) \supset \mathcal{X} \times \mathcal{Y}$.

In each iteration, we call the first-order oracle at a certain point $z^{(k)} = (x^{(k)},y^{(k)}) \in P^{(k)}$ depending on $P^{(k)}$. 
If $z^{(k)} \notin \mathcal{X} \times \mathcal{Y}$ is not feasible, then we obtain a supporting hyperplane to separate $z^{(k)}$ from the feasible region $\mathcal{X} \times \mathcal{Y}$; and if $z^{(k)} \in \mathcal{X} \times \mathcal{Y}$ is feasible, Lemma~\ref{lem:convexprop} shows that any saddle point $(x^*,y^*)$ lies in the half-space $H^{(k)}:=\{z \in \R^{n+m}: (z-z^{(k)})^\top g(z^{(k)}) \leq 0 \}$ and hence $(x^*,y^*) \in  P^{(k)} \cap H^{(k)}$. 
The algorithm continues by choosing $P^{(k+1)}$ which contains $P^{(k)} \cap H^{(k)}$.

For simplicity, we extend the gradient function $g$ as follows: 
\begin{align}\label{eq:def_ghat_x_y}
\hat{g}_{\beta}(x,y)=\begin{cases}
(\nabla_{x}f(x,y),-\nabla_{y}f(x,y)) & \text{if }(x,y)\in\interior(\mathcal{X}\times\mathcal{Y})\\
n(x,y) & \text{otherwise}
\end{cases}
\end{align}
where $n(x,y)$ is some vector of length $\beta > L$ such that $\mathcal{X}\times\mathcal{Y} \subset \{z \in \R^{n+m} : (z-(x,y))^\top n(x,y) \leq 0\}$.
The cutting plane framework is given in Algorithm~\ref{alg:cutting_plane_method}.

\begin{algorithm}\caption{Cutting Plane Method}\label{alg:cutting_plane_method}
\begin{algorithmic}[1]
\Procedure{CuttingPlaneMethod}{$\hat{g}, T,n,m$} \Comment{Theorem~\ref{thm:vaidya}}
\State $P^{(0)} = B_\infty(0,R)$ \Comment{$P^{(0)} \subseteq \R^{n + m}$}
\For{$k=0,1,2, \cdots, T$}

%	\EndIf
	\State Find a point $z^{(k)} \in P^{(k)}$ \Comment{$z^{(k)} \in \R^{n+m}$}
	\If{$z^{(k)} \notin B_\infty(0,R)$}
		\State Project $z^{(k)}$ onto the face of the boundary $\partial B_\infty(0,R)$
		\State Let $\hat{g}_{\beta}(z^{(k)})$ be the normal of that face (scaled by $\beta$). \label{line:move_z}
	\EndIf
%		Let $H^{(k)}$ be a half-space passing through $z^{(k)}$ and separating it from $\mathcal{X} \times \mathcal{Y}$ \;
%	\Else 
		\State Define $H^{(k)} := \{z \in \R^{n+m} : (z-z^{(k)})^\top \hat{g}_{\beta}(z^{(k)}) \leq 0\}$ \Comment{$H^{(k)} \subseteq \R^{n+m}$} (as the separating hyperplane)
%	\EndIf
	\State Construct the next $z^{(k+1)},P^{(k+1)}$ according to cutting plane method
\EndFor
\State \Return a point according to Lemma \ref{lem:convcombSolution}
\EndProcedure
\end{algorithmic}
\end{algorithm}

%The correctness and convergence property of the cutting plane framework is implied by Lemma~\ref{lem:convexprop}. We can use whatever cutting plane method for convex minimization here to solve the convex-concave game with the same performance guarantee.

%Two specific kinds of cutting plane methods that we will study in the following are the \textsc{CenterOfGravity} method and the \textsc{Ellipsoid} method.
%In the \textsc{CenterOfGravity} method we set
%\begin{align*}
%E^{(0)} = \mathcal{X} \times \mathcal{Y},~~ z^{(k)} = \frac{1}{\vol(E^{(k)})} \int_{E^{(k)}} z dz.
%\end{align*}
%In the \textsc{Ellipsoid} method all $E^{(k)}$'s are ellipsoids (in particular, $E^{(0)}$ is an ellipsoid containing the feasible set $\mathcal{X} \times \mathcal{Y}$), $z^{(k)}$ is the center of $E^{(k)}$, and $E^{(k+1)}$ is the ellipsoid of the smallest volume that contains $E^{(k)} \cap H^{(k)}$. 
%For both methods, the volume of $E^{(k)}$ shrink as the number of iteration increases as in the following lemma.

%\begin{lemma}[Volume Shrinking]\label{lem:volume_shrinking}
%For the \textsc{CenterOfGravity} method
%\begin{align*}
%\vol(E^{(k)}) \leq \exp\left(-\Omega(1) \cdot k \right) \cdot %\vol(\mathcal{X} \times \mathcal{Y}).
%\end{align*}
%For the \textsc{Ellipsoid} method, we have
%\begin{align*}
%\vol(E^{(k)}) \leq \exp\left(- \Omega(1) \cdot k / (n+m) \right) \cdot \vol(\mathcal{X} \times \mathcal{Y}).
%\end{align*}
%\end{lemma}
%We refer interested readers to the book~\cite{b15} for more details about these methods. 

\begin{theorem}\label{thm:vaidya}
Let $c<1$ be a universal constant. For $T\geq n+m$, we can find $z^{(k)} \in \R^{n + m}$ and $P^{(k)} \subseteq \R^{n + m}$ in total time
$O((n+m)^{2}T)$ with high probability in $n+m$ such that 
\begin{align*}
\vol ( P^{(T)} ) \leq c^{T-O(n+m) \log(n+m) } \vol ( P^{(0)} ).
\end{align*}
Furthermore, $P^{(T)}\subset\bigcap_{k\in I}H^{(k)}\cap B_{\infty}(0,R)$
with $|I|=O(n)$.
\end{theorem}

\begin{proof}
By Lemma~\ref{lem:convexprop}, any saddle point lies in $H^{(k)}$ showing that it is indeed a separating hyperplane. With this separation oracle, we apply our faster implementation of Vaidya's cutting plane method from section~\ref{sec:vaidya}.

The second part of the lemma follows from the fact that this method always maintains a polytope with $O(n)$ constraints. For the first part, note that this just paraphases Lemma~\ref{lem:ratePreserved}.
\end{proof}

\subsection{Cutting plane method for convex-concave games: generating solutions}

Recall that for minimizing a convex function $f(z)$, one can simply output the best $z^{(k)}$ found in all iterations. 
The argument here is that as long as the volume $\vol(P^{(T)})$ is small enough, then in some iteration $k$, a point close to the optimal solution $z^*$ gets removed from $P^{(k)}$ which indicates that $f(z^{(k)})$ is also close to optimal. 
For the convex-concave game, however, such a naive approach of outputting the ``best'' $z^{(k)}$ would fail as illustrated by \cite[Section 5.3]{n95}. The crucial difference here in the convex-concave game is that although we have an objective function $f$, but we are not interested in optimizing it. Instead, we are interested in finding an $\epsilon$-saddle point $(x,y)$ that satisfies~\eqref{eqn:dualityGap}.

The correct idea is to output some convex combination of all $z^{(k)}$'s that does converge to the saddle point $(x^*,y^*) \in \mathcal{X} \times \mathcal{Y}$ which we describe in the following.

Assume that we have performed $T$ steps of the cutting plane method.
%\begin{definition}[Productive step vs non-productive step]\label{def:productive_step}
%We call a step $k$ {\em productive} if $z^{(k)} \in B_\infty(0,R)$, and otherwise {\em non-productive}.
%\end{definition}

\begin{definition}[Gap function]
Let $I$ be the set of all constraints in $P^{(k)}\cap B_\infty(0,R)$. If the $k^{\text{th}}$ constraint comes from $B_\infty(0,R)$, denote by  $z^{(k)}$ the center of the face and $g^{(k)}$ the vector normal to the face scaled by $\beta$.  Otherwise, $z^{(k)}$ and $g^{(k)}$ denote the query point $z^{(k)}$ of the oracle and its output  $\hat{g}_{\beta}(z^{(k)})$. The gap function $\gamma$ is defined as 
\begin{align*}
\gamma(z) = \min_{k \in I} (z^{(k)} - z)^{\top} g^{(k)}:= \min_{k \in I} \gamma^{(k)}(z).
\end{align*}
%where $I \subset [T+2(n+m)]$ the set of all iterations $k \leq T+2(n+m)$ that are productive (Definition~\ref{def:productive_step}). %i.e., $I_T = \{ k \in [T] : k^{\mathrm{th}}\mathrm{~iteration~is~productive} \}$.
\end{definition}
%Note that without loss of generality, we can assume that $I_T$ is non-empty. 
% %%
%Consider the auxiliary function
%\begin{align*}
%g_T(z) = \min_{k \in I_T} (z^{(k)} - z)^{\top} \hat{g}_{\beta}(z^{(k)}) := \min_{k \in I_T} g^{(k)}(z).
%\end{align*}
Notice that $\gamma(z)$ is concave as it is the minimum of affine functions. 

Now we show that if $\gamma(z)$ is small for all $z$, then we can form a good approximation to the saddle point set from $z^{(k)}$.
We first prove the following lemma which states that the maximum of $\gamma(z)$ is given by some certain convex combination of $\gamma^{(k)}(z)$. 

\begin{lemma}[{Optimal Lagrange Multipliers}]
\label{lem:ConvCombSameMax}
Define $\Delta_{I}$ to be the simplex of all convex combinations of $I$, i.e., 
$ \Delta_{I} := \left\{\lambda \in \R^{I} ~\Big|~\lambda_k \geq 0, \sum_{k \in I} \lambda_k = 1 \right\}.$
Then there exist {\em optimal Lagrange multipliers} $\lambda^* = \{ \lambda_k^* \}_{k \in I} \in \Delta$ such that
\begin{align*}
\sum_{k \in I} \lambda_k^* \cdot \gamma^{(k)}(z') = \max_{z \in \R^{n+m}} \gamma(z)\qquad \forall z' \in \R^{n+m}.
\end{align*}
Furthermore, for any $0 < \eta < 1/2$, we can find $\lambda \in \Delta$ such that
\begin{align}\label{eq:apx_lag_mul}
\sum_{k \in I} \lambda_k \cdot \gamma^{(k)}(z') \leq \max_{z \in \R^{n+m}} \gamma(z) + \eta \beta R \qquad \forall z' \in B_\infty(0, R).
\end{align}
in time $O((n+m)^{\omega+o(1)} \log ( (n+m) / \eta )$.
\end{lemma}

\begin{proof}
By the minimax theorem, we have
\begin{align*}
\max_{z \in \R^{n+m}} \gamma(z) := \max_{z \in \R^{n+m}} \min_{k \in I} \gamma^{(k)}(z) = \max_{z} \min_{\lambda \in \Delta} \sum_{k \in I} \lambda_k \cdot \gamma^{(k)}(z) = \min_{\lambda \in \Delta} \max_{z \in \R^{n+m}} \sum_{k \in I} \lambda_k \cdot \gamma^{(k)}(z),
\end{align*}
where we used the fact that $\sum_{k \in I} \lambda_k \cdot \gamma^{(k)}(z)$ is bilinear in $\lambda$ and $z$. 
The first part of the lemma follows from the fact that the affine function $\sum_{k \in I} \lambda_k \cdot \gamma^{(k)}(z)$ has a maximium over $z\in\R^{n+m}$ if and only if it is a constant function. Note that this happens precisely when $\sum_{k\in I}\lambda_k \cdot g^{(k)}=0$.

The second part of the lemma follows from the observation that ``$\sum_{k \in I} \lambda_k \cdot \gamma^{(k)}(z)$ is a constant function'' is a linear constraint over $\lambda$. Therefore, finding $\lambda \in \Delta$ that forms a constant function with smallest possible value can be captured by the following linear program (in $\lambda$):

\begin{align*}
\max_{\lambda} ~&  \sum_{k\in I} z^{(k)}\cdot  g^{(k)}\lambda_{k}\\
\text{s.t.}~ & \sum_{k\in I}g^{(k)}\lambda_{k}	=0\\
& \sum_{k\in I}\lambda_{k} =1 \\
& \lambda_{k}	\geq0 , \forall k\in I .\\
\end{align*}
Using a recent LP solver from Theorem 2.1 in \cite{cls19}, in $O((n+m)^{\omega+o(1)}\log(1/\delta))$ time we can output an approximate solution $\lambda_k$ satisfying

\begin{align*}
\sum_{k\in I} z^{(k)}\cdot  g^{(k)}\lambda_{k} &\geq \sum_{k\in I} z^{(k)}\cdot  g^{(k)}\lambda_{k}^* - \delta \max_{k\in I}|z^{(k)}\cdot  g^{(k)}|\\
& \geq \sum_{k\in I} z^{(k)}\cdot  g^{(k)}\lambda_{k}^* - \delta \sqrt{n+m}R\beta\\
& = \max_{z \in \R^{n+m}} \gamma(z)  - \delta \sqrt{n+m}R\beta
\end{align*}
and 
\begin{align*}
\norm{\sum_{k\in I}g^{(k)}\lambda_{k}}_1 \leq \delta \left( \sum_k \norm{g_k}_1 \right)\leq \delta \sqrt{n+m}\beta.
\end{align*}	
where we used $\norm{g^{(k)}}_2 \leq \max (\beta,L) = \beta$ and $ z^{(k)}\in B_\infty(0,R)$. 
Now for $z' \in B_\infty(0, R)$,
\begin{align*}
&\sum_{k \in I} \lambda_k \cdot \gamma^{(k)}(z') \\
= & \sum_{k\in I} z^{(k)}\cdot  g^{(k)}\lambda_{k} + \left( \sum_{k\in I}g^{(k)}\lambda_{k} \right)\cdot z'\\
\geq & \max_{z \in \R^{n+m}} \gamma(z)  - \delta \sqrt{n+m}R\beta -  (\delta \sqrt{n+m}\beta)(\sqrt{n+m}R)\\
\geq & \max_{z \in \R^{n+m}} \gamma(z) -2(\delta (n+m)\beta R)
\end{align*}
Our result then follows by taking $\eta = 2 \delta (n+m)$.
\end{proof}

Now given the notion of approximate multipliers \eqref{eq:apx_lag_mul}, we prove the following lemma. 

\begin{lemma}[Convex Combination Provides a Good Solution]
\label{lem:convcombSolution}
Given $\lambda$ satisfying $$\sum_{k \in I} \lambda_k \cdot \gamma^{(k)}(z') \leq 9\eta \beta R (n+m)^{1/2} \qquad \forall z' \in B_\infty(0, R)$$ with $\eta<1$. Assume that $\beta \geq 3\sqrt{n+m} \cdot L$. Let 
\begin{align*}
\hat{z} = (\hat{x},\hat{y}) = \Big( \sum_{k \in J} \lambda_k \cdot z^{(k)} \Big) / \Big( \sum_{k \in J} \lambda_k \Big)
\end{align*}
where $J = \{ k \in I: z^{(k)} \in \mathcal{X} \times \mathcal{Y} \}$. Then $\hat{z}$ is feasible, i.e. $\hat{z} \in \mathcal{X} \times \mathcal{Y}$, and we have
\begin{align*}
\overline{f}(\hat{x}) - \underline{f}(\hat{y}) \leq 18\eta \beta R (n+m)^{1/2}.
\end{align*}
\end{lemma}

\begin{proof}
Since $z^{(k)} \in \mathcal{X} \times \mathcal{Y}$ for each $k \in J$, it follows from convexity that $\hat{z} \in \mathcal{X} \times \mathcal{Y}$. 
For any point $z=(x,y) \in \mathcal{X} \times \mathcal{Y}$, we have $\hat{g}_{\beta}(z^{(k)}) = g(z^{(k)})$ and from Lemma~\ref{lem:convexprop}, for any $k \in J$
\begin{align*}
\gamma^{(k)}(z) := (z^{(k)} - z)^{\top} \cdot g(z^{(k)}) \geq f(x^{(k)},y) - f(x,y^{(k)}).
\end{align*}
Let $\hat{\lambda}_k = \lambda_k / (\sum_{k \in J} \lambda_k )$. Taking weighted sum of these inequalities, we have
\begin{align*}
\sum_{k \in J} \hat{\lambda}_k \cdot \gamma^{(k)}(z) \geq \sum_{k \in J} \hat{\lambda}_k \cdot f(x^{(k)},y) - \sum_{k \in J} \hat{\lambda}_k \cdot f(x,y^{(k)}) \geq f(\hat{x},y) - f(x,\hat{y}),
\end{align*}
where the last inequality follows from the fact that $f$ is convex in $x$ and concave in $y$. 
Taking the maximum over $z \in \mathcal{X} \times \mathcal{Y}$ on both sides, we have
\begin{align}
\overline{f}(\hat{x})-\underline{f}(\hat{y}) 
\leq & ~ \max_{z\in\mathcal{X}\times\mathcal{Y}}\sum_{k\in J}\hat{\lambda}_{k}\cdot\gamma^{(k)}(z) \notag \\
\leq & ~ \max_{z\in\mathcal{X}\times\mathcal{Y}}\sum_{k\in I}\hat{\lambda}_{k}\cdot\gamma^{(k)}(z) \notag\\
= & ~ \frac{1}{\sum_{k\in J}\lambda_{k}}\max_{z\in\mathcal{X}\times\mathcal{Y}}\sum_{k\in I}\lambda_{k}\cdot\gamma^{(k)}(z) \notag \\
\leq & ~ \frac{1}{\sum_{k\in J}\lambda_{k}} 9\eta \beta R (n+m)^{1/2} \label{eq:fuflbound}
\end{align}
where the second inequality uses $\gamma^{(k)}(z) \geq 0$ for all $k \notin J$ and $z \in \mathcal{X}\times\mathcal{Y}$  because $\mathcal{X}\times\mathcal{Y} \subset \{z \in \R^{n+m} : (z-(x,y))^\top n(x,y) \leq 0\}$, the third inequality follows from the assumption.

Now we bound $\sum_{k\in J}\lambda_{k}$. For $k\in J$, because $f$ is $L$-Lipschitz and $\|z^{(k)}\|_{2}\leq \sqrt{n+m}R$
$$\gamma^{(k)}(0)=z^{(k)\top}\hat{g}_{\beta}(z^{(k)})\geq-\sqrt{n+m}LR.$$

For $k\in I\backslash J$, we claim that
$$\gamma^{(k)}(0)=z^{(k)\top}n(z^{(k)})\geq\beta R.$$ 

Recall that $z^{(k)}$ is cut off by a constraint of $B_\infty(0, R)$ with normal $n(z^{(k)})$ of length $\|n(z^{(k)})\|=\beta$. Hence $z^{(k)\top}n(z^{(k)})$ is the length of the projection of $z^{(k)}$ onto $n(z^{(k)})$, which is at least $\beta R$.

So we have
\begin{align*}
\eta LR 
\geq & ~ \sum_{k\in J}\lambda_{k}\cdot\gamma^{(k)}(0)+\sum_{k\notin J}\lambda_{k}\cdot\gamma^{(k)}(0)\\
\geq & ~ -\sqrt{n+m}LR \sum_{k\in J} \lambda_{k}+ \Big( 1 - \sum_{k\in J} \lambda_{k} \Big) \cdot \beta R.
\end{align*}
Using $\beta \geq 3\sqrt{n+m}L$, we have $\sum_{k\in J}\lambda_{k}\geq\frac{3\sqrt{n+m}-\eta}{4\sqrt{n+m}}>\frac{1}{2}$.
Now the result follows from \eqref{eq:fuflbound}.
\end{proof}

Thus, given that the maximum of the function $\gamma(z)$ is small, we can find some convex combination of the $z^{(k)}$'s in the cutting plane method that is a good approximation to the saddle point of $f$. 
And it turns out that the maximum of $\gamma$ goes to 0 as $T \rightarrow \infty$, and at the same convergence rate as the cutting plane method for convex minimization.

\begin{lemma}[Upper Bound on $\gamma$]
\label{lem:upperboundgT}
Consider solving the convex-concave game by the cutting plane method. Assume that at certain step $T$ we have
\begin{align*}
\eta = \frac{\vol( P^{(T)})^{1/(n+m)} }{\vol(\mathcal{X} \times \mathcal{Y})^{1/(n+m)}} < \frac{1}{2}.
\end{align*}
Assume that $\beta > L$, we have 
\begin{align*}
\max_{z \in \R^{n+m}} \gamma(z) \leq 8 \eta \beta R (n+m)^{1/2}.
\end{align*}
\end{lemma}

\begin{proof}
Denote $G = \mathcal{X} \times \mathcal{Y}$ for simplicity.
%To prove that $I$ is nonempty, note that if a step $k$ is non-productive, then the feasible set $G \subseteq H^{(k)}$.
%So if $I$ is empty, i.e. every step is non-productive, then $G \subseteq E^{(T)}$ and therefore $\vol(E^{(T)}) \geq \vol(G)$ which is a contradiction to the assumption that $\eta < 1$. 
%Thus, $I$ is nonempty and the function $\gamma$ is well-defined. 
First, we note that $\gamma(z) < \beta(\|x\|_{\infty}-R) \leq 0$ for $z \notin B_\infty(0,R)$ since we include all constraints of $\partial B_\infty(0,R)$
into the definition of $\gamma$.
Hence, it suffices to consider $\gamma$ over $B_\infty(0,R)$.

Let $z^* \in \R^{n + m}$ be a maximizer of $\gamma$ over $B_\infty(0,R)$. Let $\alpha = 2 \eta$ and
\begin{align*}
G^{\alpha} = z^* + \alpha (G - z^* ) := \{ (1-\alpha) z^* + \alpha z ~|~ z \in G \}.
\end{align*}

Notice that $\vol(G^{\alpha}) = \alpha^{m+n} \vol(G) > \vol(P^{(T)})$. 
Therefore, there is some $w \in G^\alpha\setminus P^{(T)}$. This point $w$ is cut off by some $\gamma^{(k)}(z)$: 
\begin{align*}
(w - z^{(k)})^{\top} g^{(k)} > 0,
\end{align*}
for some $z^{(k)} \in B_\infty(0,R)$ (by Line \ref{line:move_z} of the algorithm).
It follows that 
\begin{align*}
\gamma^{(k)}(w) = (z^{(k)} - w)^{\top} g^{(k)} < 0.
\end{align*}
Since $w = (1-\alpha)z^* + \alpha z \in \R^{n + m}$ for some $z \in G$ and the function $\gamma^{(k)} :\R^{n + m} \rightarrow \R$ is affine, we have
\begin{align*}
(1-\alpha) \cdot \gamma^{(k)}(z^*) \leq -\alpha \cdot \gamma^{(k)}(z) + \gamma^{(k)}(w) \leq -\alpha \cdot \gamma^{(k)}(z) = -\alpha \cdot (z^{(k)} - z)^{\top} g^{(k)} \leq 2\alpha \beta R (n+m)^{1/2},
\end{align*}
where we used $\|g^{(k)}\| \leq \max(\beta,L) \leq \beta$ and both $z^{(k)}$ and $z$ are in $B_\infty(0,R)$. 
%There are two cases depending if $w \in G$.
%
%If $w \in G$, it must be cut off by some subgradient of $f$: 
%\begin{align*}
%(w - z^{(k)})^{\top} g(z^{(k)}) > 0,
%\end{align*}
%for some $z^{(k)} \in G$.
%It follows that 
%\begin{align*}
%\gamma^{(k)}(w) = (z^{(k)} - w)^{\top} g(z^{(k)}) < 0.
%\end{align*}
%Since $w = (1-\alpha)z^* + \alpha z \in \R^{n + m}$ for some $z \in G$ and the function $g^{(k)} :\R^{n + m} \rightarrow \R$ is affine, we have
%\begin{align*}
%(1-\alpha) \cdot \gamma^{(k)}(z^*) \leq -\alpha \cdot \gamma^{(k)}(z) = -\alpha \cdot (z^{(k)} - z)^{\top} \cdot g(z^{(k)}) \leq 4\alpha LR,
%\end{align*}
%where we used $f$ is $L$-Lipschitz and both $z^{(k)}$ and $z$ are in $G$.
It follows that 
\begin{align*}
\max_{z \in \R^{n+m}} \gamma(z) = \gamma(z^*)  \leq 4\alpha \beta R (n+m)^{1/2} \leq 8 \eta \beta R (n+m)^{1/2}.
\end{align*}
\end{proof}

Combining Theorem~\ref{thm:vaidya}, Lemma~\ref{lem:ConvCombSameMax},~\ref{lem:convcombSolution} and~\ref{lem:upperboundgT}, we obtain the following theorem which states that solving convex-concave games is exactly as fast as minimizing convex functions.

\begin{theorem}\label{thm:saddle_point}
Given convex sets $\mathcal{X}\subset B(0,R)\subset\R^{n}$
and $\mathcal{Y}\subset B(0,R)\subset\R^{m}$ such that both $\mathcal{X}$ and $\mathcal{Y}$ contain a ball of radius $r$. Let $f(x,y):\mathcal{X}\times\mathcal{Y}\rightarrow\R$ be an $L$-Lipschitz function 
that is convex in $x$ and concave in $y$. For any $0<\epsilon\leq\frac{1}{2}$,
we can find $(\hat{x},\hat{y})$ such that
\begin{align*}
\max_{y\in\mathcal{Y}}f(\hat{x},y)-\min_{x\in\mathcal{X}}f(x,\hat{y})\leq\epsilon L r
\end{align*}
in time 
\begin{align*}
O \left( (n+m)^{3} \log \Big( \frac{n+m}{\epsilon} \frac{R}{r} \Big) +(n+m)\log\Big(\frac{n+m}{\epsilon} \frac{R}{r} \Big) \cdot \mathcal{T} \right)
\end{align*}
with high probability in $n+m$ where $\mathcal{T}$ is the cost of computing subgradient $\nabla f$.
\end{theorem}

\begin{proof}
We run our cutting plane method for $T=(n+m)\log \left(\frac{n+m}{\epsilon} \frac{R}{r} \right)$ iterations. By Theorem~\ref{thm:vaidya}, we obtain $P^{(T)}$ with volume 
\begin{align*}
\vol ( P^{(T)} ) \leq \Big( \frac{\epsilon}{n+m} \frac{r}{R} \Big)^{100(m+n)} \cdot \vol ( P^{(0)} ) \leq \Big( \frac{\epsilon}{n+m} \frac{r}{R} \Big)^{99(m+n)} \cdot \vol(\mathcal{X} \times \mathcal{Y})
\end{align*}
in $O((n+m)^{3}\log(\frac{n+m}{\epsilon} \frac{R}{r}))$ time. 
Notice that 
\begin{align*}
\eta = \frac{\vol( P^{(T)})^{1/(n+m)} }{\vol(\mathcal{X} \times \mathcal{Y})^{1/(n+m)}} \leq \left( \frac{\epsilon}{n+m} \frac{r}{R} \right)^{99} < 1/2 .
\end{align*}
Lemma~\ref{lem:upperboundgT} shows that 
\begin{align*}
\max_{z\in\R^{n+m}}\gamma(z)\leq 8 \eta \beta R (n+m)^{1/2}.
\end{align*}
Lemma~\ref{lem:ConvCombSameMax} shows in $O((n+m)^{\omega+o(1)}\log(\frac{n+m}{\eta})) = O((n+m)^{\omega+o(1)}\log(\frac{n+m}{\epsilon} \frac{R}{r}))$ time, we
can find $\lambda$ such that $$\sum\lambda_{k}\gamma^{(k)}(z)\leq 8 \eta \beta R (n+m)^{1/2}+ \eta \beta R\leq 9 \eta \beta R (n+m)^{1/2}$$
for all $z\in B_{\infty}(0,R)$. Using this $\lambda$, Lemma~\ref{lem:convcombSolution}
shows that, by taking a convex combination w.r.t. $\lambda$, we can find $(\hat{x},\hat{y})$ for which
\begin{align*}
\overline{f}(\hat{x}) - \underline{f}(\hat{y}) \leq 18\eta \beta R (n+m)^{1/2}.
\end{align*}
Picking $\beta = 3\sqrt{n+m}L$, we get
\begin{align*}
\overline{f}(\hat{x}) - \underline{f}(\hat{y}) \leq 54\eta L R (n+m) \leq 54 \epsilon L r.
\end{align*}
Now our result follows by replacing $\epsilon$ with $\epsilon / 54$, which doesn't change the asymptotic runtime.
%$54\epsilon (n+m)$ with $\epsilon'$ (which doesn't change the asymptotic runtime).
\end{proof}

Next we give a slightly refined runtime in the case where the domain is a ball using a result of Nemirovski~\cite{n95}. 
%in the case of unit-ball domains.

\begin{corollary}
Given an $L$-Lipschitz function $f(x,y):\mathcal{X}\times\mathcal{Y}\rightarrow\R$
that is convex in $x$ and concave in $y$ with $\mathcal{X}=B(0,R)\subset\R^{n}$
and $\mathcal{Y}=B(0,R)\subset\R^{m}$. For any $0<\epsilon\leq\frac{1}{2}$,
we can find $(\hat{x},\hat{y})$ such that
\begin{align*}
\max_{y\in\mathcal{Y}}f(\hat{x},y)-\min_{x\in\mathcal{X}}f(x,\hat{y})\leq\epsilon LR
\end{align*}
in time 
\begin{align*}
O((n+m)^{3}\log(1/\epsilon)+(n+m)\log(1/\epsilon) \cdot \mathcal{T})
\end{align*}
with high probability in $n+m$ where $\mathcal{T}$ is the cost of
computing $\nabla f$.
\end{corollary}

\begin{proof}
By \cite[Theorem 5.5.4]{n95}, we can minimize such a function in time $O(\frac{\mathcal{T}+n+m}{\epsilon})$.
On the other hand, Theorem~\ref{thm:saddle_point}  shows we can minimize it in time 
\begin{align*}
O((n+m)^{3}\log( ( n+m ) / \epsilon )+(n+m)\log( ( n+m ) / \epsilon ) \cdot \mathcal{T}).
\end{align*}
By using the first result when $\epsilon\geq\frac{1}{n+m}$ and using
the second result when $\epsilon\leq\frac{1}{n+m}$, we have the promised
runtime.
\end{proof}

\newpage
\section{Applications of Cutting Plane Method}\label{sec:application}

\begin{comment}
\begin{definition}[Strongly Polynomial Time Algorithm]
Assume that a problem is given by an input of $N$ rational numbers given in binary description. An algorithm for such a problem is
{\em strongly polynomial}, if it only uses elementary arithmetic operations (addition, comparison, multiplication, and
division), and the total number of such operations is bounded by $\poly(N)$. Further, the algorithm is required to run in polynomial space in the input size.
\end{definition}
\end{comment}
\subsection{Linear Arrow-Debreu markets}

\begin{table}[htp!]
    \centering
    \begin{tabular}{ | l | l | l | l | l | }
        \hline
        {\bf Reference} & {\bf Year} & {\bf Number of Operations} & {\bf Time per Operation} & {\bf Poly Type } \\ \hline
        \cite{e75}    & 1975    & Finite                 & Polynomial      & Not poly \\ \hline
        \cite{j07} & 2007  & Polynomial & Polynomial      & Weakly poly \\ \hline
        \cite{y08}         & 2008    & $O(n^6 \log(nU))$       & $M(n \log (nU))$ & Weakly poly \\ \hline
        \cite{dpsv08}  & 2008    & Polynomial & Polynomial    & Weakly poly \\ \hline
        \cite{dm15} & 2015 & $O(n^9 \log(nU))$ & $M(n \log(nU))$ &  Weakly poly \\ \hline
        \cite{dgm16}   & 2016    & $O(n^6 \log^2 (nU))$   &  $M(n \log (nU))$ & Weakly poly \\ \hline
        \cite{gv19}               & 2019    & $O(m n^{9} \log^2 n)$  & $M(n \log (nU))$  & Strongly poly \\ \hline
        Our result & 2019 & $O(m n^2 \log(nU))$ & $M(n \log (nU))$ & Weakly poly \\ \hline
    \end{tabular}
    \caption{(More detailed version of Table~\ref{tab:arrow_debreu_market_intro}) Linear Arrow-Debreu Markets. Let $n$ denote the number of agents and $m$ the number of edges. $M(l)$ denotes the time to perform a basic arithmetic operations on $l$-bit numbers.}
    \label{tab:arrow_debreu_market}
\end{table}

\paragraph{Linear Exchange (Arrow-Debreu) Markets} The input consists of $n$ agents where each agent $i \in [n]$ has a unit of divisible good $g_i$ and utility $u_{i,j} \geq 0$ for a unit of good $g_j$. 
The output is a set of prices $p: [n] \rightarrow \mathbb{R}_+$ for the goods and allocations $x: [n] \times [n] \rightarrow \mathbb{R}_+$ of the goods to agents.
Notice that $x_{i,j}$ can be understood as the amount of good $g_j$ perchased by agent $i$, and $x_{i,i}$ is the amount of good $g_i$ that agent $i$ keeps to himself. 

\paragraph{Market equilibrium}
By a {\em market equilibrium}, we mean a set of prices $p: [n] \rightarrow \mathbb{R}_+$ and allocations $x: [n] \times [n] \rightarrow \mathbb{R}_+$ satisfying the following conditions:

\begin{itemize}
	\item $\sum_{i \in [n]} x_{i,j} = 1$ for every agent $j \in [n]$, i.e. every good is fully sold.
	\item $ p_i = \sum_{j \in [n]} x_{i,j} p_j $ for every agent $i \in [n]$, i.e. the money spent by agent $i$ equals to his income $p_i$.
	\item $p_i > 0$ for every $i \in [n]$, i.e. prices are positive.
	\item $\forall i \in [n]$, if $x_{i,j} > 0$ then $u_{i,j}/p_j = \max_{j' \in [n]} u_{i,j'}/p_{j'}$, i.e. agent $i$ only buys good that attains the best bang-per-bucks.
\end{itemize}

\begin{assumption}
To ensure the existence of an equilibrium, we assume the following.
Consider the directed graph $G=(V,E)$ where $V = [n]$ and directed edges $E=\{ (i,j) : u_{i,j} > 0\}$.
\begin{itemize}
	\item For each agent $i$, there exist $j,j' \in [n]$ such that $u_{i,j} > 0$ and $u_{j',i} > 0$, i.e. each $i \in G$ has at least one incoming edge and one outgoing edge.
	\item For every strongly connected component $S \subseteq G$, if $|S| = 1$ then there is a loop incident to the node in $S$.
\end{itemize}
\end{assumption}

\begin{theorem}
A market equilibrium always exists under the above assumptions.
\end{theorem}

\paragraph{Previous work}
The celebrated result of Arrow and Debreu~\cite{ad54} shows the existence of a market equilibrium for a broad class of utility functions.  
From the computational aspects, computing the equilibrium of a market equilibrium in the case of linear utility functions admits a long line of research (see Table~\ref{tab:arrow_debreu_market}), leading to both weakly and strongly polynomial runtimes. 

\paragraph{A convex formulation}
The problem of computing a market equilibrium in linear exchange markets can be formulated as the following convex program due to~\cite{dgv16}, with variables $p_i$ representing the prices, $\beta_i$ the inverse best bang-per-bucks, and $y_{ij}$ the money paid by agent $i$ to agent $j$.
Given $u \in \R_+^{n \times n}$ as an input matrix, the goal is to solve the following minimization problem,
\begin{align}
\min_{ p \in \R^n, \beta \in \R^n, y \in \R^{n \times n} } & \quad \sum_{i \in [n]}p_i \log (p_i/\beta_i) - \sum_{ (i,j) \in E} y_{i,j} \log u_{i,j} \nonumber\\
\text{subject~to~} & \sum_{i: (i,j) \in E} y_{i,j} = p_j \quad \forall j \in [n] \label{eqn:GoodFullySold}\\
& \sum_{j: ij \in E} y_{ij} = p_i \quad \forall i \in [n] \label{eqn:SpentAllMoney} \\
& u_{ij} \beta_i \leq p_j \quad \forall (i,j) \in E \nonumber \\
& p_i \geq 1 \quad \forall i \in [n] \nonumber \\
& y,\beta \geq 0 \nonumber
\end{align}

\begin{theorem}[Theorem 1 in~\cite{dgv16}]
Consider an instance of the linear exchange market given by the graph $([n],E)$ and the utilities $u: [n] \times [n] \rightarrow \mathbb{R}_+$. The above convex program is feasible if and only if the above assumptions hold, and in this case the optimum value is $0$ and the prices $p_i$ in an optimal solution give a market equilibrium with allocations $x_{ij} = y_{ij} / p_j$.
\end{theorem}

\paragraph{Convex-concave game formulation}
The number of variables in the above convex program can be as large as $n^2$ and thus a naive application of the cutting plane method would lead to lead to a large runtime.
To solve the convex program more efficiently, we transform it into a convex-concave game with a reduced number of variables. 
Let $\lambda, \eta \in \R^n$ be the Lagrange multiplier of constraint~\ref{eqn:GoodFullySold} and~\ref{eqn:SpentAllMoney} respectively. We have that the above convex program is equivalent to:

\begin{align*}
\min_{p \geq 1 \in \R^n, \beta \in \R_+^n, u_{ij}\beta_i \leq p_j} \min_{y \in \R_{+}^{n \times n}} \max_{\lambda, \eta \in \R^n} 
& ~ \sum_{i=1}^n p_i \log ( p_i / \beta_i ) - \sum_{i=1}^n \sum_{j=1}^n y_{i,j} \log u_{i,j} \\
& ~ + \sum_{j=1}^n \lambda_j \Big( \sum_{i=1}^n y_{i,j} - p_j \Big) + \sum_{i=1}^n \eta_i \Big( \sum_{j=1}^n y_{i,j} - p_i \Big).
\end{align*}
which is equivalent to
\begin{align*}
\min_{p \geq 1, \beta \geq 0, u_{ij}\beta_i \leq p_j} \max_{\lambda, \eta \in \R^n} \min_{y \geq 0} & ~ \sum_{i=1}^n p_i \log (p_i / \beta_i) + \sum_{i=1}^n \sum_{j=1}^n y_{i,j} ( \lambda_j + \eta_i - \log u_{i,j}) \\
& ~ - \sum_{j=1}^n \lambda_j p_j - \sum_{i=1}^n \eta_i p_i \\
%\text{subject~to~} & p \geq 1, \beta \geq 0, u_{ij}\beta_i \leq p_j.
\end{align*}
which is further equivalent to the following convex-concave game
\begin{align*}
\min_{p \geq 1, \beta \geq 0, u_{ij}\beta_i \leq p_j} \max_{\lambda_j + \eta_i \geq \log u_{i,j}} & ~ \sum_{i=1}^n p_i \log (p_i / \beta_i) - \sum_{j=1}^n \lambda_j p_j - \sum_{i=1}^n \eta_i p_i \\
%\text{subject~to~} & ~ p \geq 1, \beta \geq 0, u_{ij}\beta_i \leq p_j. \\
 %& ~ \lambda_j + \eta_i \geq \log u_{i,j}
\end{align*}
Notice that the convex-concave game formulation above has $O(n)$ variables. 
The following upper bound on equilibrium prices is due to~\cite{dgv16}.
\begin{lemma}[Lemma 13 and Remark 14 in~\cite{dgv16}] \label{lem:price_upper_bound}
Assume all utilities are integers $\leq U$ and we let $\Delta = (nU)^n$.
Then there exists equilibrium prices $p$ that are quotient of two integers $\leq \Delta$, along with allocations $x$ that are quotients of two integers $\leq \Delta^2$.
\end{lemma}

Now using Theorem~\ref{thm:saddle_point}, we need $O(n^2 \log (nU))$ iterations and the first order oracle requires $O(m)$ operations. 
Therefore, the total number of operations is $O(m n^2 \log (nU))$. We formally state our result as follows. 

\begin{theorem}[Formal version of Theorem~\ref{thm:arrow_debreu_market_intro}]\label{thm:arrow_debreu_market}
There exists a weakly polynomial algorithm that computes a market equilibrium in linear exchange markets in time $O(mn^2 \log (nU))$.
\end{theorem}

\begin{comment}
\begin{corollary}
The following hold for linear exchange market.
\begin{itemize}
	\item For every agent the utility is the same at every equilibrium.
	\item The vectors $(y,p)$ a equilibrium form a convex set. 
\end{itemize}
\end{corollary}

\begin{theorem}[Theorem 3.1 in~\cite{gv19}]
There exists a strongly polynomial algorithm that computes a market equilibrium in linear exchange markets in time $O(n^{10} \log^2 n)$.
\end{theorem}
\end{comment}

\subsection{Fisher markets with spending constraint utilities}

\begin{table}[htp!]
    \centering
    \begin{tabular}{ | l | l | l | l | l | }
        \hline
        {\bf Reference} & {\bf Year} & {\bf Running Time} & {\bf Time per Operation} &  {\bf Poly Type } \\ \hline
        \cite{v10} & 2010 &  $O(n^3 (n+m)^2 \log (U)) \cdot \T_{\text{max-flow}} $ & $M(n \log (nU))$  & Weakly poly \\ \hline
        \cite{v16} & 2016 & $O(mn^3 + m^2 (m+n \log(n))\log(m))$ & $M(n \log (nU))$ & Strongly poly  \\ \hline
        \cite{wan16} & 2016 & $O(m^3 n + m^2 \log(n) (n \log(n) + m))$ & $M(n \log (nU))$ & Strongly poly \\ \hline
        Our result & 2019 & $O(m n^2 \log (nU))$ & $M(n \log (nU))$ & Weakly poly \\ \hline
    \end{tabular}
    \caption{(More detailed version of Table~\ref{tab:fisher_markets_intro}) Fisher Markets with Spending Constraint Utilities. Let $n$ denote the total number of buyers and sellers, and $m$ the total number of segments. $\T_{\text{max-flow}}$ denotes the number of operations needed for a max-flow computation. $M(l)$ denotes the time to perform a basic arithmetic operations on $l$-bit numbers.} 
    \label{tab:fisher_markets}
\end{table}

\paragraph{Fisher markets with spending constraint utilities}
In the Fisher market model with spending constraint utilities, there are $n_B$ buyers and $n_G$ perfectly divisible goods. There is a unit\footnote{This is without loss of generality since the goods are divisible.} supply of each good.
Each buyer $i \in [n_B]$ has a budget $B_i$. 
The utility of a buyer depends on the prices, as follows. 
The utility is additive across goods, that is, the total utility for a bundle of goods is the sum of utilities for each good separately. 
For a given good, the utility function is divided into segments; each segment $l$ has a constant rate of utility $u_{i,j,l}$, and a budget $B_{i,j,l}$. The utility of the buyer for $x_{i,j,l}$ amount of good $j \in [n_G]$ under segment $l \in [L]$ is $u_{i,j,l} x_{i,j,l}$ , but is subject to the constraint that $p_j x_{i,j,l} \leq B_{i,j,l}$. 
Denote $n := n_B + n_G$ the total number of buyers and goods, and $m$ the total number of segments among all pair $(i,j)$.

\paragraph{Market equilibrium} 
An allocation $x \in \R^{n_B \times n_G \times L}_+$ and price vector $p \in \R^{n_G}_+$
are an equilibrium if two conditions are satisfied. 
The first condition, {\em buyer optimality}, is that given $p$, each player maximizes his utility subject to his budget constraint.
In other words, for all $i \in [n_B]$, the allocation $x$ optimizes the following program (where $p$ is a constant):
\begin{align*}
\max_{x \in \R^{n_B \times n_G \times L}_+} \quad&\sum_{j=1}^{n_G} \sum_{l=1}^L u_{i,j,l} x_{i,j,l} \\
\text{subject to} \quad &\sum_{j=1}^{n_G} \sum_{l=1}^L p_j x_{i,j,l} \leq B_i, \\
& 0 \leq x_{i,j,l} \text{ and } p_j x_{i,j,l} \leq B_{i,j,l} & \forall j \in [n_G], l \in [L].
\end{align*}
The second equilibrium condition, {\em market clearance}, is that $\sum_{i,l} x_{i,j,l} = 1$ for all $j \in [n_G]$.

\paragraph{Previous work} The Fisher market with spending constraint utilities problem was introduced in~\cite{v10} where they gave a weakly poly time algorithm that takes $O(n^3 (n + m)^2 \log U)$ max-flow computations, where $U = \max_{i \in [n_B], j \in [n_G], l \in [L]} u_{i,j,l}$.
A convex formulation of the problem was first given in the unpublished manuscript~\cite{bdx10}.
A strongly poly time algorithm (which is also the best weakly polynomial running time) was given in~\cite{v16}which achieves $O(mn^3 + m^2(m + n \log n) \log m)$ number of operations.

\paragraph{A convex formulation}
The problem of computing a market equilibrium in Fisher markets with spending constraint utilities problem can be captured by the following convex program due to~\cite{bdx10}.
\begin{align}
\max_{b \in \R^{n_B \times n_G \times L}_+, p \in \R^{n_G}_+} \quad & \sum_{i=1}^{n_B} \sum_{j=1}^{n_G} \sum_{l=1}^L b_{i,j,l} \log u_{i,j,l} - \sum_{j=1}^{n_G} p_j \log p_j \label{eqn:fisher_market_program} \\
\text{subject to} \quad &\sum_{i=1}^{n_B} \sum_{l=1}^L b_{i,j,l} = p_j, & \forall j \in [n_G], \nonumber \\
			&\sum_{j=1}^{n_G} \sum_{l=1}^L b_{i,j,l} = B_i, & \forall i \in [n_B], \nonumber \\
			& 0 \leq b_{i,j,l} \leq B_{i,j,l}, & \forall i \in [n_B],j \in [n_G], l \in [L] \nonumber.
\end{align}

\begin{theorem}[\cite{bdx10}]
An optimum solution to the above convex program corresponds to an equilibrium for the Fisher market with spending constraint utilities with allocation given by $x_{i,j,l} = b_{i,j,l} / p_j$.
\end{theorem}

\paragraph{Convex-concave game formulation} 
We transform the convex program above to a convex-concave saddle point which allows us to apply the cutting plane method.
Let $\eta_j, \lambda_i$ be the Lagrange multipliers for the two equality constraints and $\mu_{i,j,l} \geq 0$ be the Lagrange multipliers for the inequalities $b_{i,j,l} \leq B_{i,j,l}$ in the above convex formulation.
We have the above convex program is equivalent to %{\color{red} Zhao : the following equation is so readable... it is written in a bit ugly way}
\begin{align*}
& \max_{p \geq 0} \max_{b\geq 0} \min_{\mu\geq 0, \eta, \lambda} - \sum_{j=1}^{n_G} p_j \log p_j + \sum_{i,j,l} b_{i,j,l} \log u_{i,j,l} \\
& ~ + \sum_{j=1}^{n_G} \eta_j(p_j - \sum_{i,l} b_{i,j,l}) + \sum_{i=1}^{n_B} \lambda_i (B_i - \sum_{j,l} b_{i,j,l}) + \sum_{i,j,l} \mu_{i,j,l} (B_{i,j,l} - b_{i,j,l}) \\
= & ~ \max_{p \geq 0} \min_{\mu\geq 0, \eta, \lambda} \max_{b\geq 0} - \sum_{j=1}^{n_G} p_j \log p_j  + \sum_{i,j,l} b_{i,j,l} (\log u_{i,j,l} - \eta_j - \lambda_i - \mu_{i,j,l} ) \\
& ~ + \sum_{j=1}^{n_G} \eta_j p_j + \sum_{i=1}^{n_B} \lambda_i B_i + \sum_{i,j,l} \mu_{i,j,l} B_{i,j,l} \\
= & \max_{p \geq 0} \min_{\mu\geq 0, \eta, \lambda, \eta_j + \lambda_i + \mu_{i,j,l} \geq \log u_{i,j,l}} - \sum_j p_j \log p_j \\
& ~ + \sum_{j=1}^{n_G} \eta_j p_j + \sum_i \lambda_i B_i + \sum_{i,j,l} \mu_{i,j,l} B_{i,j,l} \\
= & \max_{p \geq 0} \min_{\eta, \lambda} - \sum_{j=1}^{n_G} p_j \log p_j \\
& ~ + \sum_{j=1}^{n_G} \eta_j p_j + \sum_{i=1}^{n_B} \lambda_i B_i + \sum_{i,j,l} B_{i,j,l} \max\{0, \log u_{i,j,l} - \eta_j - \lambda_i\},
\end{align*}
where the last step is because $B_{i,j,l} \geq 0$. 
Now applying Theorem~\ref{thm:saddle_point}, we need $O(n^2 \log (nU))$ iterations and the first order oracle takes $O(m)$ operations. 
So the total number of operations is $O(m n^2 \log (nU))$.
We formally state our result as follows:

\begin{theorem}[Formal version of Theorem~\ref{thm:fisher_market_intro}]\label{thm:fisher_market}
There exists a weakly polynomial algorithm that computes a market equilibrium in Fisher markets with spending constraint utilities in time $O(mn^2 \log (nU))$.
\end{theorem}

\subsection{Walrasian equilibrium for general buyer valuations and fixed supply}\label{sec:walras}

\begin{table}[htp!]
    \centering
    \begin{tabular}{ | l | l | l | l | l | }
        \hline
        {\bf Reference} & {\bf Year} & {\bf Running Time} \\ \hline
        \cite{kc82} & 1982 & Not polynomial time \\ \hline
        \cite{p99} & 1999 & pseudo-polynomial time \\ \hline
        \cite{pu02} & 2002 & pseudo-polynomial time \\ \hline
        \cite{am02} & 2002 & pseudo-polynomial time \\ \hline
        \cite{dsv07} & 2007 & pseudo-polynomial time \\ \hline
        \cite{lw17} & 2017 &  $O(n^2 \T_{\mathrm{AD}} \log(SMn)+n^6 \log^{O(1)}(SMn))$   \\ \hline
        Our result & 2019 & $O(n^2 \T_{\mathrm{AD}} \log(SMn)+n^4 \log(SMn))$ \\ \hline
    \end{tabular}
    \caption{(More detailed version of Table~\ref{tab:walrasian_equilibrium_intro}) Walrasian equilibrium for general buyer valuations and fixed supply. } 
    \label{tab:walrasian_equilibrium} %{\color{red} Haotian: maybe it makes sense to include reference with non-polynomial running time here so that the table looks longer?}
\end{table}

Recently Paes Leme and Wong~\cite{lw17} studied the problem of computing Walrasian equilibrium in a market with \emph{arbitrary} buyers' valuation functions $v_i$ and fixed supply $s$ of indivisible goods. They showed that such equilibria are characterized by the minima of the following convex program, which can be solved in polynomial time via a suitable extension of the cutting plane methods.

Let $[m]$ and $[n]$ be the set of buyers and goods respectively. The setting of the problem is as follows.

\paragraph{Buyers' valuation and utility} Each buyer $i\in [m]$ may have an arbitrary valuation $v_i:\mathbb{N}^{n} \rightarrow \R$ over goods, i.e. for $x\in\mathbb{N}^{n}$, $v_i(x)$ is the valuation of $i$ if she gets $x_j$ units of good $j \in [n]$. Given prices $p\in \R^{[n]}$ on goods, the utility of $i$ getting $x\in\mathbb{N}^{n}$ is given by $v_i(x) - p \cdot x$. To maximize her utility, given $p$, $i$ would buy a bundle 
\begin{align*}
x^{(i)}\in \arg\max_{x\in\mathbb{N}^{n}:0\leq x\leq s} v_i(x) - p \cdot x.
\end{align*}

\paragraph{Aggregate demand oracle} To study this problem computationally, one must specify the information about the market available to the algorithm. Paes Leme and Wong showed that this problem can be solved in polynomial time assuming the aggregate demand oracle.

Given prices $p$, the aggregate demand is given by $\sum_{i\in [m]} x^{(i)}$, where 
\begin{align*}
x^{(i)} \in \arg\max_{ x \in \mathbb{N}^n : 0 \leq x \leq s } v_i(x)-p\cdot x
\end{align*}
is a utility-maximizing bundle for buyer $i$.

\paragraph{Walrasian equilibrium} In this economy the supply of goods is fixed at $s \in \mathbb{N}^{n}$. Prices $p$ are said to be at equilibrium if there exist some aggregate demand that matches the supply, i.e. 
\begin{align*}
s = \sum_{i\in [m]} x^{(i)}.
\end{align*}

\paragraph{Convex programming} One approach to computing Walrasian equilibrium is by formalating the problem as the following convex program. Previous works applied subgradient descent methods to this program to obtain pseudopolynomial time algorithms~\cite{p99,pu02,am02}. In contrast, Paes Leme and Wong showed that this program can in fact be solved in weakly polynomial time.

\begin{align*}
 \min_{u,p} & \sum_{i\in [m]} u_i + p\cdot s \\
u_i & \geq v_i(x)-p\cdot x \quad \forall i \in [m], \forall x
\end{align*}

Here $p$ denotes the prices and $u_i$ can be thought of as the utility of buyer $i$, $\forall i \in [m]$. This convex program can clearly be solved with access to the \emph{demand oracle} of each buyer $i$, i.e. given $p$, return $\arg\max_x v_i(x)-p\cdot x$, which serves as a separation oracle.

The main result of Paes Leme and Wong is that this can be done via the weaker \emph{aggregate demand oracle}, i.e. given $p$, return $\sum_i x^{(i)}$, where 
\begin{align*}
x^{(i)} \in \arg\max_x v_i(x) - p \cdot x.
\end{align*}
Their algorithm is based on a suitable extension of standard cutting plane methods:

\begin{theorem}[\cite{lw17}]
There is an algorithm that runs in time 
\begin{align*}
O(n^2 \T_{\mathrm{AD}} \log(SMn) + n^6 \log^{O(1)}(SMn) )
\end{align*}
to compute Walrasian prices whenever it exists using only access to an aggregate demand oracle. Here $T_{AD}$ denotes the runtime of the aggregate demand oracle, $n$ denotes the number of goods, $M = \max_{i\in [m],0\leq x\leq s} |v_i(x)|$ and $S=\max_{j\in [n]} s_j$ ($s$ and $v_i$ are integer-valued).
\end{theorem}

As in previous applications,  $n^2$ iterations of cutting plane are required as $\epsilon$ has to be taken to be exponentially small which results in an overall runtime overhead of $\tilde{O}(n^6)$. By leveraging our faster cutting plane method, we obtain the following improved result:

\begin{theorem}[Formal version of Theorem~\ref{thm:walrasian_equilibrium_intro}]\label{thm:walrasian_equilibrium}
There is an algorithm that runs in time
\begin{align*}
O( n^2 \T_{\mathrm{AD}} \log(S M n) + n^4 \log(S M n) )
\end{align*}
to compute Walrasian prices whenever it exists using only access to an aggregate demand oracle.
\end{theorem}
\begin{proof}
Same as Paes Leme-Wong except that we invoke our cutting plane method for convex minimization in Theorem~\ref{thm:convex} instead of LSW's.
\end{proof}

\end{document}